\documentclass[12pt]{article}

\usepackage[T1]{fontenc}
\usepackage[utf8]{inputenc}
\usepackage{amsmath,amsfonts,amssymb,amsthm,mathtools}
\usepackage{dsfont,bm}
\usepackage[cal=cm, bb=ams]{mathalpha}

\usepackage[semicolon,round,longnamesfirst]{natbib}

\usepackage[kerning=true,protrusion=true,expansion]{microtype}
\usepackage[a4paper,left=1in,right=1in,bottom=1in,top=1in]{geometry}
\usepackage{setspace}
\usepackage{comment}
\usepackage{dirtytalk}

\usepackage{sectsty}
\usepackage{titlesec}
\usepackage[title,titletoc,toc,page]{appendix}

\usepackage[usenames,dvipsnames]{xcolor}
\definecolor{mBlue}{HTML}{002FA7}
\definecolor{mGreen}{HTML}{009B55}
\definecolor{mOrange}{HTML}{FF4F00}
\definecolor{mBlack}{HTML}{004242}

\usepackage{graphics,caption,subcaption}

\usepackage{tikz,pgfplots}
\pgfplotsset{compat=newest}
\usepgfplotslibrary{fillbetween}
\pgfplotsset{soldot/.style={color=blue,only marks,mark=*}}
\usetikzlibrary{positioning,decorations.markings,arrows,arrows.meta}
\usetikzlibrary{external}
\usepgfplotslibrary{patchplots,ternary}

\usepackage{enumerate}
\usepackage[shortlabels]{enumitem}
\usepackage{accents}
\usepackage{url}
\usepackage{breakcites}
\usepackage{footmisc}
\usepackage[backref=page,breaklinks]{hyperref}
\usepackage[nameinlink,noabbrev,capitalise]{cleveref}
\usepackage{theoremref}

\usepackage{algorithm}
\usepackage{algpseudocodex}

\titleformat{\section}[hang]{\normalfont\scshape\large}{\thesection.}{1em}{\centering}
\titleformat{\subsection}[hang]{\normalfont\itshape\bf}{\thesubsection.}{1em}{\centering}
\titleformat{\subsubsection}[hang]{\normalfont\itshape}{\thesubsubsection.}{1em}{\centering}
\let\originalparagraph\paragraph
\renewcommand{\paragraph}[2][.]{\originalparagraph{#2#1}} 

\hypersetup{
	backref=true, 
	pagebackref=true, 
	hyperindex=true, 
	colorlinks=true, 
	breaklinks=true, 
	urlcolor=blue, 
	linkcolor=mBlue,  
	citecolor=NavyBlue 
}

\newtheoremstyle{mystyle}
  {}
  {}
  {\itshape}
  {}
  {\bfseries}
  {.}
  { }
  {\thmname{#1}\thmnumber{ #2}\thmnote{ (#3)}}

\theoremstyle{mystyle}
\newtheorem{lemma}{Lemma}
\newtheorem{proposition}{Proposition}
\newtheorem{theorem}{Theorem}
\newtheorem{corollary}{Corollary}
\newtheorem{assumption}{Assumption}

\makeatletter
\newcounter{parenttheorem}

\makeatother

\newtheoremstyle{mydefstyle}
  {}
  {}
  {}
  {}
  {\bfseries}
  {.}
  { }
  {\thmname{#1}\thmnumber{ #2}\thmnote{ (#3)}}

\theoremstyle{mydefstyle}
\newtheorem{definition}{Definition}
\newtheorem{example}{Example}

\theoremstyle{remark}
\newtheorem{remark}{Remark}

\newcommand{\ubar}[1]{\underaccent{\bar}{#1}}

\newcommand{\longbar}[1]{\mkern 1.5mu\overline{\mkern-1.5mu#1\mkern-1.5mu}\mkern 1.5mu}

\makeatletter
\def\mathcolor#1#{\@mathcolor{#1}}
\def\@mathcolor#1#2#3{%
  \protect\leavevmode
  \begingroup
    \color#1{#2}#3%
  \endgroup
}
\makeatother

\DeclareMathOperator{\supp}{supp}

\DeclareMathOperator{\Int}{int}

\DeclareMathOperator{\ex}{ex}
\DeclareMathOperator{\cav}{cav}
\DeclareMathOperator{\vex}{vex}
\DeclareMathOperator{\argmax}{argmax}

\DeclareMathOperator{\Id}{\mathrm{id}}
\DeclareMathOperator{\Ind}{\mathbf{1}}

\newcommand{\diff}{\mathrm{d}}

\title{\textbf{The Economics of Convex Function Intervals\thanks{We are indebted to Andreas Kleiner for numerous insightful discussions and comments. We are also particularly grateful to Sarah Auster and Eduardo Perez-Richet for their guidance and support. We would also like to thank Thomas Brzustowski, Gregorio Curello, Francesc Dilmé, Piotr Dworczak, Pia Ennuschat, Atulya Jain, Stephan Lauermann, Qianjun Lyu, Franz Ostrizek, Philipp Strack and Yimeng Zhang for helpful feedback at various stages of the project. We also thank audiences at Bonn, Mannheim, the CRC TR 224 EPoS Young Researchers Workshop, the 2025 IMD Days in Warsaw, the 2025 SAET conference in Ischia, and the 8\textsuperscript{th} Lindau Nobel
Meeting in Economic Sciences. Support by the German Research Foundation (DFG) through CRC TR 224 EPoS (Projects B01 (Augias) and B03 (Uhe)) is gratefully acknowledged. All remaining errors are ours.}}}
\author{Victor Augias\thanks{University of Bonn, Department of Economics\textemdash e-mail: \href{mailto:vaugias@uni-bonn.de}{\texttt{vaugias@uni-bonn.de}}} \and Lina Uhe\thanks{University of Bonn, Department of Economics, Bonn Graduate School of Economics\textemdash e-mail: \href{mailto:lina.uhe@uni-bonn.de}{\texttt{lina.uhe@uni-bonn.de}}}}
\date{\today}

\begin{document}

\maketitle

\begin{abstract}
    We introduce \emph{convex function intervals} (CFIs): families of convex functions satisfying given level and slope constraints. CFIs naturally arise as constraint sets in economic design, including problems with \emph{type‑dependent} participation constraints and \emph{two‑sided (weak) majorization} constraints. Our main results include: (i) a geometric characterization of the \emph{extreme points} of CFIs; (ii) sufficient optimality conditions for \emph{linear programs} over CFIs; and (iii) methods for nested optimization on their \emph{lower level boundary} that can be applied, e.g., to the optimal design of outside options. We apply these results to four settings:  screening and delegation problems with type-dependent outside options, contest design with limited disposal, and mean-based persuasion with informativeness constraints. We draw several novel economic implications using our tools. For instance, we show that better outside options lead to larger delegation sets, and that posted price mechanisms can be suboptimal in the canonical monopolistic screening problem with nontrivial, type-dependent participation constraints.
\end{abstract}

\clearpage

\section{Introduction}

This paper studies \emph{convex function intervals} (CFIs)---sets of continuous convex functions subject to prescribed level and slope boundary constraints---and their applications to mechanism and information design. We show that various problems in economic design admit CFIs as their feasible sets. Our goal is to offer a systematic analysis of the mathematical structure of CFIs, and to provide new techniques for solving \emph{linear} optimization problems defined on them. We apply our framework to derive feasible and optimal mechanisms in adverse selection problems subject to type-dependent participation constraints, large contest design problems with allocative constraints, and the Bayesian persuasion problem with informativeness restrictions. Our techniques are potentially applicable to a varied array of other economic design problems.

Specifically, we define CFIs as sets of convex functions defined on a compact interval of the real line, which slopes lie within a prescribed interval, and whose values are sandwiched between two convex boundary functions. We show that CFIs are \emph{compact} and \emph{convex} sets. Therefore, CFIs are fully described by their \emph{extreme points}. Our first main result (\cref{thm:ext_pt}) provides a geometric characterization of these extreme points. It shows, in particular, that each extreme point is characterized by a countable collection of intervals. Outside these intervals an extreme point must coincide with one of the two boundary functions. On each interval the extreme point is an affine segment. Each of these affine segments must, in addition, satisfy at least one of the following conditions: (i) it extends linearly from a tangency point with the upper level boundary; (ii) it is a chord connecting two points of its own graph that are not confounded with the upper level boundary; or (iii) it has either the minimal or the maximal permitted subgradient. Moreover, when condition (ii) or (iii) is satisfied, each endpoint of the affine segment must either coincide with the lower level boundary of the CFI, or be adjacent to a segment satisfying condition (i). In our applications, we highlight how these affine segments relate to bunching and ironing in mechanism design, as well as to the mean-preserving spreads and contractions operations.

In applications, the representation of the feasible set as a CFI is usually not explicit but stems from the underlying economic and incentive constraints. In quasi-linear adverse selection environments, for instance, physical and incentive constraints translate into CFIs of agents' \emph{indirect utility functions}, following the characterization of \citet{Rochet1987}. Similarly, sets of monotone functions ordered by the \emph{majorization} relation---shown by \citet*{Kleiner2021} to play a central role in capturing such constraints---can also be expressed as CFIs. Our \cref{thm:ext_pt} therefore recovers their characterization of the extreme points majorization sets as a special case, and extends it to the case of \emph{two-sided} (weak) majorization constraints, that is, to the sets of monotone functions that are both (weak) mean-preserving spreads of one fixed monotone function and (weak) mean-preserving contractions of another.

Beyond characterizing the entire feasible set, extreme points also provide insights about the structure of optimal solutions. Indeed, maximizers of linear programs on CFIs can always be found among the extreme points. In economic applications, however, one often needs to check whether a particular class of extreme points is optimal for an economic design problem. Our second main result, \cref{thm:charac_opt}, is a \emph{verification theorem}: it gives sufficient conditions on the objective that guarantee the optimality of a prescribed extreme point. Intuitively, these conditions ensure that the extreme point is \emph{locally} optimal on each cell of a partition of its domain. This partition is built from the collection of intervals that characterize the extreme point (\cref{thm:ext_pt}). To be more precise, in the economic design problems we consider, the designer's objective functional can be written as the expectation of some element of a CFI with respect to a \emph{signed measure}. This measure captures the designer's \emph{marginal incentives} and is typically derived from the model's primitives. For instance, in screening and delegation problems, it corresponds to a transformation of the principal's virtual value. The optimality conditions in \cref{thm:charac_opt} require this measure to satisfy particular \emph{convex dominance} properties on each cell of the partition.\footnote{In doing so, we follow a similar approach to \cite{Rochet1998}, \cite{Daskalakis2017}, \cite{Kleiner2019}, and \cite{Kleiner2022}, who derive analogous optimality conditions in the multiple goods monopolist and multidimensional delegation problems.} 

Together, \cref{thm:ext_pt,thm:charac_opt} thus offer a methodology for addressing economic design problems that are regarded as challenging. In these problems, since the signed measure is derived from the economic primitives, it also inherits their \emph{regularity}. Under mild assumptions, the measure therefore admits a (signed) \emph{density function}. Since \cref{thm:ext_pt} restricts the possible shapes of candidates for optimality, one can often make an educated \emph{guess} about the form of optimal extreme points simply by inspecting the sign of the density on the domain of the CFI. \cref{thm:charac_opt} then allows one to \emph{verify} this guess. In particular, when the measure is sufficiently regular, the convex dominance conditions stated in \cref{thm:charac_opt} become easy to check, as they only bear on its density function.

We apply our theoretical results to derive concrete implications for \emph{constrained} versions of several canonical problems in economic design. We derive the feasible, extreme, and optimal mechanisms in the monopolistic screening problem, the optimal delegation problem, optimal contest design with a large number of agents, and Bayesian persuasion when the sender's payoff depends only on the posterior mean. Relative to their classical formulations, representing feasible mechanisms as CFIs in these applications makes it possible to incorporate constraints that are intractable with existing methods. Next, we apply \cref{thm:ext_pt} to characterize the extremal mechanisms. We also impose specialized conditions on the primitives---which determine the corresponding signed measure---and use \cref{thm:charac_opt} to identify features of optimal extreme points. 

In the screening and delegation problems, we enrich the classical frameworks by allowing agents to select from a rich menu of outside options, which induces \emph{type-dependent} participation constraints \citep{Jullien2000}. Building on \citet{Rochet1987} and \citet{Kleiner2022}, we represent the feasible mechanisms in both problems as CFIs of the induced indirect utilities, with the participation constraints encoded in the lower level boundary. Following our \cref{thm:ext_pt}, extremal screening and delegation mechanisms feature randomization and exclusion. Using \cref{thm:charac_opt}, we further show that posted prices need not be optimal when the menu of outside options is sufficiently rich. We also show a novel comparative statics result: in the canonical delegation environment with log-concave type density and constant bias \citep{Martimort2006}, improving the agent's outside options leads the principal to grant him more discretion under the optimal mechanism.

We also study the role of \emph{limited disposal} in large contest design. Large contest models usually retain the assumption that either the principal must assign all available prizes, or that the principal can dispose of them freely. Our framework covers all the intermediate cases in which the principal faces a lower bound on the average quality he must allocate. Under this constraint, we show that the set of feasible prize assignments forms a \emph{two-sided weak-majorization} CFI whose elements correspond to cumulative expected quantile assignment rules. Using our \cref{thm:ext_pt}, we show that the structure of extremal assignments can depend on whether the disposal constraint binds. We also show that, when the type distribution is sufficiently regular, effort-maximizing contests exclude types below some cutoff and implement the positive assortative assignment for types above the cutoff. In particular, we show that the exclusion cutoff decreases when disposal becomes more limited. Economically, this means that setting more stringent minimal average quality standards forces any effort-maximizing designer to reward more types.

Finally, we study the impact of \emph{informativeness} constraints in optimal (mean-based) Bayesian persuasion. Imposing bounds on informativeness implies that the feasible set of posterior-mean distributions is given by a \emph{two-sided majorization} CFI. As a corollary of \cref{thm:ext_pt}, we characterize all extremal posterior-mean distributions that satisfy some informativeness bounds. We further show a new comparative statics result: in the canonical setting with an S-shaped value function for the sender \citep{Kolotilin2022}, increasing the minimal required informativeness shrinks the optimal censorship region (equivalently, enlarge the full-revelation region).

In addition to \cref{thm:ext_pt,thm:charac_opt}, we provide a further theoretical contribution. Consider a simplified version of the problem introduced by \cite{Dworczak2024}: in a linear monopolistic screening environment, a welfare-maximizing planner \emph{designs} a menu of outside option for the agents taking into account that, downstream, the agents will face a selling mechanism designed by a revenue-maximizing principal. In this setting, the lowest implementable indirect utility for the principal corresponds to the one induced by the planner's chosen menu. Therefore, \citeauthor{Dworczak2024}'s problem can equivalently be represented in our framework as the planner optimizing the \emph{lower level boundary} of the CFI of indirect utilities subject to the principal's revenue-maximization constraint. We analyze such \emph{nested design} problems in an abstract form. Our \cref{thm:opt_lower_bd} shows that, for CFIs (i) which upper level boundary is affine and either has the minimal or maximal admissible slope, and (ii) which lower level boundary coincides with the upper level boundary at one endpoint of the domain, the planner's value is \emph{linear} in the lower level boundary. This class of CFIs includes the screening environment considered by  \cite{Dworczak2024}. We use \cref{thm:opt_lower_bd} to derive welfare-optimal menus of outside options in the monopolistic screening problem, and show, consistent with \citet{Dworczak2024}, that they take the form of menus including an ``\emph{option-to-own}'', i.e., allowing agents either to opt out or to purchase the good at a fixed price.

\subsection{Discussion of the Literature}

\paragraph{Extreme points in mechanism design}

We follow a recent and active strand of literature that derives insights in different economic design problems sharing an overarching mathematical structure. Many of these structures are \emph{convex}, which naturally leads to the study of their \emph{extreme points}. The contribution of \citet*{Kleiner2021} provides such a characterization for the sets of non-decreasing functions that \emph{majorize}, or are \emph{majorized} by, another non-decreasing function.\footnote{They extend the characterization of \citet{Ryff1967}, who do not impose monotonicity. See also \citet*{Arieli2023}, who obtain an identical characterization of the extreme points of mean-preserving contractions of a given distribution.} They apply this characterization to several economic problems, including multi-unit auctions, large contests, mean-based persuasion, and optimal delegation. Subsequent work further develops this approach. \citet{Nikzad2022} characterize the extreme points of the sets studied by \citet{Kleiner2021} under finitely many additional linear constraints. \citet{Kleiner2025} identify the (Lipschitz-)\emph{exposed}\footnote{Exposed points are extreme points that maximize uniquely some continuous linear function.} points of \emph{fusions}, the multidimensional analogue of mean-preserving contractions. In parallel, \citet{Yang2024} characterize the extreme points of \emph{monotone function intervals}---sets of non-decreasing functions bounded above and below by two given functions---and apply this framework to gerrymandering, quantile-based persuasion, apparent misconfidence, and security design. \citet{Yang2025} extend the analysis to higher dimensions by characterizing the extreme points of multidimensional monotone functions. They apply this characterization to optimal mechanism design with correlated values and information, and optimal information design subject to privacy constraints. Relatedly, \citet{Lahr2024} characterize the extremal elements of the set of incentive-compatible and individually-rational mechanisms in multidimensional linear adverse selection environments. They exploit the theory of \emph{indecomposable convex bodies}\footnote{Convex bodies are compact convex subsets of Euclidean space with nonempty interior.} to identify extremal mechanisms in menu form. They extend an earlier contribution of \citet{Manelli2007}, who identify extreme points and faces of the set of implementable (hence convex) indirect utility functions in the multidimensional monopolistic screening problem. In contrast, we study one-dimensional adverse selection problems but allow type-dependent participation constraints, as we discuss now.

\paragraph{Participation constraints in adverse selection models}

Standard treatments of adverse selection assume that agents' outside option is fixed and \emph{independent} of type. In richer screening environments this assumption typically fails \citep{Lewis1989,Maggi1995}. Type-dependent participation constraints, however, are difficult to handle even in one dimension.\footnote{To the best of our knowledge, solving multidimensional screening problems with type-dependent participation constraints remains out of reach. A notable exception is \citet{Rochet1998}, who consider default menus with a single non-null outside option in the monopolistic screening problem.} The usual approach, following \citet{Jullien2000} and \citet{Amador2013}, consists of putting Lagrange multipliers on the continuum of participation constraints.\footnote{See also \citet{Martimort2022}, who extend \citeauthor{Jullien2000}'s (\citeyear{Jullien2000}) framework to possibly discontinuous objectives.} These multipliers can be seen as a measure supported on the set of types for which the participation constraint binds.  Solving the model then requires (i) conjecturing the form of the (infinite-dimensional) Lagrange multipliers, (ii) solving the relaxed problem given the conjecture, and (iii) verifying global optimality via a sufficiency theorem. \citet{Dworczak2024} propose an alternative for \emph{linear} environments, extending \citeauthor{Myerson1981}'s (\citeyear{Myerson1981}) ironing to incorporate participation constraints. We propose a complementary methodology for \emph{linear} adverse selection with type-dependent outside options: rather than Lagrangian or ironing arguments, we characterize optimal solutions directly as extreme points of the set of \emph{indirect utilities} satisfying participation constraints.

\paragraph{Extremal convex functions}

Our paper also connects to early contributions in mathematics that characterize the \emph{extreme rays} of the convex cone of all convex functions.\footnote{A subset $R$ of a cone $K$ is called a ray if $R = \{\lambda x \; \vert \; \lambda\ge 0\}$ for some $x\in K\smallsetminus\{0\}$. A ray $R$ is called \emph{extreme} if the following is satisfied: for any $x,y\in K$, if $x+y\in R$ then $x,y\in R$. By contrast, an \emph{extreme point} of a convex set $C$ is $x\in C$ such that $x=\lambda y+(1-\lambda)z$ with $y,z\in C$ and $\lambda\in(0,1)$ implies $x=y=z$. Extreme rays play the same role for closed convex pointed cones than extreme points for compact convex sets: any such cone is generated as the conical hull of its extreme rays.} In one dimension, \cite{Blaschke1916} show that all extreme rays are generated by functions of the form $a\vee b$ where $a$ and $b$ are affine functions. \cite{Johansen1974} extends this analysis to convex functions on $\mathbb{R}^{2}$, and \cite{Bronshtein1978} generalizes it to $\mathbb{R}^{d}$. Both \citeauthor{Johansen1974} and \citeauthor{Bronshtein1978} prove that the set of extremal convex functions is \emph{dense} in the cone.\footnote{\cite{Lahr2024} stress the same challenge in multidimensional screening problems: in one dimension, extremal mechanisms are tractable, but in multiple dimensions the set of extremal mechanisms is essentially as large as the full set of incentive-compatible mechanisms.}  Additionally, \citeauthor{Bronshtein1978}'s Theorems 5.1 and 5.2 characterize the extreme points of certain compact convex subsets of one-dimensional bounded convex functions, which our \cref{thm:ext_pt} recovers as special cases. Recently, \cite{Baillo2022} also characterize the extreme points of the (compact and convex) set of convex functions $f\colon [0,1] \to [0,1]$ that satisfy $f(0)=0$, $f(1)=1$ and have a fixed area, such as families of Lorenz curves with a fixed Gini index.

\section{The Structure of Convex Function Intervals}\label{sec:structure_CFI}

\subsection{Notations and Definitions}

\subsubsection{Mathematical Preliminaries}

For any compact interval $X\subset\mathbb{R}$, we let $\mathcal{C}(X)$ be the set of real-valued continuous functions on $X$, and $\mathcal{K}(X)$ be the set of real-valued continuous convex functions on $X$.

For any $u\in\mathcal{K}(X)$ and $x\in X$, the subdifferential of $u$ at $x$ is defined as
\begin{equation*}
    \partial u(x) = \Bigl\{s\in \mathbb{R} \; \big\vert \; \forall y\in X, \; u(y) \geq u(x) + s(y-x)\Bigr\}.
\end{equation*}

Any $s\in \partial u(x)$ is called a subgradient of $u$ at $x$. We let $\partial u(X) \coloneqq \bigcup_{x\in \Int(X)} \partial u(x)$. For any convex function $u$, its left and right derivatives $\partial_{-}u$ and $\partial_{+}u$ exist everywhere on $X$, are non-decreasing functions, and satisfy $\partial_{-}u \leq \partial_{+}u$ \cite[][Theorems 0.6.3 and 0.6.4]{Hiriart2001}. Furthermore,  $\partial u(x) = [\partial_{-}u(x), \partial_{+}u(x)]$ for all $x\in \Int(X)$, the set $\bigl\{x\in X \; \vert \; \partial_{-}u(x)< \partial_{+}u(x)\bigr\}$ is countable and, hence, $u$ is differentiable almost everywhere on $X$, and $\partial u(x)=\{u'(x)\}$ wherever $u$ is differentiable \citep[][Theorem B.4.2.3]{Hiriart2001}.

For any function $u\colon \mathbb{R} \to \mathbb{R}$, and any $A\subset \mathbb{R}$, we denote the restriction of $u$ to $A$ as $u\vert_{A}$. For any set $A\subset \mathbb{R}$, $\Id_{A}$ denotes the identity function on $A$, and $\Ind_{A}$ denotes the indicator function of $A$. When $A=[0,1]$ we omit the explicit dependence on $A$.

For any $x\in X$, we denote as $\delta_{x}$ the Dirac measure at $x$, and the Lebesgue measure as $\lambda$. For any (signed) measure $\mu$ on $(\mathbb{R}, \mathcal{B}(\mathbb{R}))$ and any (Borel) measurable set $A\in\mathcal{B}(\mathbb{R})$, we let $\mu\vert_{A}$ be the conditional measure of $\mu$ given $A$.\footnote{We refer to \cite{Bogachev2007}, Chapter 10, for a formal definition of conditional measures.}

\subsubsection{Convex Function Intervals}

Having introduced all the necessary preliminaries, we now turn to our main object of study.

\begin{definition}[Convex function intervals]\label{def:CFI}
The set $\mathcal{U}$ is called a \emph{convex function interval} (CFI for short) if there exist two compact intervals $X\coloneqq [a,b]\subset\mathbb{R}$ and $S\coloneqq [\ubar{s},\bar{s}]\subset\mathbb{R}$ and two functions $\ubar{u},\bar{u}\in \mathcal{K}(X)$ such that $\ubar{u}\leq\bar{u}$ and $\partial \ubar{u}(X),\partial \bar{u}(X)\subseteq S$, and
\begin{equation*}
  \mathcal{U} = \Bigl\{ u\in \mathcal{K}(X) \; \big\vert \; \ubar{u} \leq u \leq \bar{u}, \, \partial u(X) \subseteq S \Bigr\}.
\end{equation*}
\end{definition}

A convex function interval, $\mathcal{U}$, is therefore the set of real-valued continuous and convex functions that are defined on some compact interval of real numbers $X$, which slopes lie between the two prescribed bounds $\ubar{s}$ and $\bar{s}$, and which graphs are sandwiched between two continuous, convex boundary functions $\ubar{u}$ and $\bar{u}$ (assumed to satisfy the slope constraints themselves). Note that any convex function must be continuous in the interior of its domain \citep[][Theorem 6.2]{Hiriart2001}. Hence, imposing continuity on the whole $X$ is in fact equivalent to only requiring that any $u$ in $\mathcal{U}$ is continuous at the endpoints of $X$.

We also maintain the following assumption throughout the paper for tractability.
\begin{assumption}\label{assu:diff}
    The upper level boundary $\bar{u}$ is differentiable on $X$. Hence, its derivative $\bar{u}'$ exists and is a continuous function on $X$.\footnote{A convex function is differentiable if and only if it continuously differentiable \citep[see][Remark 6.2.6]{Hiriart2001}.}
\end{assumption}

In most of \cref{sec:structure_CFI,sec:optimization_CFI} and of the proofs in the Appendix, we let the domain $X$ be the unit interval $[0,1]$. This assumption is a normalization and is without loss of generality. We also adopt the shorthand notations $\mathcal{C}\coloneqq \mathcal{C}\bigl([0,1]\bigr)$ and $\mathcal{K}\coloneqq \mathcal{K}\bigl([0,1]\bigr)$.

\begin{remark}\label{remark:assumptions}
    Several assumptions in our setup are made for expositional clarity and unity, and can be relaxed without fundamentally altering our results. Our characterization of extreme points (\cref{thm:ext_pt}) can be extended to cases where $\bar{u}$ is not differentiable. Similarly, \cref{thm:ext_pt} extends to ``open intervals'' of convex functions, where one of the level boundaries takes an infinite value. However, allowing infinite boundaries breaks the compactness of CFIs, which is needed if one wants to apply Choquet's theorem (claim (ii) of \cref{prop:representation} below). Finally, all our results apply verbatim to concave function intervals by reversing the signs.
\end{remark}

\subsection{Monopolistic Screening with (Endogenous) Type-Dependent Participation Constraints}\label{sec:screening_run_example}

In the remainder of \cref{sec:structure_CFI} and in the upcoming \cref{sec:optimization_CFI}, we regularly use the following running example to illustrate how our abstract results apply in a workhorse mechanism design setting. In \cref{sec:applications}, we cover in detail additional economic applications (and implications) of our framework.

\subsubsection{Model}

\paragraph{Primitives}

A principal has a unit mass of homogeneous goods that he can allocate to a population of agents, modeled as a continuum of unit mass. Each agent has a privately known type $\theta\in \Theta=[0,1]$ representing her willingness to pay for the good. An agent of type $\theta$ receiving a good with probability $x\in[0,1]$ at a monetary cost $t\in\mathbb{R}_{+}$ hence obtains a surplus equal to $\theta x - t$.\footnote{One can also always interpret $x$ as a quantity of goods, or the good's quality, instead of an allocation probability, where the maximal quantity, or quality, has been normalized to one.} Agents' types are identically and independently distributed according to the cumulative distribution function $F$, which admits a differentiable and strictly positive density function $f$.

\paragraph{Selling mechanisms}

By the revelation principle \citep{Myerson1981}, it is without loss of generality to focus on direct selling mechanisms $(x,t)\colon\Theta\to[0,1]\times \mathbb{R}_{+}$, which consist of an allocation rule $x\colon \Theta\to [0,1]$ and a transfer rule $t\colon\Theta\to\mathbb{R}_{+}$ that specify, respectively, the probability $x(\theta)$ with which an agent receives a good, and the transfer $t(\theta)$ the agent must pay to the principal when reporting type $\theta$.

\paragraph{Incentive constraints}

A mechanism satisfies \emph{incentive-compatibility} if all the agents have an incentive to report their types truthfully. Formally,
\begin{equation}\label{eqn:IC_screening}
    \forall \theta,\theta'\in\Theta, \quad \theta x(\theta) - t(\theta) \geq \theta x(\theta') - t(\theta'). \tag{IC}
\end{equation}

We assume that the agents can always flexibly choose their preferred option from a (compact) \emph{default menu} of outside options $M_{0}$ such that
\begin{equation*}
    \bigl\{(0,0)\bigr\}\subseteq M_{0}\subset \bigl\{(x,t) \; \vert \;  x \in[0,1], \; t\in \mathbb{R}_{+} \bigr\},
\end{equation*}
rather than participating in the mechanism proposed by the principal. A mechanism satisfies \emph{individual-rationality} if it guarantees each type at least as much surplus as its favorite option in $M_{0}$. Formally,
\begin{equation}\label{eqn:IR_screening}
    \forall\theta\in\Theta, \quad \theta x(\theta) - t(\theta) \geq u_{0}(\theta) \coloneqq \max_{(x,t)\in M_{0}} \; \theta x-t. \tag{IR}
\end{equation}

\paragraph{Indirect utility functions}\label{sec:screening_run_example_CFI}
 
The agents' \emph{indirect utility function} induced by some mechanism $(x,t)$ is defined as
\begin{equation*}
  \forall \theta\in\Theta, \quad u(\theta) = \max_{\theta'\in \Theta} \; \theta x(\theta') - t(\theta').
\end{equation*}

An indirect utility function $u\colon\Theta\to \mathbb{R}$ is \emph{implementable} if there exists a mechanism $(x,t)$ which satisfies \labelcref{eqn:IC_screening,eqn:IR_screening} such that $u(\theta) = \theta x(\theta) - t(\theta)$ for every $\theta\in\Theta$.

\subsubsection{Feasible mechanisms as a CFI}

\paragraph{The screening interval}

We follow the approach of \cite{Rochet1987} by characterizing mechanisms satisfying \labelcref{eqn:IC_screening,eqn:IR_screening} by their induced indirect utility functions. Specifically, we show that, for any menu $M_{0}$, the set of implementable indirect utility functions is a CFI, that we call the \emph{screening interval}.

\begin{lemma}\label{lem:screening_interval}
    An indirect utility function $u$ is \emph{implementable} if and only if $u\in \mathcal{U}_{\mathrm{S}}$, where
    \begin{equation*}
        \mathcal{U}_{\mathrm{S}} \coloneqq \Bigl\{u\in \mathcal{K}(\Theta) \; \big\vert \; u_{0} \leq u \leq \Id_{\Theta}, \, \partial u(\Theta) \subseteq [0,1] \Bigr\}.
    \end{equation*}
    
    Moreover, if a mechanism $(x,t)$ implements $u\in\mathcal{U}_{\mathrm{S}}$, then $x(\theta)\in\partial u(\theta)$ and $t(\theta)= \theta x(\theta) - u(\theta)$ for all $\theta\in\Theta$.
\end{lemma}

\cref{lem:screening_interval} follows from standard arguments that we omit for brevity \citep[see][Proposition 2]{Rochet1987}. Note that the upper bound on implementable utilities, $\Id_{\Theta}$, is achieved by the mechanism that gives away the good for free to all agents.

\begin{remark}
    By \labelcref{eqn:IR_screening}, $u_{0}$ is convex since it is defined as the supremum of a family of affine functions. According to \cref{lem:screening_interval}, the principal  can therefore implement $u_{0}$ by some \emph{default mechanism} $(x_{0},t_{0})$ that satisfies \labelcref{eqn:IC_screening}.\footnote{Moreover, by the \emph{Taxation Principle}, $(x_{0},t_{0})$ equivalently represents the agents' optimal choice function from the menu $M_{0}$.} The underlying assumption that $(x_{0},t_{0})$ satisfies \labelcref{eqn:IC_screening} is what \cite{Jullien2000} calls \emph{homogeneity}. This property is crucial for $\mathcal{U}_{\mathrm{S}}$ to be a CFI. Intuitively, homogeneity ensures that the \labelcref{eqn:IC_screening} and \labelcref{eqn:IR_screening} constraints never conflict with each other.
\end{remark}

\subsection{The Extreme Points of Convex Function Intervals}\label{subsec:extremepoints}

We now examine the structural properties of CFIs.  As a preliminary, we establish in \cref{sec:vexpact} that CFIs are compact and convex subsets of the space of continuous functions and are hence representable as the closed convex hull of their extreme points (\cref{prop:representation}). In \cref{sec:extpt_charac}, we state our first main result (\cref{thm:ext_pt}) which is a geometric characterization of CFIs' extreme points. We describe this characterization intuitively in \cref{sec:interpretation_ext_pt}. We then sketch the argument for the proof of \cref{thm:ext_pt} in \cref{sec:proof_sketch_ext_pt}. We conclude in \cref{sec:screening_run_example_ext_pt} by illustrating this characterization in the example introduced in \cref{sec:screening_run_example}.

\subsubsection{Representation of CFIs}\label{sec:vexpact}

An \emph{extreme point} of a convex set $C$ is an element $x$ of $C$ such that for all $\lambda\in (0,1)$ and $y,z\in C$, the equality $x = \lambda y + (1-\lambda) z$ implies that $x=y$ or $x=z$.

\begin{proposition}\label{prop:representation}
  Let $\mathcal{U}$ be a CFI. Then, the two following claims are true:
  \begin{enumerate}[(i)]
      \item The set $\mathcal{U}$ is a non-empty, convex, and compact subset of $\mathcal{C}$ endowed with the supremum norm $\lVert\cdot\rVert_{\infty}$. Therefore, $\mathcal{U}$ admits a non-empty set of extreme points, denoted as $\ex(\mathcal{U})$.
      \item For any $u \in \mathcal{U}$, there exists a probability measure $\mu$ such that $\supp(\mu)\subseteq \ex(\mathcal{U})$ and
      \begin{equation*}
          u=\int_{\ex(\mathcal{U})} u^{\star} \, \diff \mu(u^{\star}).
      \end{equation*}
      
      Hence, $\mathcal{U}=\overline{\mathrm{co}}\bigl(\ex(\mathcal{U})\bigr)$.
  \end{enumerate}
\end{proposition}

The proof of \cref{prop:representation} can be found in \cref{secap:proof_representation}. Claim (i) follows from the Arzelà-Ascoli and Krein-Milman theorems. Claim (ii) follows from Choquet's theorem (a strengthening of Krein-Milman Theorem), and states that any element of a CFI, $\mathcal{U}$, can be obtained as a (potentially infinite) convex combination of some extreme points of $\mathcal{U}$. 

\subsubsection{Characterization of Extreme Points}\label{sec:extpt_charac}

Although identifying extreme points in functional spaces is typically a challenging endeavor, we derive a relatively simple geometric characterization of CFIs' extreme points:

\begin{theorem}[Extreme points]\label{thm:ext_pt}
  Let $\mathcal{U}$ be a CFI. A function $u\in\mathcal{U}$ belongs to $\ex(\mathcal{U})$ if and only if there exists a (possibly empty) countable collection $\mathcal{X}=\{X_{n}\}_{n\in\mathbb{N}}$ of maximal and non-degenerate\footnote{By non-degenerate we mean that, for each $n\in \mathbb{N}$, $a_{n}<b_{n}$. Maximality refers to the property that for each $X_{n}=[a_n, b_n] \in \mathcal{X}$, for all $\varepsilon>0$, either $u\vert_{[a_n-\varepsilon, b_n]}$ or $u\vert_{[a_n, b_n + \varepsilon]}$ do not satisfy the statement in condition \labelcref{cond:not_level_satur}.} intervals $X_{n}=[a_{n},b_{n}]\subseteq X$ such that:
  \begin{enumerate}
      \item For all $x\notin \bigcup_{n\in\mathbb{N}} X_{n}$, $u(x)\in\bigl\{\ubar{u}(x), \bar{u}(x)\bigr\}$. \label{cond:level_satur}
      \item For each $n\in\mathbb{N}$, $u\vert_{X_{n}}$ is affine, $\ubar{u}<u\vert_{\Int(X_{n})}<\bar{u}$, and at least one of the following conditions holds: \label{cond:not_level_satur}
      \begin{enumerate}
        \item There exists $y\in\{a_{n},b_{n}\}$ such that, for all $x\in X_{n}$, $u(x)=\bar{u}(y)+s(x-y)$ with
        \begin{equation*}
            s \left\{
            \begin{array}{ll}
                = \bar{u}'(y)& \text{if $y \in (0,1)$} \\[3pt]
                \in \bigl\{\ubar{s}, \partial_{+}\bar{u}(y) \bigr\} & \text{if $y=0$} \\[3pt]
                \in \bigl\{\bar{s}, \partial_{-}\bar{u}(y) \bigr\} & \text{if $y=1$}.
            \end{array}
            \right.
        \end{equation*}
        \label{cond:tangent_bounds}
        \item For each $x \in \{a_n, b_n\}$, either there exists $m\in\mathbb{N}$ such that $b_m=a_n$ or $b_n=a_m$ and $u\vert_{X_{m}}$ satisfies condition \labelcref{cond:tangent_bounds}, or $u(x)=\ubar{u}(x)$.\label{cond:chord_satur}
        \item Either $a_{n}=0$, $u'\vert_{X_{n}}=\ubar{s}$, and either there exists $m\in\mathbb{N}$ such that $a_m=b_n$ and $u\vert_{X_{m}}$ satisfies condition \labelcref{cond:tangent_bounds} or $u(b_{n})=\ubar{u}(b_{n})$.
        
        Or, symmetrically, $b_{n}=1$, $u'\vert_{X_{n}}=\bar{s}$, and either there exists $m\in\mathbb{N}$ such that $a_n=b_m$ and $u\vert_{X_{m}}$ satisfies condition \labelcref{cond:tangent_bounds} or $u(a_{n})=\ubar{u}(a_{n})$. \label{cond:matching_grad}
        \item Either $a_n=0$, $u(a_n) \in \{\ubar{u}(a_n), \bar{u}(a_n)\}$ and either there exists $m\in\mathbb{N}$ such that $a_m=b_n$ and $u\vert_{X_{m}}$ satisfies condition \labelcref{cond:tangent_bounds} or $u(b_{n})=\ubar{u}(b_{n})$.
        
        Or, symmetrically, $b_n=1$, $u(b_n)\in \{\ubar{u}(b_n), \bar{u}(b_n)\}$ and either there exists $m\in\mathbb{N}$ such that $a_n=b_m$ and $u\vert_{X_{m}}$ satisfies condition \labelcref{cond:tangent_bounds} or $u(a_{n})=\ubar{u}(a_{n})$. \label{cond:boundary_cases}
      \end{enumerate}
  \end{enumerate}
\end{theorem}

The complete proof of \cref{thm:ext_pt} can be found in \cref{secap:ext_pt_proof}, and a sketch of the argument can be found in \cref{sec:proof_sketch_ext_pt}. We now explain in more detail the conditions stated in \cref{thm:ext_pt}.

\subsubsection{Interpretation of the Conditions in \cref{thm:ext_pt}}\label{sec:interpretation_ext_pt}

The elements of CFIs must satisfy three kinds of constraints: \emph{convexity}; the \emph{level} boundaries $\ubar{u}$ and $\bar{u}$; and the \emph{slope} boundaries $\ubar{s}$ and $\bar{s}$. A heuristic principle is that extreme points of any compact and convex set must saturate at least one of the constraints that define this set. The extreme points of CFIs are no exception to this rule. As \cref{thm:ext_pt} shows, this leads to a set of simple geometric conditions that characterize the extreme points of $\mathcal{U}$. In the following paragraphs, we explain each of those conditions. In doing so, we also provide a terminology that will be used throughout the rest of the paper. 

Let $u\in \ex(\mathcal{U})$, and $\mathcal{X}=\{X_{n}\}_{n\in\mathbb{N}}$ be its corresponding collection of intervals from \cref{thm:ext_pt}. 

\paragraph{Level saturation}

As illustrated in \cref{fig:extreme_point}, at any point $x\in X$, $u$ either make one of the level boundary constraints bind if $x\notin \bigcup_{n\in\mathbb{N}} X_{n}$ (condition \labelcref{cond:level_satur}), or leaves them slack if there exists $n\in\mathbb{N}$ such that $x\in \Int(X_{n})$ (condition \labelcref{cond:not_level_satur}). 
\begin{figure}
    \centering
    \begin{subfigure}[t]{0.495\linewidth}
        \centering
        \includegraphics{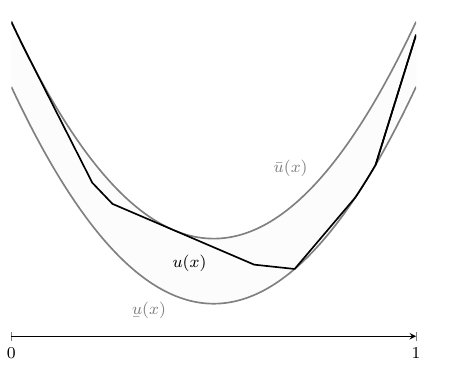}
        \caption{An extreme point $u$ of $\mathcal{U}$.}
        \label{fig:extreme_point}
    \end{subfigure}
    \begin{subfigure}[t]{0.495\linewidth}
        \centering
        \includegraphics{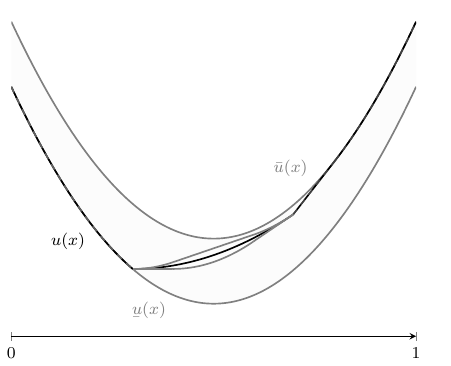}
        \caption{A non-extreme point $u$ of $\mathcal{U}$.}
        \label{fig:not_extreme_point}
    \end{subfigure}
    \label{fig:ext_non_ext}
    \caption{Extreme points of CFIs. The function on the left panel is an extreme point. In particular, its rightmost affine segment is assumed to have a slope equal to $\bar{s}$. The function on the right panel is not an extreme point: there exists $h\in\mathcal{C}$ such that $u\pm h\in\mathcal{U}$.}
\end{figure}
We thus say that $u$ satisfies \emph{level saturation} at $x$ if $x\notin \bigcup_{n\in\mathbb{N}} X_{n}$. Note that since $\mathcal{X}$ can be empty, $\ubar{u}$ and $\bar{u}$ are themselves extreme points of $\mathcal{U}$.

\paragraph{Convexity saturation}

We say that $u$ satisfies \emph{convexity saturation} on some subinterval of $X$ if it is \emph{affine} on this subinterval (condition \labelcref{cond:not_level_satur}).\footnote{Note that, according to condition \labelcref{cond:not_level_satur}, if $u$ is not level-saturated at some $x$, then $u$ must saturate the convexity constraint on an interval around $x$. Nevertheless, $u$ can also saturate \emph{both} the level and convexity constraint if $u$ matches $\ubar{u}$ (resp.~$\bar{u}$) on an interval where $\ubar{u}$ (resp.~$\bar{u}$) is affine.} As will be detailed in \cref{sec:applications}, convexity saturation is crucial in economic applications as it relates closely to the concepts of \emph{bunching}, \emph{ironing}, and \emph{mean-preserving spreads and contractions}. We further divide convexity saturation into three subconditions: \emph{tangential saturation}, \emph{slope saturation} and \emph{chordal saturation}.
    
\paragraph{Tangential saturation}

For any $n\in\mathbb{N}$, we say that $u$ satisfies \emph{tangential saturation} on $X_{n}$ if $u\vert_{X_{n}}$ satisfies condition \labelcref{cond:tangent_bounds}. These segments correspond to a special case of convexity saturation, whereby $u$ expands linearly either to the right or to the left from a point of tangency with $\bar{u}$. Tangential saturation determines both the slope and the level of $u$ on a specific interval. This is a consequence of $\bar{u}$ being differentiable on $X$ (\cref{assu:diff}), which forces the slope of $u$ to match the unique subgradient of $\bar{u}$ at the point of tangency.

\paragraph{Slope saturation}

For any $n\in\mathbb{N}$, we say that $u$ satisfies \emph{slope saturation} on $X_{n}$ if $u\vert_{X_{n}}$ satisfies condition \labelcref{cond:matching_grad}. These segments correspond to intervals where the derivative of $u$ meets one of the slope boundaries $\ubar{s}$ or $\bar{s}$. The convexity of $u$ implies that $u'\vert_{X_{n}}$ can only be equal to $\ubar{s}$ (resp.~$\bar{s}$) on an interval that contains $x=0$ (resp.~$x=1$). This is a consequence of the subdifferential of convex functions being monotone.\footnote{See Proposition D.6.1.1 in \cite{Hiriart2001}.} Moreover, $u'\vert_{X_{n}}$ reaching one the boundary values  $\ubar{s}$ or $\bar{s}$ is not sufficient for $u$ to be an extreme point, because the level of $u$ must also be disciplined. This is why condition \labelcref{cond:matching_grad} also imposes that, at the endpoints of $I$ that are not endpoints of $X$, $u\vert_{I}$ must either touch the lower bound $\ubar{u}$, or must be followed or preceded by an affine segment satisfying the condition \labelcref{cond:tangent_bounds}.

\paragraph{Chordal saturation}

For any $n\in\mathbb{N}$, we say that $u$ satisfies \emph{chordal saturation} on $X_{n}$ if $u\vert_{X_{n}}$ satisfies condition \labelcref{cond:chord_satur}. These segments correspond to the intervals where $u$ is affine but satisfies neither tangential saturation nor slope saturation. In other words, $u$ is given by a \emph{chord} segment linking two points on the graph of $u$ that are not confounded with $\bar{u}$ with slope that is neither $\ubar{s}$ nor $\bar{s}$. Note that segments of chordal saturation must satisfy additional boundary conditions. If $u\vert_{X_{n}}$ does not coincide with $\ubar{u}$ at one of the endpoints $a_{n}$ or $b_{n}$, then $u$ needs to be preceded and followed by intervals of tangential saturation to discipline its level.

\subsubsection{Extreme Points of the Screening Interval}\label{sec:screening_run_example_ext_pt}

Let us return to the screening example from \cref{sec:screening_run_example}. Applying our \cref{thm:ext_pt} to $\mathcal{U}_{\mathrm{S}}$ and recalling from \cref{lem:screening_interval} that any mechanism implementing $u\in \mathcal{U}_{\mathrm{S}}$ must satisfy $x(\theta)\in \partial u(\theta)$ for all $\theta\in \Theta$, we find that extremal indirect utilities are characterized by three types of intervals: \emph{non-participation} intervals where the principal offers agents their preferred outside option within $M_{0}$; \emph{bunching} intervals where the principal's mechanism offers the good stochastically with a constant allocation probability; and an ultimate interval where the good is allocated \emph{deterministically}. On each bunching interval, the allocation probability is determined by \emph{ironing} the \emph{default allocation rule}. We illustrate this structure in \cref{fig:extremal_u_and_x_screen}, and formalize it in \cref{cor:imp_ext_pt_screening}.

\begin{figure}[t]
    \begin{subfigure}[t]{0.49\linewidth}
    \centering
        \includegraphics{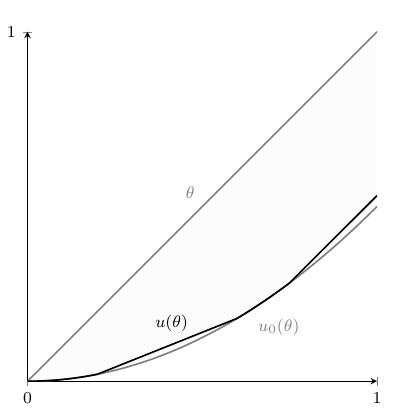}
        \caption{An extreme point $u$ of $\mathcal{U}_{\mathrm{S}}$.}
        \label{fig:extreme_point_screening}
    \end{subfigure}
    \begin{subfigure}[t]{0.49\linewidth}
    \centering
        \includegraphics{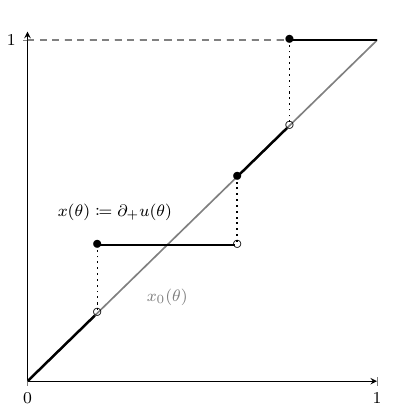}
        \caption{An allocation rule $x$ implementing $u$.}
        \label{fig:ext_alloc}
    \end{subfigure}
    \caption{An extremal indirect utility and its  implementing allocation rule. In this picture, we take $x_{0}(\theta)=\theta$ for all $\theta\in\Theta=[0,1]$.}
    \label{fig:extremal_u_and_x_screen}
\end{figure}

\begin{corollary}[Extremal allocation rules] \label{cor:imp_ext_pt_screening}
  Let $x_{0}(\theta)\in \partial u_{0}(\theta)$ for all $\theta\in \Theta$. Any $u\in \ex(\mathcal{U}_{\mathrm{S}})$ can be implemented by an allocation rule $x$ described by a collection $\{\Theta_{n}\}_{n\in\mathbb{N}}$ of maximal non-degenerate intervals $\Theta_{n}=[\ubar{\theta}_{n},\bar{\theta}_{n}]\subseteq\Theta$ such that:
  \begin{enumerate}[(i)]
    \item For all $\theta\notin\bigcup_{n\in\mathbb{N}} \Theta_{n}$, $x(\theta)=x_{0}(\theta)$.
    \item For each $n\in\mathbb{N}$, $x\vert_{\Theta_{n}}$ is constant and
    \begin{enumerate}
      \item either $\displaystyle x(\theta)=\frac{\int_{\ubar{\theta}_{n}}^{\bar{\theta}_{n}} x_{0}(s) \, \diff s}{\bar{\theta}_{n}-\ubar{\theta}_{n}}$ for all $\theta\in \Theta_{n}$; or,
      \item $\bar{\theta}_{n}=1$ and $x\vert_{\Theta_{n}}=1$.
    \end{enumerate}
  \end{enumerate}
  We call allocation rules satisfying these conditions \emph{extremal}.
\end{corollary}

\cref{cor:imp_ext_pt_screening} is consistent with the findings of \cite{Dworczak2024}, who consider the same environment and obtain extremal allocations as optimal solutions for a general class of linear objective functions over allocation rules (see their Lemma 2).\footnote{The only difference is that \citeauthor{Dworczak2024} do not impose bounds on transfers, which yields an additional bottom region where $\ubar{\theta}_{n}=0$ and $x\vert_{\Theta_{n}}=0$. Our result would incorporate such a region if we allowed $t(\theta) \geq \ubar{t}$ with $\ubar{t}<0$ or unbounded transfers (which could be accommodated as discussed in \cref{remark:assumptions}).} We discuss the connection to their generalized ironing approach in more detail in \cref{remark:link_DM2024} below.

\subsubsection{Sketch of the Proof of \cref{thm:ext_pt}}\label{sec:proof_sketch_ext_pt}

The proof adopts a perturbation approach. An equivalent definition of extreme points is the following: $u\in\ex(\mathcal{U})$ if and only if there exists no $h \in \mathcal{C}$, $h\neq 0$, such that both $u \pm h \in \mathcal{U}$. Geometrically, this means that at an extreme point, it is impossible to remain within the feasible set $\mathcal{U}$ when moving in two opposite directions.

We prove the sufficiency part by contradiction. Consider some function $u$ that satisfies the properties described in \cref{thm:ext_pt}. Towards a contradiction, we assume that there exists $h\in \mathcal{C}$, $h\neq 0$, such that $u\pm h \in \mathcal{U}$. First, observe that if there exists $x_0\notin \bigcup_{n\in\mathbb{N}} X_{n}$ such that $h(x_0)>0$, then $(u-h)(x_0)<\ubar{u}(x_0)$ or $(u+h)(x_0)>\bar{u}(x_0)$ which is impossible by the definition of the CFI. Next, we assume that there exists $x_0\in I \in \mathcal{X}$ such that $h(x_0)>0$. We then look at what happens to $u\pm h$ on the interval $I=[a, b]$ and adjacent intervals and find contradictions to $u\pm h$ belonging to $\mathcal{U}$. If $u$ satisfies condition \labelcref{cond:tangent_bounds} on the interval $I$, then $u-h$ cannot be convex as there must exist $\hat{x}\in I$, $\hat{x}>a$, such that $(u-h)'(\hat{x})<\bar{u}'(a)=(u-h)'(a)$.\footnote{Note that, since $h=\hat{u}-u$ for some $\hat{u}\in\mathcal{U}$, its derivative $h'$ must exist except maybe at a countable number of points in $X$ (as the difference between two convex functions).} If $u$ satisfies condition \labelcref{cond:chord_satur} on that interval, then the previous steps imply that $h(a)=h(b)=0$. Moreover, remark $u\vert_{I}$ is the chord segment that links $u(a)$ and $u(b)$. As such, $u\vert_{I}$ is the pointwise highest continuous and convex function on $I$ that links $u(a)$ and $u(b)$. Hence, $(u+h)\vert_{I}$ cannot be convex. If $u$ satisfies condition \labelcref{cond:matching_grad} on the interval containing $x_{0}$, then $u\pm h$ has to break the slope constraints.

We prove necessity by contraposition. Consider some function $u\in\mathcal{U}$ that violates the conditions in \cref{thm:ext_pt}. The proof unfolds in two steps. 

Assume first that, for some $x\in X$ and some interval $I\subseteq X$ containing $x$, the function $u$ lies strictly in between $\ubar{u}$ and $\bar{u}$, has a subdifferential contained in the interior of $S$, but is not affine on $I$. As illustrated in \cref{fig:not_extreme_point}, we then construct a continuous perturbation $h$, such that $u - h$ and $u + h$ are both convex, bounded between $\ubar{u}$ and $\bar{u}$, and have subdifferentials contained in the interior of $S$, contradicting $u$ being an extreme point. Specifically, we construct $h$ using the Bregman divergence of $u$ at the endpoints of $I$ and show that $h$ is not identically zero if and only if $u\vert_{I}$ cannot be written as a piecewise affine function with at most two kinks.\footnote{All the details about the construction of this perturbation can be found in \cref{secap:Bregman_perturb} and its properties are proven in \cref{lem:perturbation}.} Therefore, if $u$ is an extreme point, for any $x\in X$, either $u(x)$ is confounded with $\ubar{u}(x)$ or $\bar{u}(x)$, or there exists a maximal, non-degenerate interval $[\ubar{x}, \bar{x}]$ containing $x$ with the following properties: (i) $u\vert_{[\ubar{x}, \bar{x}]}$ is piecewise affine with at most two kinks; (ii) $u$ is confounded with $\ubar{u}$ or $\bar{u}$ at $\ubar{x}$ and $\bar{x}$; and, (iii) $u$ lies strictly between $\ubar{u}$ and $\bar{u}$ on $(\ubar{x}, \bar{x})$. This implies the existence of a countable collection $\mathcal{X}=\{X_{n}\}_{n\in\mathbb{N}}$ of maximal non-degenerate intervals $X_{n}\subseteq X$ such that, for each $n\in\mathbb{N}$, $\ubar{u}<u\vert_{\Int(X_{n})}<\bar{u}$ and $u\vert_{X_{n}}$ is affine.

The second step of the proof thus addresses the case where $u$ is affine on an interval $\tilde{X} \in \mathcal{X}$, yet none of the conditions \labelcref{cond:tangent_bounds}, \labelcref{cond:matching_grad}, \labelcref{cond:chord_satur}, or \labelcref{cond:boundary_cases} is satisfied. In that case, we show that one can always construct a suitable non-zero, continuous, and piecewise affine perturbation, $h$, such that both $u\pm h\in\mathcal{U}$, again implying that $u$ is not an extreme point.

\section{Optimization on Convex Function Intervals}\label{sec:optimization_CFI}

We start by studying linear programming on CFIs in \cref{sec:lin_prog}. Our second main result (\cref{thm:charac_opt}), stated in \cref{sec:opt_ext_pt}, provides \emph{sufficient} optimality conditions, that we describe intuitively in \cref{sec:interpretation_opt_cond}. In \cref{sec:screening_optimalallocations}, we then use \cref{thm:charac_opt} to identify optimal extreme points in the screening interval. We turn to optimization problems defined directly on the \emph{lower bound} in \cref{sec:lower_bound}. We state our third main result (\cref{thm:opt_lower_bd}) in \cref{sec:theorem_opt_lower_bd}, and illustrate how it can be used to derive optimal menus of outside options in the monopolistic screening problems in \cref{sec:opt_default_screening}.

\subsection{Linear Programming on CFIs}\label{sec:lin_prog}

We consider the class \emph{linear programming} problems over CFIs:
\begin{equation}\label{eqn:Lin_Pb}
    \max_{u\in \mathcal{U}} \int_{X} u \, \diff \mu,\tag{$\mathrm{LP}_{\mathcal{U},\mu}$}
\end{equation}
where $\mu$ is a finite \emph{signed} Radon measure on $(X,\mathcal{B}(X))$.\footnote{By the \emph{Riesz-Markov} representation theorem \citep[][Theorems 19.54 and 19.55]{Yeh2014}, a functional $\ell$ on $(\mathcal{C}(X),\lVert\cdot\rVert_{\infty})$ is continuous and linear if and only if there exists a unique $\mu\in \mathcal{M}(X)$ such that, for every $u\in\mathcal{C}(X)$, $L(u)=\int_{X} u \, \diff \mu$.} We denote the set of all such measures as $\mathcal{M}(X)$ and, for any $\mu\in \mathcal{M}(X)$, denote $\mu^{+}$ and $\mu^{-}$ as the positive and negative parts of $\mu$, respectively. As before, we omit the dependence on the domain when it equals the unit interval, and therefore let $\mathcal{M} \coloneqq \mathcal{M}\bigl([0,1]\bigr)$.

As we shall see in \cref{sec:applications}, in mechanism and information design applications, the signed measure $\mu$ emerges from the problem's primitives. Intuitively, $\mu$ determines the marginal value of raising the function $u$ pointwise in problem \labelcref{eqn:Lin_Pb}. In regions where $\mu$ assigns positive mass, increasing the values of $u$ increases the objective value, and vice-versa.

\subsubsection{Optimal Extreme Points}\label{sec:opt_ext_pt}

Determining which particular extreme point $u^{\star}\in \ex(\mathcal{U})$ is optimal for \labelcref{eqn:Lin_Pb} depends on fine properties of $\mathcal{U}$ and $\mu$. Fix an extreme point $u^{\star} \in \ex(\mathcal{U})$ and let $\mathcal{X}= \{X_{n}\}_{ n \in \mathbb{N} }$ be the corresponding collection of intervals where $u^{\star}$ is affine, as described in \cref{thm:ext_pt}. 

\paragraph{Partitioning the domain}

To verify optimality, we first partition the domain $X$ according to how the candidate extreme point $u^{\star}$ \emph{saturates} the different constraints defining $\mathcal{U}$. Specifically, we let $\mathcal{Y}$ be the coarsest partition of $X$ that is the union of the following collections:
\begin{itemize}
	\item $\mathcal{Y}_{0} \coloneqq \Bigl\{ \{ x \}\subset X \; \big\vert \; \ubar{u}(x)=\bar{u}(x) \Bigr\} $.
	\item $\mathcal{Y}_{1} \coloneqq \Bigl\{ \{ x \}\subset X \setminus \bigcup_{\{x\} \in \mathcal{Y}_0} \{x\} \; \big\vert \; u^{\star}(x)=\bar{u}(x) \Bigl\} \; \cup \; \mathcal{T}$, where $\mathcal{T}$ is the coarsest partition of the set $\bigcup_{n\in\mathbb{N}} \Bigl\{ X_{n} \in \mathcal{X} \; \big\vert \; \text{$u^{\star}\vert_{X_{n}}$ satisfies tangential saturation} \Bigr\}$ such that, for each $Y\in \mathcal{T}$, $u^{\star}\vert_{Y}$ is affine.\footnote{The difference between $\mathcal{T}$ and $\bigcup_{n\in\mathbb{N}} \{ X_{n} \in \mathcal{X} \; \vert \; \text{$u^{\star}\vert_{X_{n}}$ satisfies tangential saturation}\} $ is that if there exists $n,m\in\mathbb{N}$ such that $b_{n}=a_{m}$ and both $u^{\star}\vert_{X_{n}}$ and $u^{\star}\vert_{X_{m}}$ satisfy tangential saturation, then $\mathcal{T}$ contains $X_{n}\cup X_{m}$.}
	\item $\mathcal{Y}_{2} \coloneqq \Bigl\{ \{ x \}\subset X \setminus \bigcup_{\{x\} \in \mathcal{Y}_0} \{x\}  \; \big\vert \; u^{\star}(x)=\ubar{u}(x) \Bigr\}$.
	\item $\mathcal{Y}_{3} \coloneqq \Bigl\{ X_{n} \in \mathcal{X}  \; \big\vert \; \text{$u^{\star}\vert_{X_{n}}$ satisfies upper slope saturation}\Bigr\} $.
	\item $\mathcal{Y}_{4} \coloneqq \Bigl\{ X_{n} \in \mathcal{X}  \; \big\vert \; \text{$u^{\star}\vert_{X_{n}}$ satisfies lower slope saturation}\Bigr\} $.
	\item $\mathcal{Y}_{5} \coloneqq \Bigl\{ X_{n} \in \mathcal{X} \; \big\vert \; \text{$u^{\star}\vert_{X_{n}}$ satisfies chordal saturation} \Bigr\}$.
\end{itemize}

\paragraph{Convex orders of measures}

Our second main result, stated as \cref{thm:charac_opt} below, provides sufficient conditions on $\mu$ that can be used to verify the optimality of $u^{\star}$ for \labelcref{eqn:Lin_Pb}. These conditions require that $\mu$ satisfies certain \emph{convex dominance} conditions on each element of the partition $\mathcal{Y}$.

\begin{definition}
    A measure $\mu\in \mathcal{M}(X)$ is dominated by $\nu\in \mathcal{M}(X)$ in the \emph{convex order}, denoted $\mu \leq_{\mathrm{cx}} \nu$, if
    \begin{equation}\label{eqn:cvx_order}
         \int_{X} u \, \diff \mu \leq \int_{X} u \, \diff \nu
    \end{equation}
    for every $u\in\mathcal{K}(X)$ such that the integrals in inequality \labelcref{eqn:cvx_order} both exist.
    
    The \emph{increasing convex order} $\leq_{\mathrm{icx}}$ and \emph{decreasing convex order} $\leq_{\mathrm{dcx}}$ are defined analogously, with \labelcref{eqn:cvx_order} required to hold for all increasing convex and decreasing convex functions, respectively.
\end{definition}

\paragraph{Verification theorem}

We can now state our sufficient optimality conditions:
\begin{theorem}\label{thm:charac_opt}
Take any CFI $\mathcal{U}$. Let $u^{\star}\in \ex(\mathcal{U})$ and $\mathcal{Y}$ be its corresponding partition, and let $\mu\in \mathcal{M}(X)$. If all the following conditions hold then $u^{\star}$ is optimal for \labelcref{eqn:Lin_Pb}:
\begin{enumerate}[(i)]
	\item For every $Y \in \mathcal{Y}_{1}$, $\mu\vert_{Y}(Y) \geq 0$ and $\mu\vert_{Y} \leq_{\mathrm{cx}}\delta_{Y}$, where $\delta_{Y}$ is a point measure of mass $\mu\vert_{Y}(Y)$ at the unique point where $u^{\star}\vert_{Y}$ touches $\bar{u}$. \label{cond:opti1}
    
	\item For every $Y \in \mathcal{Y}_{2}$, $\mu\vert_{Y}(Y) \leq 0$. \label{cond:opti2}
	
	\item For every $Y \in \mathcal{Y}_{3}$, $\mu\vert_{Y}^{+} \leq_{\mathrm{dcx}}\mu\vert_{Y}^{-}$. \label{cond:opti3}
	
	\item For every $Y \in \mathcal{Y}_{4}$, $\mu\vert_{Y}^{+} \leq_{\mathrm{icx}}\mu\vert_{Y}^{-}$. \label{cond:opti4}
	
	\item For every $Y \in \mathcal{Y}_{5}$, $\mu\vert_{Y} \leq_{\mathrm{cx}}\mathbf{0}$, where $\mathbf{0}$ denotes the null measure. \label{cond:opti5}
\end{enumerate}
\end{theorem}
The proof of \cref{thm:charac_opt} in \cref{secap:proof_charac_opt} is technical. It hinges on formulating the appropriate dual of \labelcref{eqn:Lin_Pb}. We then prove the sufficiency of conditions (i)–(v) by explicitly constructing dual Lagrange multipliers that certify the optimality of a given $u^{\star}\in \ex(\mathcal{U})$.\footnote{For completeness, we also show in \cref{secap:strongduality} that strong duality holds for problems of the form \labelcref{eqn:Lin_Pb}. Strong duality implies that any optimizer $u^{\star}\in \mathcal{U}$ admits dual multipliers that certify its optimality. Our proof of \cref{thm:charac_opt} does not rely on strong duality, since we explicitly construct multipliers such that the dual attains the same value as the primal.}

\subsubsection{Interpretation of the Optimality Conditions in \cref{thm:charac_opt}}\label{sec:interpretation_opt_cond}

The convex order conditions satisfied by the measure $\mu$ in \cref{thm:charac_opt} can be interpreted as $u^{\star}$ satisfying local optimality on each region of the partition $\mathcal{Y}$. These conditions ensure that, for each region $Y \in\mathcal{Y}$, the value of $\int_{Y} u \, \diff \mu$ is maximized by choosing $u=u^{\star}$, thus proving the global optimality of $u^{\star} $ for \labelcref{eqn:Lin_Pb}.

The case where $Y \in \mathcal{Y}_{0}$ is trivial. For $Y \in \mathcal{Y}_{2}$, $u^{\star}$ coincides with $\ubar{u}$ at $Y$. This can only be locally optimal if $\mu\vert_{Y}$ is weakly negative; otherwise, the objective value on $Y$ could be increased by choosing $u\in\mathcal{U}$ such that $u\vert_{Y}> \ubar{u}\vert_{Y}$. Similarly, if $Y \in \mathcal{Y}_{1}$ and is a singleton, then $u^{\star}$ and $\bar{u}$ coincide, which is optimal only if $\mu\vert_{Y}$ is weakly positive. If $Y \in \mathcal{Y}_{1}$ is not a singleton then, for any $u\in\mathcal{U}$,
\begin{equation*}
    \int_{Y} u \, \diff \mu\vert_{Y} \leq \int_{Y} u \, \diff \delta_{Y}\leq \int_{Y} u^{\star} \, \diff \delta_{Y} = \int_{Y} u^{\star} \, \diff  \mu \vert_{Y},
\end{equation*}
where the first inequality follows from $\mu\vert_{Y}\leq_{\mathrm{cx}} \delta_{Y}$, the second from $u\leq u^{\star}$ and $\delta_{Y}$ having positive mass, and the equality from $u^{\star}\vert_{Y}$ being affine. For $Y \in \mathcal{Y}_{3}$, local optimality of $u^{\star}$ on $Y$ follows from a similar argument: for all $u\in \mathcal{U}$,
\begin{equation*}
    \int_{Y} u \, \diff \mu\vert_{Y} = \int_{Y} (u-u^{\star}) \, \diff \mu\vert_{Y} + \int_{Y} u^{\star} \, \diff  \mu\vert_{Y} \leq \int_{Y} u^{\star} \, \diff  \mu\vert_{Y},
\end{equation*}
where the inequality follows from $(u-u^{\star})\vert_{Y}$ being a decreasing (since $u^{\star}\vert_{Y}$ has the largest possible slope) and convex function (since $u$ is convex and $u^{\star}\vert_{Y}$ is affine), and $\mu\vert_{Y}^{+} \leq_{\mathrm{dcx}}\mu\vert_{Y}^{-}$. The argument is symmetric for $Y \in \mathcal{Y}_{4}$. Finally, if $Y \in \mathcal{Y}_{5}$, then no other element of $\mathcal{U}$ can achieve a higher value than $u^{\star}$ because $\mu\vert_{Y} \leq_{\mathrm{cx}}\mathbf{0}$. Specifically, for any $u\in\mathcal{U}$,
\begin{equation*}
    \int_{Y} u \, \diff \mu\vert_{Y} \leq \int_{Y} u \, \diff \mathbf{0} = 0 = \int_{Y} u^{\star} \, \diff \mu\vert_{Y},
\end{equation*}
where the first inequality follows from $\mu\vert_{Y} \leq_{\mathrm{cx}}\mathbf{0}$, and the last equality follows from $u^{\star}\vert_{Y}$ being affine.

\subsubsection{Optimal Extreme Points in the Screening Interval} \label{sec:screening_optimalallocations}

We return to the running example from \cref{sec:screening_run_example_ext_pt}. We show that well-known screening problems naturally admit the representation
\begin{equation}\label{eqn:LP_screening}
    \max_{u \in \mathcal{U}_{\mathrm{S}}} \int_\Theta u \, \diff \mu, \tag{$\mathrm{LP}_{\mathcal{U}_{\mathrm{S}},\mu}$}
\end{equation}
where the measure $\mu\in \mathcal{M}(\Theta)$ emerges from the economic primitives, and we derive optimal screening mechanisms under specific assumptions on the environment.

\paragraph{Revenue-maximizing mechanisms}

Suppose the principal chooses a mechanism $(x,t)$ satisfying \labelcref{eqn:IC_screening,eqn:IR_screening} to maximize 
\begin{equation}\label{eqn:exp_revenue}
    \int_\Theta t(\theta) \, \diff  F(\theta).
\end{equation}

Let $u\in\mathcal{U}_{\mathrm{S}}$. Following \cref{lem:screening_interval}, $u'(\theta) = x(\theta)$ and $t(\theta)=\theta u'(\theta) - u(\theta)$ almost everywhere on $\Theta$. We thus obtain:
\begin{align*}
    \int_{0}^{1} t(\theta) \, \diff F(\theta) &= \int_{0}^{1} \bigl(\theta u'(\theta) -u(\theta) \bigr) f(\theta) \, \diff \theta \\[5pt]
    &= \bigl[\theta f(\theta) u(\theta)\bigr]_{0}^{1} - \int_{0}^{1} \biggl\{f(\theta) + \frac{\diff}{\diff\theta}\bigl[\theta f(\theta)\bigr] \biggr\}u(\theta)\, \diff \theta 
\end{align*}
where the second equality follows from integration by parts. Hence, we have the following equivalence:
\begin{lemma}\label{lem:rev_max_pb}
Maximizing \labelcref{eqn:exp_revenue} subject to \labelcref{eqn:IC_screening,eqn:IR_screening} can be written as
    \begin{equation}\label{eqn:rev_max_pb}
    \max_{u \in \mathcal{U}_{\mathrm{S}}} \int_\Theta u \, \diff \mu_{\mathrm{R}}, \tag{RevMax}
\end{equation}
where, for any Borel measurable $A\subseteq\Theta$,
\begin{equation}\label{eqn:measure_rev_max}
    \mu_{\mathrm{R}}(A) = \int_{A} \psi_{\mathrm{R}} \, \diff \nu,
\end{equation}
with $\nu\coloneqq\delta_{1} -\delta_{0} + \lambda$ and
\begin{equation}\label{eqn:rev_max_density}
    \psi_{\mathrm{R}}(\theta)=\left\{\begin{array}{ll}
        - \biggl\{f(\theta) + \frac{\diff}{\diff\theta}\bigl[\theta f(\theta)\bigr] \biggr\} & \text{if $\theta\in(0,1)$,} \\[3pt]
        \theta f(\theta) & \text{if $\theta\in \{0,1\}$.}
    \end{array}
    \right.
\end{equation}
\end{lemma}

The measure $\mu_{\mathrm{R}}$ is determined by the principal's preferences, the agents' incentive constraints, and the type distribution $f$. It captures the incremental variation in the principal's revenue resulting from marginal changes in agents' \emph{information rents}.

\begin{remark}
    \cite{Daskalakis2017} study a multidimensional version of \labelcref{eqn:rev_max_pb}. The measure $\mu_{\mathrm{R}}$ defined through \labelcref{eqn:measure_rev_max,eqn:rev_max_density} corresponds precisely to the one-dimensional version of their \emph{transformed measure} (see their Definition 3). However, \cite{Daskalakis2017} do not consider type-dependent participation constraints, which are explicitly incorporated in our framework through the lower level boundary $u_{0}$ of $\mathcal{U}_{\mathrm{S}}$.
\end{remark}

We now apply \cref{thm:charac_opt} to derive revenue-maximizing mechanisms under type-dependent participation constraints. For each $\theta^{\star}\in \Theta$, define
\begin{equation}\label{eqn:cutoff_utility}
    u_{\theta^{\star}}(\theta) = \left\{
    \begin{array}{ll}
        u_{0}(\theta) & \text{if $\theta\in [0,\theta^{\star})$,} \\[5pt]
        u_{0}(\theta) + (\theta - \theta^{\star}) & \text{if $\theta\in [\theta^{\star},1]$.}
    \end{array}
    \right.
\end{equation}

By \cref{thm:ext_pt}, $u_{\theta^{\star}}$ is an extreme point of $\mathcal{U}_{\mathrm{S}}$ for any $\theta^{\star}\in \Theta$, as it exhibits lower-level saturation on $[0,\theta^{\star})$ and upper-slope saturation on $[\theta^{\star},1]$. Furthermore, \cref{lem:screening_interval} implies that $u_{\theta^{\star}}$ is implemented by a \emph{cutoff} allocation rule, defined by
\begin{equation*}
    x_{\theta^{\star}}(\theta) = x_{0}(\theta) \Ind_{\theta < \theta^{\star}} + \Ind_{\theta \geq \theta^{\star}},
\end{equation*}
for every $\theta\in \Theta$.

In the standard case where $x_{0}(\theta)=0$ for all $\theta \in \Theta$, cutoff rules are known to be optimal \emph{regardless} of the type distribution \citep{Myerson1981,Riley1983}.\footnote{See also \cite{Skreta2006} and \cite{Borgers2015} Chapter 1.} Our next result (\Cref{prop:opt_rev_cutoff}) shows that if the type distribution is regular, cutoff mechanisms continue to maximize revenue even when agents have access to rich menus of outside options. However, we also construct a counterexample (\cref{example:bunch_screen} below) showing that the optimality of cutoff mechanisms is not robust once type-dependent participation constraints are introduced: the structure of revenue-maximizing mechanisms depends crucially on the shape of the \emph{type distribution}.

\begin{proposition}\label{prop:opt_rev_cutoff}
    Suppose that $F$ is \emph{Myerson-regular}. That is, the \emph{virtual value function}, defined by
    \begin{equation*}
        v(\theta) = \theta - \frac{1-F(\theta)}{f(\theta)},
    \end{equation*}
    for all $\theta\in \Theta$, is non-decreasing. Then $u_{\theta^{\star}}$ solves \labelcref{eqn:rev_max_pb}, where $\theta^\star \in (0,1)$ is the unique type such that $v(\theta^\star) = 0$.
\end{proposition}

The proof of \cref{prop:opt_rev_cutoff} is given in \cref{secap:opt_screening_proofs}. 
Here, we sketch the argument to illustrate how the convex dominance conditions from \cref{thm:charac_opt} can be checked in a concrete economic setting, 
assuming that the type density $f$ is log-concave.\footnote{Log-concavity of $f$ is sufficient for $F$ to be Myerson-regular. 
A weaker sufficient condition is that $F$ has a monotone hazard rate.}
 
\begin{figure}[h!]
\begin{subfigure}[b]{0.496\linewidth}
    \centering
    \includegraphics{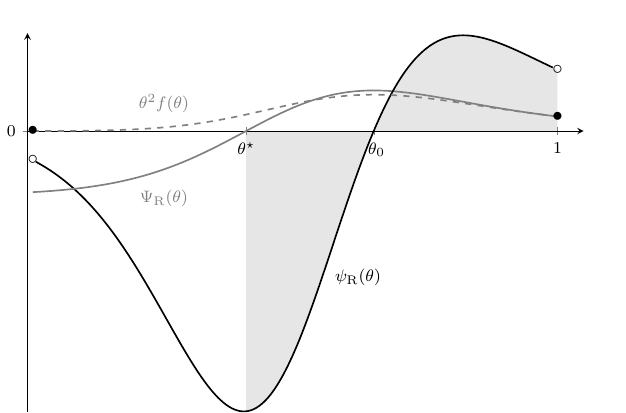}
    \subcaption{Measure $\mu_{\mathrm{R}}$ for log-concave $f$.}
    \label{fig:meas_log_conc_f}
\end{subfigure}
\begin{subfigure}[b]{0.496\linewidth}
    \centering
    \includegraphics{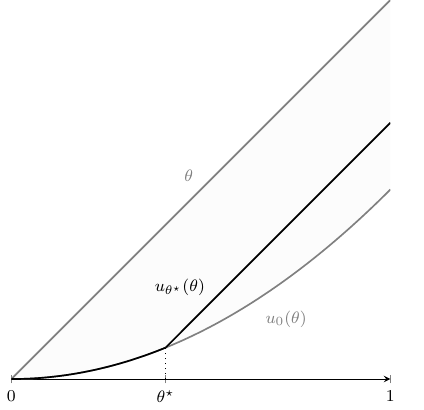}
    \caption{Optimal extreme point $u_{\theta^{\star}}$ for $\mu_{\mathrm{R}}$.}
    \label{fig:opt_extreme_point_screening}
\end{subfigure}
\caption{A revenue-maximizing extreme point for a log-concave type density. In this picture, we let $F \sim \mathrm{Logistic}(\mu,\sigma)$ truncated on $[0,1]$ with $\mu = 0.5$ and $\sigma=0.15$.}
\label{fig:rev_max_cutoff}
\end{figure}

For each $\theta\in \Theta$, define
\begin{align*}
    \Psi_{\mathrm{R}}(\theta) & = \mu_{\mathrm{R}}\bigl([\theta,1]\bigr) \\[5pt]
    &=\psi_{\mathrm{R}}(1) + \int_{\theta}^{1}\psi_{\mathrm{R}}(t) \, \diff t \\[5pt]
    &=\theta f(\theta) - \bigl(1-F(\theta)\bigr) \\[5pt]
    &=f(\theta) v(\theta).
\end{align*}

Under log-concavity of $f$, both $\psi_{\mathrm{R}}$ and $\Psi_{\mathrm{R}}$ are \emph{single-crossing from below} on $\Theta$ (as depicted in \cref{fig:meas_log_conc_f}), and $0<\theta^{\star}<\theta_{0}\leq 1$ where $\psi_{\mathrm{R}}(\theta_{0})=0$ and $\Psi_{\mathrm{R}}(\theta^{\star})=0$.

Consider the partition $\mathcal{Y}$ induced by $u_{\theta^{\star}}$. Inspection of \cref{fig:opt_extreme_point_screening} reveals that $\mathcal{Y}_{0}=\{0\}$, $\cup_{Y\in \mathcal{Y}_{2}} Y=(0,\theta^{\star})$, and $\mathcal{Y}_{3}=[\theta^{\star},1]$, while $\mathcal{Y}_{i}=\varnothing$ for all $i\in\{1,4,5\}$. By \cref{thm:charac_opt}, it thus suffices to show that $\psi_{\mathrm{R}}(\theta)\leq 0$ for all $\theta \in (0,\theta^{\star})$ and that $\mu_{\mathrm{R}}\vert_{[\theta^{\star},1]}^{+} \leq_{\mathrm{dcx}}\mu_{\mathrm{R}}\vert_{[\theta^{\star},1]}^{-}$ to verify the optimality of $u_{\theta^{\star}}$.

Since $\theta^{\star}<\theta_{0}$, it is immediate that $\psi_{\mathrm{R}}(\theta)\leq 0$ for all $\theta \in (0,\theta^{\star})$. Furthermore, we have $\mu_{\mathrm{R}}^{+}\bigl([\theta^{\star},1]\bigr)=\mu_{\mathrm{R}}^{-}\bigl([\theta^{\star},1]\bigr)$ (illustrated by the shaded region in \cref{fig:meas_log_conc_f}). Hence, $\mu_{\mathrm{R}}\vert_{[\theta^{\star},1]}^{+}$ and $\mu_{\mathrm{R}}\vert_{[\theta^{\star},1]}^{-}$ have an equal (non-zero) mass, and can thus be normalized to probability measures. 

By integration by parts  
\begin{equation}
\label{eqn:barycenter_dcx_screen}
    \begin{aligned}
        \int_{[\theta,1]} t \, \diff\mu_{\mathrm{R}}(t)
        &=\psi_{\mathrm{R}}(1) + \int_{\theta}^{1} t \psi_{\mathrm{R}}(t) \, \diff t \\[5pt]
        &= {\theta}^{2} f(\theta) \geq 0,
    \end{aligned}
\end{equation}
for every $\theta\in \Theta$. Note that \labelcref{eqn:barycenter_dcx_screen} implies
\begin{equation*}
    \int_{[\theta^{\star},1]} \theta \, \diff\mu_{\mathrm{R}}\vert^{+}_{[\theta^{\star},1]}(\theta) 
    \;\geq\; \int_{[\theta^{\star},1]} \theta \, \diff\mu_{\mathrm{R}}\vert^{-}_{[\theta^{\star},1]}(\theta).
\end{equation*}

Therefore, by Theorem 4.A.2 in \citet{Shaked2007}, $\mu_{\mathrm{R}}\vert_{[\theta^{\star},1]}^{+} \leq_{\mathrm{dcx}}\mu_{\mathrm{R}}\vert_{[\theta^{\star},1]}^{-}$  
is equivalent to the following \emph{weak majorization} condition:
\begin{align*}
    \forall \theta \in [\theta^{\star},1], &\quad 
    \int_{\theta}^{1} \max\{0,\Psi_{\mathrm{R}}(t)\} \, \diff t
    \;\geq\; \int_{\theta}^{1} \max\{0,-\Psi_{\mathrm{R}}(t)\} \, \diff t , \\[5pt]
    \iff \forall \theta \in [\theta^{\star},1], &\quad 
    \int_{\theta}^{1} \Psi_{\mathrm{R}}(t) \, \diff t \;\geq\; 0,
\end{align*}
which holds because $\Psi_{\mathrm{R}}(\theta) \geq 0$ for all $\theta \in [\theta^{\star},1]$.

\begin{example}\label{example:bunch_screen}
    In \cref{fig:rev_max_cutoff} below, we present a counterexample showing that revenue maximizing mechanisms do not necessarily take the form of simple cutoff rules when there are type-dependent participation constraints and the type distribution is irregular.
    \begin{figure}[h!]
        \begin{subfigure}{0.496\linewidth}
        \centering
        \includegraphics{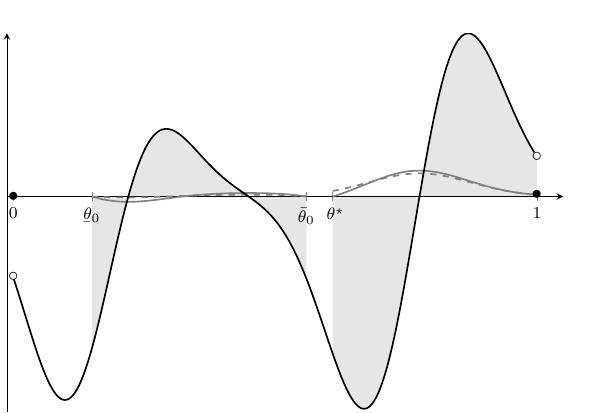}
        \caption{Signed density $\psi_{\mathrm{R}}$.}
        \label{fig:meas_bunch_screen}
        \end{subfigure}
        \begin{subfigure}[b]{0.496\linewidth}
        \centering
        \includegraphics{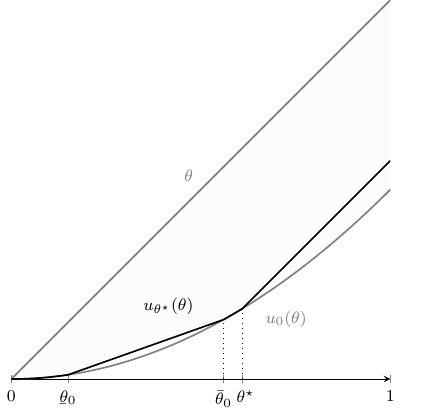}
        \caption{Optimal extreme point $u_{\theta^{\star}}$ for $\mu_{\mathrm{R}}$.}
        \label{fig:opt_extreme_point_bunch_screen}
        \end{subfigure}
    \caption{A revenue-maximizing extreme point for an irregular type distribution. For each $\theta\in \Theta$, we let $F(\theta)=\frac{G(\theta)-G(0)}{G(1)-G(0)}$ where $G(\theta)=\alpha \Phi\bigl(\frac{\theta-m_0}{\sigma_0}\bigr)+(1-\alpha) \Phi\bigl(\frac{\theta-m_1}{\sigma_1}\bigr)$ and $\Phi$ is the standard Gaussian cumulative distribution function. We fix $\alpha=0.7$, $m_{0}\approx 0.128$, $m_{1}=0.75$ and $\sigma_0=\sigma_1=0.1$.}
    \label{fig:rev_max_cutoff}
    \end{figure}

    We take a type distribution $F$ that is a mixture of two Gaussians truncated on $[0,1]$. Its density $f$ is bimodal and $F$ is not Myerson-regular. As a result, the signed density $\psi_{\mathrm{R}}$ changes sign four times.\footnote{The precise computations are available upon request.} It is initially negative and ultimately positive (see \cref{fig:meas_bunch_screen}). Hence, the principal would like to minimize rents for low and intermediate types, and provide rents to intermediate and high types.
    
    Let us first focus on the interval $[\ubar{\theta}_{0},\bar{\theta}_{0}]$, which contains the intermediate interval where $\psi_{\mathrm{R}}\geq 0$. In \cref{fig:meas_bunch_screen}, the solid gray curve tracks the cumulative mass of $\mu_{\mathrm{R}}$ starting from $\ubar{\theta}_{0}$, while the dashed curve tracks its barycenter on $[\ubar{\theta}_{0},\theta]$ for every $\theta\leq\bar{\theta}_{0}$. Since both equal zero at $\theta=\bar{\theta}_{0}$, the total mass and the barycenter of $\mu_{\mathrm{R}}$ both vanish on $[\ubar{\theta}_{0},\bar{\theta}_{0}]$. Hence, by Theorem 3.A.44 (condition 3.A.57) in \cite{Shaked2007}, the positive part of $\mu_{\mathrm{R}}$ on $[\ubar{\theta}_{0},\bar{\theta}_{0}]$ is dominated in convex order by its negative part, so $\mu_{\mathrm{R}}\vert_{[\ubar{\theta}_{0},\bar{\theta}_{0}]}\leq_{\mathrm{cx}}\mathbf{0}$. A similar argument to the proof of \cref{prop:opt_rev_cutoff} implies $\mu_{\mathrm{R}}\vert_{[\theta^{\star},1]}^{+}\leq_{\mathrm{dcx}}\mu_{\mathrm{R}}\vert_{[\theta^{\star},1]}^{-}$. Finally, on the complementary intervals, we have $\psi_{\mathrm{R}}\leq 0$.
    
    Applying our \cref{thm:charac_opt} thus shows that the extreme point in \cref{fig:opt_extreme_point_bunch_screen} is a revenue-maximizing indirect utility. By \cref{cor:imp_ext_pt_screening}, the corresponding mechanism features \emph{bunching} and a \emph{non-convex} exclusion set: an interval of intermediate types is bunched, while both lower and higher types are assigned their favorite outside options in $M_{0}$ (see \cref{fig:opt_extreme_point_bunch_screen}).\footnote{It is well known that with multidimensional heterogeneity, revenue-maximizing mechanisms can exhibit intricate bunching patterns and depend subtly on the distribution of types \citep[e.g.,][]{Manelli2007,Daskalakis2017,Lahr2024}. \cref{example:bunch_screen} shows that, even in one dimension, type-dependent participation constraints alone can generate comparable complexity.} It is noteworthy that bunching is also caused by the richness of $M_{0}$---equivalently, the fact that $u_{0}$ has sufficiently many kink points. For instance, if $u_{0}$ were affine on the interval where the default allocation is ironed out (i.e., $[\ubar{\theta}_{0},\bar{\theta}_{0}]$), the optimal mechanism would merely coincide with the default one on that interval. In the limiting case where $u_{0}$ is affine on all $\Theta$ \citep[as in, e.g.,][]{Rochet1998}---so that $M_{0}$ offers a single outside option preferred by all types to $(0,0)$---\cref{cor:imp_ext_pt_screening} implies that the non-participation region must be an interval $[0,\theta^{\star}]$.
\end{example}

\paragraph{Welfare-maximizing mechanisms}

A designer chooses a mechanism $(x,t)$ subject to \labelcref{eqn:IC_screening,eqn:IR_screening} so as to maximize a combination of weighted-utilitarian welfare and profit:
\begin{equation}\label{eqn:welfare_function}
    \int_\Theta \pi(\theta)\bigl(\theta x(\theta) -t(\theta)\bigr) \, \diff  F(\theta) + \alpha \int_\Theta \bigl(t(\theta)-c x(\theta)\bigr) \, \diff  F(\theta),
\end{equation}
where $\pi(\theta)\geq 0$ is the Pareto weight assigned to the welfare of type-$\theta$ agents, $c\in \mathbb{R}_{+}$ is the marginal cost of providing the good (or the principal's value from retaining it), and $\alpha\in \mathbb{R}_{++}$ is the weight placed on profit relative to agents' welfare.\footnote{This weight can be interpreted either as a preference parameter reflecting the principal's opportunity cost of funds (e.g., when revenues are used for purposes outside the mechanism), or as the Lagrange multiplier on a budget-balance constraint.} Such objectives are central to the \emph{redistributive} market design literature, where $\theta\in \Theta \mapsto \pi(\theta)$ is usually assumed non-increasing.\footnote{For recent work, see \cite{Condorelli2013,Dworczak2021,Akbarpour2024a,Akbarpour2024b,Kang2023,Kang2024a,Kang2024b,Pai2025}.} 

Using similar arguments as for \cref{lem:rev_max_pb}, we obtain:
\begin{lemma}
Maximizing \labelcref{eqn:welfare_function} subject to \labelcref{eqn:IC_screening,eqn:IR_screening} can be written as
\begin{equation}\label{eqn:welf_max_pb}
    \max_{u \in \mathcal{U}_{\mathrm{S}}} \int_\Theta u \, \diff \mu_{\mathrm{W}}, \tag{WelMax}
\end{equation}
where, for any Borel measurable $A\subseteq\Theta$,
\begin{equation}\label{eqn:measure_welf_max}
    \mu_{\mathrm{W}}(A) = \int_{A} \psi_{\mathrm{W}} \, \diff \nu,
\end{equation}
with $\nu\coloneqq\delta_{1}-\delta_{0}+\lambda$ and
\begin{equation}\label{eqn:welf_max_density}
    \psi_{\mathrm{W}}(\theta)=\left\{\begin{array}{ll}
        \alpha (\theta-c) f(\theta) & \text{if $\theta\in \{0,1\}$}\\[5pt]
        \pi(\theta)f(\theta) - \alpha \Bigl\{f(\theta) + \frac{\diff}{\diff\theta}\bigl[(\theta-c) f(\theta)\bigr]\Bigr\} & \text{if $\theta\in(0,1)$} 
    \end{array}
    \right.
\end{equation}
\end{lemma}
Again, the measure $\mu_{\mathrm{W}}$ is determined by the economic primitives. In this case, it captures the incremental variation in \emph{total welfare} resulting from marginal changes in agents' information rents.

We translate \cref{thm:charac_opt} into general sufficient optimality conditions for problems of the form \labelcref{eqn:welf_max_pb} that only bear on the density function $\psi_{\mathrm{W}}$.\footnote{Note that those conditions also apply to the problem \labelcref{eqn:rev_max_pb} since it corresponds to the particular version of \labelcref{eqn:welf_max_pb} where $\pi(\theta)=0$ for every $\theta \in \Theta$, $\alpha=1$ and $c=0$.} To do so, consider $u^{\star}\in\ex(\mathcal{U}_{\mathrm{S}})$ with its corresponding partition $\mathcal{Y}$ as defined in \cref{sec:opt_ext_pt}. Thanks to \cref{cor:imp_ext_pt_screening} we can deduce the form for $\mathcal{Y}$. Since $u_0(0)=0$, we have $\mathcal{Y}_{0}=\bigl\{\{0\}\bigr\}$. Combined with $\ubar{s}=0$, this implies no affine pieces satisfy lower slope saturation, so $\mathcal{Y}_{4}=\varnothing$. Moreover, since $\bar{u}$ is affine with slope $\bar{s}=1$, the function $u^{\star}$ does not satisfy tangential saturation on $\Theta$, meaning $\mathcal{Y}_{1}=\varnothing$. Recall also from \cref{cor:imp_ext_pt_screening} that $\mathcal{Y}_{3}$ only contains a unique interval of the form $[\theta^{\star},1]$, where $\theta^{\star}$ is the cutoff type above which the allocation rule is deterministic. Therefore, $\mathcal{Y}=\{0\} \cup \mathcal{Y}_{2} \cup [\theta^{\star},1] \cup \mathcal{Y}_{5}$, where $\mathcal{Y}_{2}$ is the collection of all singletons of lower-level saturation and $\mathcal{Y}_{5}$ is the collection of all the chordal saturation intervals (i.e., the intervals where the default allocation is ironed out).

\begin{corollary}\label{cor:welf_opt_allocations}
    Let $u^{\star}\in\ex(\mathcal{U}_{\mathrm{S}})$. If all the following conditions are satisfied, then $u^{\star}$ is optimal for \labelcref{eqn:welf_max_pb}:
    \begin{enumerate}[(i)]
        \item For all $\{\theta\} \in \mathcal{Y}_{2}$, $\psi_{\mathrm{W}}(\theta) \leq 0$.
        \item  $ \psi_{\mathrm{W}}(1)+\int_{\theta^{\star}}^1 \psi_{\mathrm{W}}(\theta) \, \diff \theta = 0$, $\psi_{\mathrm{W}}(1)+ \int_{\theta^{\star}}^{1} \theta \psi_{\mathrm{W}}(\theta) \, \diff \theta \geq 0$, and 
        \begin{equation*}
            \forall \theta \in [\theta^{\star},1], \; \int_{\theta^{\star}}^{\theta}(\theta - t) \psi_{\mathrm{W}}(t) \, \diff t \leq 0.
        \end{equation*}
        \item For all $Y \in \mathcal{Y}_5$ strictly in the interior of $\Theta$, $\int_{Y} \psi_{\mathrm{W}}(\theta) \, \diff \theta = \int_{Y} \theta \psi_{\mathrm{W}}(\theta) \, \diff \theta =0$. Furthermore, 
        \begin{equation*}
            \forall \theta \in Y, \; \int_{\inf Y}^{\theta}(\theta - t) \psi_{\mathrm{W}}(t) \, \diff t \leq 0,
        \end{equation*}
        with equality at $\theta=\sup Y$.
    \end{enumerate}
\end{corollary}
\begin{proof}
    \cref{cor:welf_opt_allocations} follows from applying Theorems 3.A.1 and 4.A.2 from \cite{Shaked2007} to the (scaled versions of the) positive and negative parts of the conditional measures from \cref{thm:charac_opt}.
\end{proof}

\begin{remark}\label{remark:link_DM2024}
    \cite{Dworczak2024} approach a problem similar to \labelcref{eqn:welf_max_pb} (with the only difference being flexibility in $u(0)$) by a concavification approach.\footnote{See also \cite{Kleiner2021} Proposition 2.} This approach can be used for problems like \labelcref{eqn:welf_max_pb} to obtain a partition $\mathcal{Y}$ of $X$ that satisfies the optimality conditions of \cref{thm:charac_opt}. We make this connection precise in \cref{secap:concavification}.
\end{remark}

\subsection{Optimal Lower Level Boundary} \label{sec:lower_bound}

\subsubsection{Theory}\label{sec:theorem_opt_lower_bd}

Let $\mathcal{U}$ be a CFI with upper level boundary $\bar{u}$ and slope set $S=[\ubar{s},\bar{s}]$. The set of lower boundaries compatible with $\mathcal{U}$ is given by
\begin{equation*}
    \mathcal{K}_{\bar{u},S}\coloneqq \Bigl\{u\in \mathcal{K} \; \big\vert \; u\leq \bar{u}, \, \partial u(X) \subseteq S \Bigr\}.
\end{equation*}

We consider the problem
\begin{align}\label{eqn:optim_low_bd}
    \max_{\ubar{u}\in \mathcal{K}_{\bar{u},S}} \; &\int_{X} u_{\ubar{u}} \, \mathrm{d}\nu \tag{OptLB}\\
    \text{s.t.} \;  & \; u_{\ubar{u}} \in \underset{u \in \mathcal{U}_{\ubar{u}}}{\argmax} \; \int_{X} u \, \diff  \mu, \notag
\end{align}
where $\mu$ and $\nu$ are two finite signed Radon measures, and $\mathcal{U}_{\ubar{u}}$ is the CFI defined by
\begin{equation*}
    \mathcal{U}_{\ubar{u}}=\mathcal{K}_{\bar{u},S}\cap \bigl\{u\in\mathcal{K} \; \vert \; \ubar{u} \leq u \bigr\}.
\end{equation*}

Let $V\colon\mathcal{K}_{\bar{u}, S} \to \mathbb{R}$ be the functional defined by $V(\ubar{u}) = \int_{X} u_{\ubar{u}} \, \mathrm{d}\nu$.

The following theorem shows that the problem of designing an optimal lower level boundary of a CFI is well behaved in two useful cases: when the measures $\mu$ and $\nu$ coincide and when we restrict attention to lower boundaries that lead to affinely bounded CFIs. The CFI $\mathcal{U}_{\ubar{u}}$ is \emph{$\bar{s}$-affinely bounded} (or \emph{upper affinely bounded}) if $\bar{u}'=\bar{s}$ and $\ubar{u} \in \mathcal{K}_{\bar{u}, S}^0 \coloneqq \big\{u \in \mathcal{K}_{\bar{u}, S} \; \vert \; u(0)=\bar{u}(0) \bigr\}$. Likewise, we say that $\mathcal{U}_{\ubar{u}}$ is \emph{$\ubar{s}$-affinely bounded} (or \emph{lower affinely bounded}) if $\bar{u}'=\ubar{s}$ and $\ubar{u} \in \mathcal{K}_{\bar{u}, S}^1 \coloneqq \bigl\{u \in \mathcal{K}_{\bar{u}, S} \; \vert \; u(1)=\bar{u}(1) \bigr\}$.

\begin{theorem}[Optimal lower level boundary]\label{thm:opt_lower_bd}
    Assume that $\mu$ admits a density on $(0,1)$. The following claims are true:
    \begin{enumerate}
        \item If $\bar{u}'=\bar{s}$, the mapping
        \begin{equation*}
            \ubar{u} \in \mathcal{K}_{\bar{u}, S}^0 \mapsto u_{\ubar{u}}
        \end{equation*}
        is linear. In particular, the functional $V\vert_{\mathcal{K}_{\bar{u}, S}^0}$ is linear. 
        \label{cond:opt_lb_afffinebound_0}
        
        \item Symmetrically, if $\bar{u}'=\ubar{s}$, the mapping
        \begin{equation*}
            \ubar{u} \in \mathcal{K}_{\bar{u}, S}^1 \mapsto u_{\ubar{u}}
        \end{equation*}
        is linear. In particular, the functional $V \vert_{\mathcal{K}_{\bar{u}, S}^1}$ is linear.
        \label{cond:opt_lb_afffinebound_1}
        
        \item If $\nu=\mu$, then $V$ is concave. \label{cond:opt_lb_concave}
    \end{enumerate}
\end{theorem}

The formal proof can be found in \cref{proof:thm:opt_lower_bd}. The argument for Condition \labelcref{cond:opt_lb_concave} is straightforward: for every $\ubar{u} \in \mathcal{K}_{ \ubar{s}, \bar{s} }$, there is an extreme point $u^{\star} \in \ex\bigl(\mathcal{U}(\ubar{u}, \bar{u})\bigl)$ which attains $\max_{ u \in \mathcal{U}(\ubar{u}, \bar{u}) } \int_{X} u \, \diff  \mu$. However, a convex combination of two extreme points is in general not an extreme point. This implies that for convex combinations of lower bounds, the convex combination of the respectively optimal extreme points may not achieve the maximum. The intuition behind the connection between affine boundedness of the CFI and of the functional (Conditions \labelcref{cond:opt_lb_afffinebound_0,cond:opt_lb_afffinebound_1}) is somewhat more subtle. Roughly, for affinely bounded CFIs, $\mathcal{Y}_{1}$ is empty for any extreme point. This means that for any $\ubar{u}$ the respective partition $\mathcal{Y}$ of the optimal solution is going to be the same. As a consequence, the optimal solution depends on $\ubar{u}$ in a linear way: the partition $\mathcal{Y}$ only determines the \emph{ironing intervals}\footnote{That is, the intervals where an extreme point satisfies chordal saturation by connecting two points on the lower bound in an affine way.} and the point at which the upper slope constraints starts to be binding. Thus, $\ubar{u} \mapsto u_{\ubar{u}}$ is linear in $\ubar{u}$. In particular, for an affinely bounded CFI, maximization over the lower bound is again going to be a linear problem on a CFI. For the case where affine boundedness fixes the value of $\ubar{u}(0)$ and $\bar{u}'=\bar{s}$, this CFI is given by $\bigl\{u \in \mathcal{K}_{ \ubar{s}, \bar{s} } \; \vert \; \forall x \in X, \, u(x) \geq \bar{u}(0) \bigr\}$. \cite{Dworczak2024} make use of this fact to derive an optimal menu of outside options in a screening problem.

\subsubsection{Optimal Outside Options in the Screening Interval}\label{sec:opt_default_screening}

We now consider the optimal design of the default menu $M_{0}$ itself. We assume that, given any menu, a principal chooses an \labelcref{eqn:IC_screening,eqn:IR_screening} mechanism $(x,t)$ to maximize expected revenue. Additionally, a benevolent planner chooses $M_{0}$ to maximize a generalized welfare objective of the form \labelcref{eqn:welfare_function}, internalizing how the default mechanism affects the choice of the principal's mechanism downstream.

We can write the planner's problem as
\begin{align}\label{eqn:optim_low_bd_screening}
    \max_{u_{0} \in \mathcal{K}_{\mathrm{S}}} \; &\int_{\Theta} u_{u_{0}} \, \mathrm{d}\mu_{\mathrm{W}} \tag{OptLBScreen}\\
    \text{s.t.} \;  & \; u_{u_{0}} \in \underset{u \in \mathcal{U}_{\mathrm{S}}^{u_{0}}}{\argmax} \; \int_{\Theta} u \, \diff  \mu_{\mathrm{R}}, \notag
\end{align}
where $\mu_{\mathrm{W}}$ and $\mu_{\mathrm{R}}$ are defined according to \labelcref{eqn:measure_welf_max,eqn:measure_rev_max}, respectively, and $\mathcal{K}_{\mathrm{S}}\coloneqq \mathcal{K}^{0}_{\Id ,[0,1]}$ and $\mathcal{U}_{\mathrm{S}}^{u_{0}} \coloneqq \mathcal{K}_{\mathrm{S}} \cap \{u\in\mathcal{K} \; \vert \; u_{0} \leq u \}$.

Since $\mathcal{U}_{\mathrm{S}}$ is affinely bounded, this problem is linear in $\ubar{u}$, and we can obtain the following result as a corollary to \cref{thm:opt_lower_bd}.

\begin{corollary}\label{cor:screening_opt_lowerb}
    There always exists an optimal default menu $M^{\star}_{0}=\bigl\{(0,0),(1,p)\bigr\}$ with $p\in\mathbb{R}_{+}$.
\end{corollary}

\cref{cor:screening_opt_lowerb} states that the optimal menu of outside options induces an outside indirect utility $u_{0}$ that is an increasing convex and piecewise affine function with a single kink. The intuition behind this is simple: an optimal menu of outside options is an extreme point of $\mathcal{K}_{\mathrm{S}}$ by linearity. Following \cref{thm:ext_pt}, extreme points of $\mathcal{K}_{\mathrm{S}}$ are exactly those functions that coincide with the constant function equal to $\ubar{\theta}$ for some interval $[\ubar{\theta}, \theta^{\star}]$ and then have slope equal to $1$.\footnote{Those correspond to jump functions in the allocation space.}

Economically, the designer's optimal menu takes the form of an \emph{``option-to-own''} the good at a fixed price, corroborating the main result of \citet{Dworczak2024}. As illustrated in \cref{fig:opt_default_menu}, it follows from \cref{thm:ext_pt} that the menu chosen by the designer induces the monopolist either to offer the good at a posted price below the menu price (\cref{fig:posted-price}), or to introduce a lottery at an even lower price (\cref{fig:lottery}). Which of these outcomes arises depends on the specific assumptions on $\pi$, $\alpha$, $c$, and $f$.
\begin{figure}[h!]
    \begin{subfigure}[b]{0.496\linewidth}
        \centering
        \includegraphics{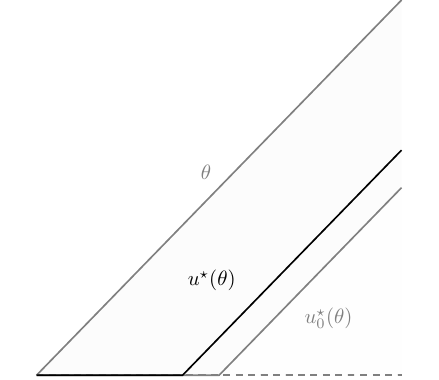}
        \subcaption{Posted price.}
        \label{fig:posted-price}
    \end{subfigure}
    \begin{subfigure}[b]{0.496\linewidth}
        \centering
        \includegraphics{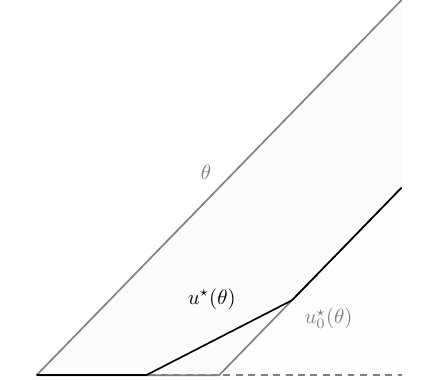}
        \subcaption{Additional lottery.}
        \label{fig:lottery}
    \end{subfigure}
\caption{The designer's optimal menu induces a lower level boundary $u_{0}^{\star}$ with at most one kink. The only possible extreme points (excluding $\Id_{\Theta}$) in the corresponding CFI are depicted in black.}
\label{fig:opt_default_menu}
\end{figure}

\section{Applications}\label{sec:applications}

We begin by showing in \cref{sec:majorization} that we can recover the characterization of extreme points of \citet{Kleiner2021} using our \cref{thm:ext_pt} and extend it to two-sided (weak) majorization constraints. We then apply our results to derive economic implications in three settings: the delegation problem with type-dependent participation constraints (\cref{sec:delegation}), large contest design with allocative constraints (\cref{sec:contest}), and mean-based Bayesian persuasion with informativeness constraints (\cref{sec:persuasion}).

\subsection{Majorization Intervals} \label{sec:majorization}

\subsubsection{Definitions}

We begin by recalling some majorization theoretic concepts from \cite{Kleiner2021}.\footnote{The mathematical theory of majorization originated with the foundational work of \cite{Hardy1929,Hardy1934}. For a comprehensive review of this literature, we refer to \cite{Marshall2011}.} Consider the space $\mathcal{F}=\bigl\{f\in L^1(X) \; \vert \; \text{$f$ non-decreasing}\bigr\}$. For any $f\in \mathcal{F}$, we denote the \emph{mean value} of $f$ on $X$ as $m_{f}\coloneqq \int_{0}^{1} f(x)\, \diff x$. For $f, g\in\mathcal{F}$, $f$ is said to \emph{weakly majorize} $g$, denoted $f\succsim_{\mathrm{w}} g$, if
\begin{equation}\label{eqn:weak-maj-def}
    \int_{0}^{x} f(s) \, \mathrm{d}s -m_f \leq \int_{0}^{x} g(s) \, \mathrm{d}s -m_g,
\end{equation}
for every $x\in X$. When $m_{g}=m_{f}$ holds in addition to \labelcref{eqn:weak-maj-def}, $f$ is said to \emph{majorize} $g$, denoted $f\succsim g$.

\subsubsection{Two-sided majorization constraints as CFIs}

We establish in \cref{lem:two-sided-maj,lem:two-sided-weak-maj} below that two-sided majorization and weak majorization comparisons between monotone functions can always be represented through CFIs. To do so, we simply let
\begin{equation}\label{eqn:integral_maj}
    I_{\varphi}(x)= \int_{0}^{x} \varphi(s) \, \diff s - m_{\varphi},
\end{equation}
for any $\varphi\in\mathcal{F}$ and $x\in X$.

\begin{lemma}[Majorization intervals]\label{lem:two-sided-maj}
    For every $f,g\in\mathcal{F}$ such that $f\succsim g$, the set
    \begin{equation}\label{eqn:maj_CFI}
        \mathcal{I}_{f,g} \coloneqq \Bigl\{ I\in \mathcal{K} \; \big\vert \; I_{f} \leq I \leq I_{g}, \, \partial I(X) \subseteq \bigl[f(0),f(1)\bigr] \Bigr\}, \tag{$\text{Maj}$}
    \end{equation}
    is a CFI. Furthermore, $\varphi\in\mathcal{F}$ satisfies $f\succsim \varphi \succsim g$ if and only if there exists $I\in \mathcal{I}_{f,g}$ such that $I=I_\varphi$. 
\end{lemma}

A proof of \cref{lem:two-sided-maj} is provided in \cref{secap:two_sided_maj_proof}. For weak majorization intervals, we require one additional (mild) technical condition: the functions being compared must share a uniform lower bound on their range. This condition is needed because weak majorization does not impose equality of means.\footnote{It is typically satisfied when comparing cumulative distribution functions, or quantile functions with a fixed range.} We thus consider the space $\mathcal{F}_{\ubar{s}} = \bigl\{f\in L^1(X) \; \vert \; \text{$f$ non-decreasing}, f(0)\geq \ubar{s}\bigr\}$ and obtain the following result.

\begin{lemma}[Weak majorization intervals]\label{lem:two-sided-weak-maj}
    For every $f,g\in\mathcal{F}_{\ubar{s}}$ such that $f\succsim_w g$, the set
    \begin{equation}\label{eqn:weak_maj_CFI}
        \mathcal{I}^{\mathrm{w}}_{f,g} = \Bigl\{ I\in \mathcal{K} \; \big\vert \; I_{f} \leq I \leq I_{g}, \, \partial I(X) \subseteq \bigl[\ubar{s},f(1)\bigr] \Bigr\} \tag{$\text{wMaj}$}
    \end{equation}
    is a CFI. Furthermore, $\varphi\in\mathcal{F}_{\ubar{s}}$ satisfies $f\succsim_{\mathrm{w}} \varphi \succsim_{\mathrm{w}} g$ if and only if there exists $I\in \mathcal{I}^{\mathrm{w}}_{f,g}$ such that $I=I_\varphi$.
\end{lemma}

The proof of \cref{lem:two-sided-weak-maj} follows the same logic as that of \cref{lem:two-sided-maj} and is therefore left to the reader.

\begin{remark}[One-sided majorization intervals]\label{rmk:one-sided-maj}
    The majorization sets in \cite{Kleiner2021} can be recovered as special cases of \labelcref{eqn:maj_CFI,eqn:weak_maj_CFI}. If $f(x)=\Ind_{x\geq \xi_{g}}$ for all $x\in X$, with $\xi_{g}\coloneqq \frac{g(1)-m_{g}}{g(1)-g(0)}$, then $\mathcal{I}_{f,g}$ corresponds to the set of mean-preserving contractions of $g$. Conversely, when $g(x)=m_f$ for all $x\in X$, then $\mathcal{I}_{f,g}$ corresponds to the set of mean-preserving spreads of $f$. Weak majorization intervals can be obtained as special cases of $\mathcal{I}^{\mathrm{w}}_{f,g}$ by taking $f(x)=\ubar{s}$ or $g(x)=f(1)$ for all $x\in X$. 
    
    We illustrate the one-sided majorization intervals in \cref{fig:maj_intervals} using the case where $F$ is some cumulative distribution function with support $[0,1]$.
    \begin{figure}[h!]
    \centering
        \includegraphics{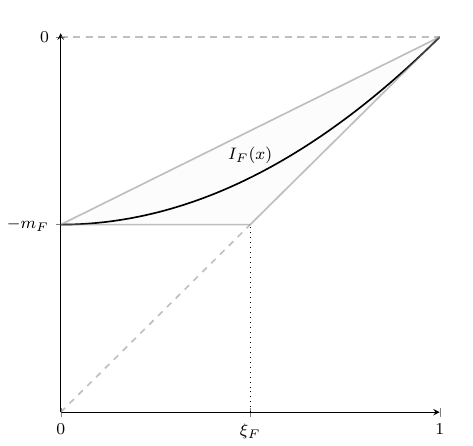}
        \caption{Majorization intervals. We set $F(x)=x$ for all $x\in [0,1]$, so $m_{F}=\xi_{F}=\tfrac{1}{2}$ and $I_{F}(x)=\frac{1}{2}x^{2}-m_{F}$. All convex functions in the shaded area correspond to integrals of functions that share the same mean as $F$. Those lying above $I_{F}$ are mean-preserving spreads of $F$, and those below are mean-preserving contractions of $F$. Convex functions that extend partially outside the shaded area correspond to integrals of functions that have a different mean value than $F$ (given by their value at $x=0$). Those lying below $I_{F}$ and above the increasing dashed line are ``mean-increasing'' contractions of $F$ (i.e., weakly majorize $F$) while those between $I_{F}$ and $0$ are ``mean-decreasing'' spreads of $F$ (i.e., are weakly majorized by $F$).}
        \label{fig:maj_intervals}
    \end{figure}

\end{remark}

\begin{remark}\label{rmk:isomorphism}
    \Cref{lem:two-sided-maj,lem:two-sided-weak-maj} show that (weak) majorization sets can be represented through CFIs. In fact, there exists a stronger one-to-one correspondence between these mathematical structures. We show in \cref{secap:bijection_ext_points} that $\varphi \mapsto I_{\varphi}$ is a continuous linear map that has a continuous linear inverse. Therefore, it preserves \emph{extreme points}. In other words, for every extreme point of $I\in \mathcal{I}_{f,g}$ (resp.~$\mathcal{I}^{\mathrm{w}}_{f,g}$), any selection $\varphi$ of $\partial I$ is an extreme point, and vice-versa.
\end{remark}

\subsubsection{Extreme points of majorization intervals}
 
\begin{corollary}\label{cor:ext_pt_maj}
    Let $f,g\in\mathcal{F}$ with $g$ continuous\footnote{We impose continuity of $g$ so that \cref{assu:diff} holds. \cite{Kleiner2021} impose the same condition (see their Theorem 2).} on $X$ and $f\succsim g$, and let $\mathcal{I}_{f,g}$ be defined according to \labelcref{eqn:maj_CFI}. Then, $I\in \ex(\mathcal{I}_{f,g})$ if and only if there exists a (possibly empty) countable collection $\mathcal{X}=\{X_{n}\}_{n\in \mathbb{N}}$ of non-degenerate and disjoint intervals $X_{n} = [a_n, b_n] \subseteq X$ such that:
    \begin{enumerate}
        \item For all $x\notin \bigcup_{n\in\mathbb{N}} X_{n}$, $I(x)\in \bigl\{I_{f}(x), I_{g}(x)\bigr\}$.
        \item For all $n\in \mathbb{N}$, $I\vert_{X_{n}}$ is affine, $I_{f}< I\vert_{\Int(X_{n})} < I_{g}$, and at least one of the following conditions holds:
        \begin{enumerate}[(a)]
            \item There exists $y\in\{a_n, b_n\}$ such that, for all $x\in X_{n}$, $I(x)=I_{g}(y)+ I_{g}'(y)(x-y)$. \label{cond:sat_tan_maj}
            \item For each $x\in\{a_n, b_n\}$, either there exists $m\in\mathbb{N}$ such that $b_m=a_n$ or $b_n=a_m$ and $I\vert_{X_{m}}$ satisfies condition \labelcref{cond:sat_tan_maj}, or $I(x)=I_{f}(x)$.
        \end{enumerate}
    \end{enumerate}
\end{corollary}

\cref{cor:ext_pt_maj} recovers Theorems 1 and 2 of \cite{Kleiner2021} as special cases and extends their results to the case of functions that are simultaneously mean-preserving spreads of some function $f$ and mean-preserving contractions of some function $g$.\footnote{Following \cref{rmk:isomorphism}, one can obtain the analog of \cref{cor:ext_pt_maj} directly expressed in the space $\Phi_{f,g}$ by taking the (right-)derivatives of $I$, $I_{f}$, and $I_{g}$ in the statement.}

The characterization of extreme points for weak majorization intervals $\mathcal{I}^{\mathrm{w}}_{f,g}$ is analogous but has one notable difference with \cref{cor:ext_pt_maj}. Since $I_{f}(0)$ and $I_{g}(0)$ need not coincide, extreme points of $\mathcal{I}^{\mathrm{w}}_{f,g}$ may also exhibit slope saturation on an interval containing $x=0$, where the derivative equals $\ubar{s}$. Since the formal statement would closely mirror \cref{cor:ext_pt_maj} with this additional case, we omit the formal statement for conciseness and instead illustrate it through the example of large contest problems in \cref{sec:contest}.

\subsection{Delegation with Type-Dependent Participation Constraints} \label{sec:delegation}

Consider the classical question posed by \citet{Holmstrom1977,Holmstrom1984}: how much discretion should be granted to a decision-maker who is better informed but potentially biased? We extend this \emph{optimal delegation} problem by studying how \emph{type-dependent} participation constraints alter the principal's design problem. Such constrains may arise, for example, from from veto rights in bargaining situations \citep{Kartik2021}, or fixed costs in monopoly regulation \citep{Amador2022}. We adopt a general perspective on the problem by considering \emph{stochastic} mechanisms for the principal, and \emph{rich} menus of outside options for the agent.

\subsubsection{Model}

\paragraph{Primitives}

A state of the world $\theta \in \Theta \coloneqq [\ubar{\theta}, \bar{\theta}] \subset \mathbb{R}$ is drawn according to a strictly increasing and absolutely continuous cumulative distribution function $F$, which admits a strictly positive and differentiable density $f$ on $[\ubar{\theta}, \bar{\theta}]$. An agent privately observes the state $\theta$ and makes a decision $a$ that affects herself and an uninformed principal. The set of feasible actions is $A \coloneqq \mathbb{R}$. Following \cite{Amador2013} and \cite{Krahmer2016}, we consider the following payoff specifications for the agent and the principal, respectively:
\begin{align*}
    u_{\mathrm{A}}(a,\theta) &= \theta a + b(a), \\[5pt]
    u_{\mathrm{P}}(a,\theta) &= \bigl(\theta + \beta(\theta)\bigr)a + b(a),
\end{align*}
where $\beta \colon \Theta \to \mathbb{R}$ is continuously differentiable and $b \colon A \to \mathbb{R}$ is strictly concave and differentiable.\footnote{This specification nests the standard quadratic loss formulation. Take $b(a)=-a^2/2$ . Add the (decision-irrelevant) terms $-(\theta+\beta(\theta))^2/2$ to $u_{\mathrm{P}}$ and $-\theta^2/2$ to $u_{\mathrm{A}}$. Then multiply these utilities by $2$.}

For each $\theta \in \Theta$, define
\begin{equation}
    \bar{u}(\theta) \coloneqq \max_{a \in A} \; \theta a + b(a),
\end{equation}
as agent $\theta$'s \emph{full discretion} payoff. We assume that $\lim_{a \to -\infty} b'(a) > - \ubar{\theta}$ and $\lim_{a \to +\infty} b'(a) < - \bar{\theta}$. Under these conditions, the agent's \emph{favorite action} in state $\theta$, denoted $\bar{a}(\theta)$, is uniquely determined by the first-order condition $b'\bigl(\bar{a}(\theta)\bigr) = -\theta$ for all $\theta \in \Theta$, and the mapping $\theta \in \Theta\mapsto \bar{a}(\theta)$ is continuous and strictly increasing. It follows that $\theta \in \Theta \mapsto \bar{u}(\theta)$ is differentiable and strictly convex.

\paragraph{Delegation mechanisms}

The principal may restrict the set of actions available to the agent by committing to a (potentially stochastic) \emph{direct delegation mechanism} $\Gamma \colon \Theta \to \Delta(A)$, which prescribes a lottery over actions $\Gamma(\,\cdot \mid \theta)$ for each possible state $\theta$ reported by the agent.\footnote{With the exception of \cite{Kovac2009} and \cite{Kleiner2022}, stochastic mechanisms are usually ruled-out in the optimal delegation literature.} For any mechanism $\Gamma$, an agent of type $\theta$ who reports $\theta'$ derives expected payoff $\theta a_{\Gamma}(\theta') + b_{\Gamma}(\theta')$ where, for each $\theta \in \Theta$,
\begin{equation*}
    a_{\Gamma}(\theta) \coloneqq \int_{A} a \, \diff \Gamma(a \, \vert \, \theta),
\end{equation*}
denotes the \emph{expected action} induced by $\Gamma$, and
\begin{equation*}
    b_{\Gamma}(\theta) \coloneqq \int_{A} b(a) \, \diff \Gamma(a \, \vert \, \theta).
\end{equation*}

In order to make things well-defined, we restrict attention to mechanisms $\Gamma$ such that $a_{\Gamma}(\theta) < + \infty$ for all $\theta\in \Theta$.

\paragraph{Incentive constraints}

A mechanism $\Gamma$ satisfies \emph{incentive-compatibility} if all the agents have an incentive to report their types truthfully under $\Gamma$. Formally,
\begin{equation}\label{eqn:IC_delegation}
    \forall \theta,\theta'\in\Theta, \quad \theta a_{\Gamma}(\theta) + b_{\Gamma}(\theta) \geq \theta a_{\Gamma}(\theta') + b_{\Gamma}(\theta'). \tag{IC-D}
\end{equation}

We assume that the agents can always flexibly choose their preferred option from a (compact) \emph{default menu} of outside options $M_{0}$ such that
\begin{equation*}
    \bigl\{\delta_{\bar{a}(\ubar{\theta})},\delta_{\bar{a}(\bar{\theta})}\bigr\} \subseteq M_{0} \subseteq \Delta(A),
\end{equation*}
rather than participating in the mechanism proposed by the principal. A mechanism $\Gamma$ satisfies \emph{individual-rationality} if it guarantees each type at least the expected payoff from its favorite option in $M_{0}$. Formally,
\begin{equation}\label{eqn:IR_delegation}
    \forall\theta\in\Theta, \quad \theta a_{\Gamma}(\theta) + b_{\Gamma}(\theta) \geq  u_{0}(\theta) \coloneqq \max_{\alpha \in M_{0}} \int_{A} \bigl(\theta a + b(a) \bigr) \diff\alpha(a). \tag{IR-D}
\end{equation}

\begin{remark}
    The inclusion of $\bigl\{\delta_{\bar{a}(\ubar{\theta})},\delta_{\bar{a}(\bar{\theta})}\bigr\}$ in $M_0$ is without loss of generality in a sense made precise in \cite{Kolotilin2025}: if indirect utilities are defined on an interval that is large enough, these two action will never be chosen. Moreover, since $u_0$ is a supremum over affine functions, $u_0$ is convex and tangent to $\bar{u}$ at $\ubar{\theta}$ and $\bar{\theta}$.
\end{remark}

\paragraph{Indirect utility functions}
 
The agents' \emph{indirect utility function} induced by some mechanism $\Gamma$ is defined as
\begin{equation*}
    \forall \theta\in\Theta, \quad u(\theta) = \max_{\theta'\in \Theta} \; \theta a_{\Gamma}(\theta') + b_{\Gamma}(\theta').
\end{equation*}

An indirect utility utility function $u\colon\Theta\to \mathbb{R}$ is \emph{implementable} if there exists mechanism $\Gamma$ which satisfies \labelcref{eqn:IC_delegation,eqn:IR_delegation} such that $u(\theta) = \theta a_{\Gamma}(\theta) + b_{\Gamma}(\theta)$ for all $\theta\in \Theta$.

\subsubsection{Feasible, Extreme and Optimal Delegation Mechanisms}

\paragraph{Reformulating the principal's problem}

The principal's problem consists in maximizing 
\begin{equation*}
    \int_{\ubar{\theta}}^{\bar{\theta}} \Bigl\{\bigl(\theta + \beta(\theta)\bigr) a_{\Gamma}(\theta)+ b_{\Gamma}(\theta)\Bigr\} \, \diff F(\theta),
\end{equation*}
over mechanisms $\Gamma$ that satisfy \labelcref{eqn:IC_delegation,eqn:IR_delegation}.
We follow the approach of \cite{Krahmer2016} and \cite{Kleiner2022} by characterizing delegation mechanisms which satisfy \labelcref{eqn:IC_screening,eqn:IR_screening} by their induced indirect utility functions. Specifically, building on \cite{Kleiner2022}, we show that for any menu $M_{0}$, the set of implementable indirect utility functions is a CFI $\mathcal{U}_{\mathrm{D}}$. We call it the \emph{delegation interval}. Furthermore, the principal's problem can be written as a linear program on $\mathcal{U}_{\mathrm{D}}$. See \cref{secap:delegation_cfi_lp_proof} for a short proof.

\begin{lemma}\label{lem:delegation_cfi_lp}
    An indirect utility function $u$ is implementable if and only if $u\in \mathcal{U}_{\mathrm{D}}$, where
    \begin{equation*}
        \mathcal{U}_{\mathrm{D}} \coloneqq \Bigl\{u\in \mathcal{K}(\Theta) \; \big\vert \; u_{0} \leq u \leq \bar{u}, \, \partial u(\Theta) \subseteq \bigl[\bar{a}(\ubar{\theta}), \bar{a}(\bar{\theta})\bigr]\Bigr\}.
    \end{equation*}
    
    Moreover, the principal's problem can be written as
    \begin{equation}\label{eqn:delegation_LP}
        \max_{u\in\mathcal{U}_{\mathrm{D}}} \int_{\Theta} u \, \mathrm{d}\mu_{\mathrm{D}},\tag{Del}
    \end{equation}
    where, for all $B\in\mathcal{B}(\Theta)$,
    \begin{equation*}\label{eqn:delegation_measure}
        \mu_{\mathrm{D}}(B) = \int_{B} \psi_{\mathrm{D}} \, \mathrm{d}\nu,
    \end{equation*}
    with $\nu\coloneqq \delta_{\bar{\theta}} - \delta_{\ubar{\theta}} + \lambda$ and
    \begin{equation*}
        \psi_{\mathrm{D}}(\theta) = \left\{\begin{array}{ll}
           \beta(\theta)f(\theta)  & \text{if $\theta\in\{\ubar{\theta},\bar{\theta}\}$,} \\[5pt]
            f(\theta) - \frac{\diff}{\diff\theta}\bigl[\beta(\theta) f(\theta)\bigr] & \text{if $\theta\in (\ubar{\theta},\bar{\theta})$.}
        \end{array}
        \right.
    \end{equation*}
\end{lemma}

\paragraph{Extremal Delegation Mechanisms}

The characterization of extreme points in \cref{thm:ext_pt} allows us to describe qualitative properties of extremal delegation mechanisms. \cref{thm:ext_pt} implies that an extremal mechanism features at most a countable number of discontinuities in the \emph{mean action function} $\theta\mapsto a_{\Gamma}(\theta)$ and, thus, might only induce a countable number of stochastic actions in addition to the (potentially stochastic) outside options.\footnote{This extends an insight from the literature showing that, in the absence of outside options, it is often sufficient to restrict attention to delegation mechanisms that induce a countable \citep{Kleiner2021} or even finite number of discontinuities in the (mean-)action function \citep{Saran2024,Amador2025}.} For instance, \cref{fig:extremal_delegation_mechanisms} displays two different kinds of mechanisms: an \emph{interval} delegation mechanism with a \emph{floor action}, and an extremal mechanism with a \emph{stochastic} action. Stochastic actions that are offered in addition to $M_0$ correspond to affine pieces of the indirect utility function that lie strictly below $\bar{u}$ and are not tangent to it. Furthermore, there can be at most \emph{one} stochastic action that is chosen by neighboring types (in addition to the ones contained in $M_0$).

\begin{figure}[t]
        \centering
        \begin{subfigure}[t]{0.495\linewidth}
            \centering
            \includegraphics{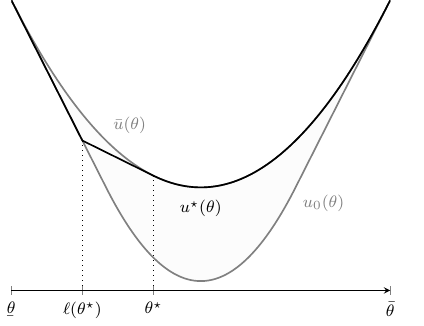}
        \caption{Action floor.}
        \label{fig:interval_delegation}
        \end{subfigure}
        \begin{subfigure}[t]{0.495\linewidth}
            \centering
            \includegraphics{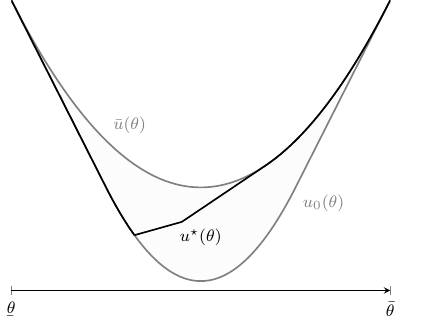}
          \caption{One stochastic action.}
        \end{subfigure}
    \caption{Extremal mechanisms.}
    \label{fig:extremal_delegation_mechanisms}
\end{figure}

\begin{remark}[Equivalence to Bayesian Persuasion]
\cite{Kolotilin2025} show that linear \emph{balanced delegation}---where $M_{0}=\bigl\{\delta_{\bar{a}(\ubar{\theta})},\delta_{\bar{a}(\bar{\theta})}\bigr\}$---and the standard mean-measurable Bayesian persuasion problem \citep{Gentzkow2016} are mathematically equivalent. The equivalence extends readily to delegation problem with type-dependent outside options and \emph{constrained persuasion} as introduced in \cref{sec:persuasion}. This becomes apparent in \cref{fig:equiv_deleg_pers}: the two CFIs are isomorphic. Moreover, both the principal's problem in delegation as well as the sender's problem in persuasion can be written as linear problems on these CFIs.
\begin{figure}[t]
\begin{subfigure}[t]{0.495\linewidth}
    \centering
    \includegraphics{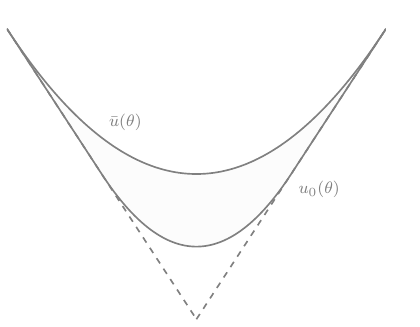}
    \caption{The delegation interval.}
\end{subfigure}
\begin{subfigure}[t]{0.495\linewidth}
    \centering
    \includegraphics{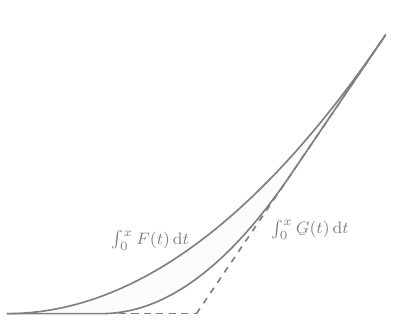}
    \caption{The persuasion interval (c.f.~\cref{sec:persuasion}).}
\end{subfigure}
\caption{The Persuasion-Delegation Equivalence.}
\label{fig:equiv_deleg_pers}
\end{figure}
\end{remark}

\paragraph{Constant Bias}

We make the following assumption for the rest of this section:
\begin{assumption}\label{ass:regular_delegation}
    The principal has a constant bias towards higher actions. Formally, there exists $\beta\in\mathbb{R}_{++}$ such that $\beta(\theta)=\beta$ for all $\theta\in \Theta$. Moreover, the distribution of the state admits a log-concave density function $f$.
\end{assumption}

Under \cref{ass:regular_delegation}, optimal mechanisms have a simple structure. For all $\theta^{\star} \in \Theta$, define
\begin{equation}
    u_{\theta^{\star}}(\theta) = \left\{\begin{array}{ll}
        \ubar{u}(\theta)  & \text{if $\theta \in \bigl[\ubar{\theta}, \ell(\theta^{\star})\bigr]$} \\[5pt]
        \bar{u}(\theta^{\star}) + \bar{u}'(\theta^{\star}) (\theta-\theta^{\star}) & \text{if $\theta\in \bigl(\ell(\theta^{\star}), \theta^{\star}\bigr)$} \\[5pt]
        \bar{u}(\theta) & \text{if $\theta\in [\theta^{\star},\bar{\theta}]$}
    \end{array}
    \right.,
\end{equation}
where $\ell(\theta^{\star})$ satisfies $\bar{u}(\theta^{\star}) + \bar{u}'(\theta^{\star}) \bigl(\ell(\theta^{\star})-\theta^{\star}\bigr) = \ubar{u}\bigl(\ell(\theta^{\star})\bigr)$.\footnote{By the properties of $\mathcal{U}^\mathrm{D}$, $\ell(\theta^{\star})\in \Theta$ exists and is unique for all $\theta^{\star}\in \Theta$. For a graphical illustration, see \cref{fig:interval_delegation}.\label{foot:exist_cutoff_floor}} By \cref{thm:ext_pt}, $u_{\theta^{\star}}$ is an extreme point of $\mathcal{U}^\mathrm{D}$ for all $\theta^{\star}\in \Theta$.

The mechanism that induces the indirect utility $u_{\theta^{\star}}$ has a simple indirect implementation: it corresponds to setting a lower bound, depending on $\theta^{\star}$, for the agent's action. The agent can either choose an action above the lower bound (his \emph{delegation set}) or an item from the menu of outside options. Under \cref{ass:regular_delegation}, optimal mechanisms take this form:

\begin{proposition}\label{prop:delegation_logconc}
    Suppose \cref{ass:regular_delegation} holds. Then there exists $\theta^{\star}$ such that the three following conditions hold:
    \begin{enumerate}[(i)]
    \setlength\itemsep{0.33cm}
        \item $\psi_{\mathrm{D}}(\theta) \leq 0$ for all $\theta \in \bigl[\ubar{\theta}, \ell(\theta^{\star})\bigr]$.
        \item $\displaystyle\int_{\ell(\theta^{\star})}^{\theta^{\star}} (t-\theta^{\star}) \psi_{\mathrm{D}}(t) \, \diff t=0$.
        \item $\psi_{\mathrm{D}}(\theta) \geq 0$ for all $\theta\in [\theta^{\star},\bar{\theta}]$.
    \end{enumerate}
    These conditions imply that $u_{\theta^{\star}}$ is optimal for \labelcref{eqn:delegation_LP}.
\end{proposition}

The proof of \cref{prop:delegation_logconc} in \cref{secap:delegation_logconc_proof} provides a constructive method for identifying $\theta^{\star}$ and establishes its optimality using \cref{thm:charac_opt}. The intuition is as follows. Log-concavity of $f$ implies that the state is likely to be moderate. Moreover, the agent prefers lower actions than the principal for every state. As a result, the principal optimally mitigates the downward bias of the agent by imposing a an \emph{action floor}. Furthermore, the optimal mechanism gives the agent his preferred outside option in $M_{0}$ for all state realizations such that the action floor induces lower utility for the agent relative to that outside option.

We further show that the richness of the default menu of outside options affects the size of the optimal delegation set. There is a direct relationship between the value the menu delivers to the agent and the degree of discretion the principal grants in the optimal mechanism. Intuitively, a richer menu of outside options strengthens the agent's bargaining power, which forces the principal to expand the delegation set to preserve the agent's participation incentives. The following corollary formalizes this connection, and a short proof is provided in \cref{secap:delegation_compstat_proof}.

\begin{corollary}\label{cor:delegation_compstat}
    Suppose that \cref{ass:regular_delegation} holds. For any two menus of outside options that induce indirect utilities $\ubar{u}_1$,$\ubar{u}_2$ such that $\ubar{u}_1 \geq \ubar{u}_2$, the optimal delegation set is larger under $\ubar{u}_1$ than under $\ubar{u}_2$.
\end{corollary}

\subsection{Large Contest Design with Limited Disposal} \label{sec:contest}

We now consider large contest problems with limited disposal: in contrast to both free disposal and mandatory allocation of all goods, the principal is free to dispose of some goods in our framework. We show that an assignment is implementable if and only if it satisfies a two-sided majorization constraint. This two-sided majorization constraint depends on the distribution of available qualities and the limits imposed on disposing goods.

\subsubsection{Model}

\paragraph{Primitives}

A principal holds a continuum of indivisible prizes of unit mass with differentiated quality $x \in (0,1]$.\footnote{We normalize the highest quality to $1$ without loss of generality.} Prize qualities are distributed according to the cumulative distribution function $G$. The principal allocates these prizes to a continuum of agents of unit mass. She may also choose not to allocate any prize, which we model as assigning the ``null prize'' of quality $x=0$, available in infinite supply.  

Agents expend resources---effort, money, or time, all measured in monetary units---to obtain prizes. An agent of type $\theta \in \Theta \coloneqq [0,1]$ derives utility $\theta x - t$ from receiving a prize of quality $x$ at monetary cost $t$. Agent types are distributed according to the strictly increasing and absolutely continuous cumulative distribution function $F$ with continuous density $f$.  

\begin{remark}
    The literature has offered different interpretations of the variable $t$. \citet{Loertscher2022} and \citet{Bergemann2025} interpret the framework as a monopolistic screening problem in which the seller holds a fixed stock of quality-differentiated goods, so $t$ represents a monetary transfer to the seller. \citet{Akbarpour2024a} and \citet{Ashlagi2024} view it as the allocation of public resources, where a planner decides how to distribute scarce goods among agents with heterogeneous needs, so $t$ corresponds to the revenue collected by the planner. \citet{Olszewski2016,Olszewski2020} and \citet{Kleiner2021} cast it as a contest environment in which a large number of participants compete for prizes of varying qualities by sending costly signals or exerting effort, so $t$ denotes the signaling/effort cost.
\end{remark}

\paragraph{Assignment mechanisms and allocative constraints}

Let $\mathrm{CDF}[0,1]$ denote the set of all cumulative distribution functions with support included in $[0,1]$. The principal commits to a \emph{direct assignment mechanism} $(\Gamma,t)$, consisting of a (probabilistic) assignment rule $\Gamma\colon \Theta \to \mathrm{CDF}[0,1]$ and a monetary cost rule $t\colon \Theta \to \mathbb{R}_{+}$. Under mechanism $(\Gamma,t)$, an agent reporting type $\theta$ incurs a monetary cost $t(\theta)\geq 0$ to receive a prize drawn from the lottery characterized by the cumulative distribution function $\Gamma(\,\cdot\; \vert \; \theta)$.

Let $0 \leq m \leq \int_{0}^{1} x \, \diff G(x)$. An assignment $\Gamma$ is called \emph{feasible} if it satisfies

\begin{align}
    \forall x\in [0,1], \quad \int_{\Theta} \Gamma(x \, \vert \, \theta) \, \diff F(\theta) \geq G(x), \tag{$\mathrm{F}_1$} \label{eqn:feasibility_contest_capacity} \\[5pt]
    \int_{\Theta} \int_{0}^{1} x \, \diff \Gamma (x \, \vert \, \theta) \, \diff F (\theta) \geq m. \tag{$\mathrm{F}_2$} \label{eqn:feasibility_contest_mean}
\end{align}

Condition \labelcref{eqn:feasibility_contest_capacity} is a physical constraint that requires the distribution of the assigned prize qualities to be first-order-stochastically dominated by $G$.\footnote{Intuitively, the mass of goods of quality less than $x \in [0,1]$ that is assigned to agents under $\Gamma$ has to be higher than $G(x)$ since the principal can always achieve downward FOSD shifts by randomizing with the null prize.} Condition \labelcref{eqn:feasibility_contest_mean} is an allocative constraint that imposes a minimal average quality the principal has to allocate. This specification generalizes two standard benchmarks in the literature. Under \emph{mandatory allocation} \cite[][Section 4.2]{Kleiner2021}, the principal must distribute all available prizes (so $m= \int_{0}^{1} x \, \diff G(x)$). Under \emph{free disposal} \citep[][]{Loertscher2022,Akbarpour2024a,Ashlagi2024,Bergemann2025}, the principal can assign null prizes to all agents (so $m=0$). Our framework captures the intermediate case of \emph{limited disposal}, where the principal faces a lower limit on the average quality he must assign.

\paragraph{Expected assignments and incentive constraints}

For any assignment $\Gamma$ and type $\theta\in \Theta$, we denote as
\begin{equation*}
  x_{\Gamma}(\theta) \coloneqq \int_{0}^{1} x \, \diff \Gamma(x \, \vert \, \theta),
\end{equation*}
the corresponding \emph{expected assignment}.

Since agents' utilities are linear in prize quality, incentive-compatibility and individual-rationality of mechanism $(\Gamma,t)$ can be expressed solely in terms of expected assignments and monetary costs:
\begin{align}
    \label{eqn:IC_contest} \forall \theta,\theta'\in \Theta,& \quad \theta x_{\Gamma}(\theta) - t(\theta) \geq 
    \theta x_{\Gamma}(\theta') - t(\theta'), \tag{IC-C}  \\[5pt]
    \label{eqn:IR_contest} \forall \theta\in \Theta,& \quad \theta x_{\Gamma}(\theta) - t(\theta) \geq 0. \tag{IR-C}
\end{align}

\subsubsection{Feasible, Extreme and Optimal Assignments}

\paragraph{Expected quantile assignments}

For any mechanism $(\Gamma,t)$, it will be convenient to consider the corresponding expected \emph{quantile} assignment $\chi_{\Gamma}(q) \coloneqq x_{\Gamma}\bigl(F^{-1}(q)\bigr)$ for all $q\in [0,1]$. We say that an expected quantile assignment $\chi\colon[0,1]\to [0,1]$ is \emph{implementable} if there exists a mechanism $(\Gamma,t)$ satisfying \labelcref{eqn:feasibility_contest_capacity}, \labelcref{eqn:feasibility_contest_mean}, \labelcref{eqn:IC_contest,eqn:IR_contest} such that $\chi(q) = \chi_{\Gamma}(q)$ for all $q\in [0,1]$.

Characterizations of implementable assignments are well-known in the two benchmark cases of mandatory allocation and free disposal. In the mandatory allocation case ($m= \int_{0}^{1} x \, \diff G(x)$), an expected quantile assignment $\chi\colon [0,1]\to[0,1]$ is implementable if and only if it is non-decreasing and majorized by the positive assortative quantile assignment $q\in[0,1]\mapsto G^{-1}(q)$ \citep[][Proposition 4]{Kleiner2021}. Under free disposal ($m=0$), this characterization changes to weak majorization by the positive assortative quantile assignment (\citealp{Hart2015}, Theorem 1; \citealp{Kleiner2021}, Theorem 3).

We obtain a similar characterization in the case of limited disposal, including the two benchmarks as extremes:

\begin{proposition}\label{prop:imp_quantil_assign}
    An expected quantile assignment $\chi\colon[0,1]\to [0,1]$ is implementable if and only if it is non-decreasing and $G^{-1}\succsim_{\mathrm{w}} \chi \succsim_{\mathrm{w}} H_{m}^{-1}$ where $H_{m}(x)=\Ind_{x \geq m}$ for all $x \in [0,1]$.
\end{proposition}

We prove \cref{prop:imp_quantil_assign} in \cref{secap:imp_quantil_assign_proof}. It captures a salient intuition: expected assigned qualities across types have to be less equal than giving everyone the same expected quality. This expected quality is bounded below by \labelcref{eqn:feasibility_contest_mean}. Jointly, these two aspects induce the lower bound $H_{m}^{-1}$ in the weak majorization order. On the other hand, as in the benchmark cases, expected assigned qualities have to be more equal than the assortative quantile assignment $G^{-1}$, as captured by the upper bound in the weak majorization order.

\paragraph{Feasible assignments as a CFI}

\cref{prop:imp_quantil_assign} and \cref{lem:two-sided-weak-maj} jointly imply that the set of implementable expected quantile assignments can be represented by the CFI $\mathcal{I}^{\mathrm{w}}_{G^{-1},H^{-1}_{m}}$ (defined according to \labelcref{eqn:weak_maj_CFI}), which elements correspond to \emph{cumulative} expected quantile assignments:
\begin{equation}\label{eqn:cumulative_exp_assignment}
    I_{\chi}(q)\coloneqq \int_{0}^{q} \chi(s) \, \mathrm{d}s - m_{\chi},
\end{equation}
for each $q\in[0,1]$.\footnote{Recall from \cref{sec:majorization} that $m_{\chi}\coloneqq \int_{0}^{1} \chi(q) \, \mathrm{d}q$.}

\paragraph{Extremal prize assignments}

Following \cref{thm:ext_pt}, \cref{fig:contest_interval} illustrates the extreme points of $\mathcal{I}_{G^{-1},H^{-1}_{m}}$.
\begin{figure}[h!]
\centering
    \begin{subfigure}[b]{0.49\linewidth}
        \centering
        \includegraphics{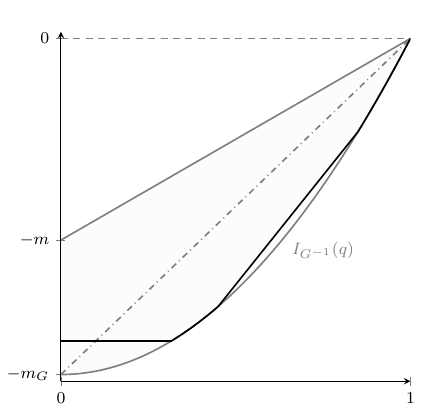}
        \subcaption{\labelcref{eqn:feasibility_contest_mean} is slack.}
        \label{fig:contest_interval_classic_ext_point}
    \end{subfigure}
    \begin{subfigure}[b]{0.49\linewidth}
    \centering
        \includegraphics{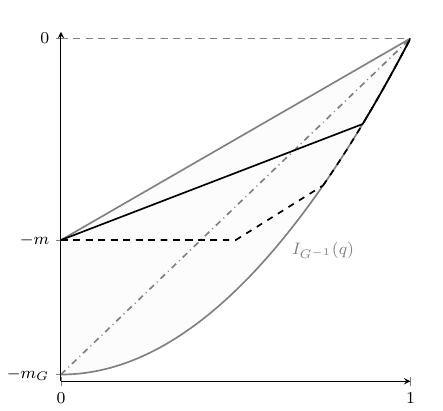}
        \subcaption{\labelcref{eqn:feasibility_contest_mean} binds.}
        \label{fig:contest_interval_new_ext_point}
    \end{subfigure}
    
    \caption{The contest interval $\mathcal{I}^{\mathrm{w}}_{G^{-1},H_{m}}$ is shown as the shaded gray area. The gray dashed-dotted line represents the extreme point corresponding to random assignment of all prizes. The black lines in the left and right panels (solid and dashed) depict other extreme points of $\mathcal{I}^{\mathrm{w}}_{G^{-1},H_{m}}$.}
    \label{fig:contest_interval}
\end{figure}

In \cref{fig:contest_interval_classic_ext_point}, the extremal assignment yields a mean quality $-I^{\star}(0)>m$, so \labelcref{eqn:feasibility_contest_mean} is slack. The lowest constant segment (hence lower-slope saturated) corresponds to the principal assigning an expected quality of zero to all type-quantiles in the interval. This segment is followed by alternating intervals where the designer either implements the positive assortative assignment, i.e., $I$ coincides with $I_{G^{-1}}$, or \emph{irons} the assortative allocation, thereby providing a \emph{lottery} that assigns the same expected quality to all type-quantiles in that interval. In contrast, \cref{fig:contest_interval_new_ext_point} shows extremal assignments that are caused by the limitation on disposal, i.e., \labelcref{eqn:feasibility_contest_mean} binds. Because of this constraint, extremal assignments can contain lotteries that consist of both ironing and disposal of qualities, as illustrated by the increasing affine pieces in \cref{fig:contest_interval_new_ext_point}.

\paragraph{Optimal prize assignments}

We now assume that the principal seeks to maximize aggregate effort:
\begin{equation}\label{eqn:contest_pb}
    \max_{(\Gamma,t)} \; \int_{\Theta} t(\theta) \, \diff F(\theta),
\end{equation}
subject to $(\Gamma,t)$ satisfying \labelcref{eqn:feasibility_contest_capacity}, \labelcref{eqn:feasibility_contest_mean}, \labelcref{eqn:IC_contest,eqn:IR_contest}.

We show in \cref{lem:contest_cfi_lp} below that problem \labelcref{eqn:contest_pb} can be represented as a linear program on the CFI $\mathcal{I}_{G^{-1},H^{-1}_{m}}$:

\begin{lemma}\label{lem:contest_cfi_lp}
    The contest design problem \labelcref{eqn:contest_pb} can be written as
    \begin{equation}\label{eqn:contest_LP}
        \max_{I\in\mathcal{I}_{G^{-1},H^{-1}_{m}}} \int_{[0,1]} I \, \mathrm{d}\mu_{\mathrm{C}},\tag{Cont}
    \end{equation}
    where, for all $A\in\mathcal{B}\bigl([0,1]\bigr)$,
    \begin{equation*}\label{eqn:contest_measure}
        \mu_{\mathrm{C}}(A) = \int_{A} \psi_{\mathrm{C}} \, \mathrm{d}\nu,
    \end{equation*}
    with $\nu\coloneqq \delta_1 - \delta_{0} + \lambda$ and
    \begin{equation*}
        \psi_{\mathrm{C}}(q) = \left\{\begin{array}{ll}
           v\bigl(F^{-1}(q)\bigr)  & \text{if $q\in\{0,1\}$,} \\[5pt]
           -v'\bigl(F^{-1}(q)\bigr)/f\bigl(F^{-1}(q)\bigr) & \text{if $q\in (0,1)$,}
        \end{array}
        \right.
    \end{equation*}
    where
    \begin{equation*}
        v(\theta) \coloneqq \theta - \frac{1-F(\theta)}{f(\theta)},
    \end{equation*}
    for each $\theta \in \Theta$.
\end{lemma}

We prove \cref{lem:contest_cfi_lp} in \cref{secap:contest_cfi_lp_proof}. The proof shows that optimal contest design with limited disposal reduces to a linear program on the CFI $\mathcal{I}_{G^{-1},H^{-1}_{m}}$. Similarly as in \cref{sec:screening_optimalallocations}, this is not limited to effort maximization: \cref{lem:contest_cfi_lp} holds in a similar form for more general welfare objectives as long as they are linear in the expected allocation.\footnote{We thus recover the setting of \cite{Akbarpour2024a}---in which $m=0$---as a special case.}

In particular, \cref{lem:contest_cfi_lp} allows us to derive effort-maximizing contests under \emph{limited disposal}:
\begin{proposition}\label{prop:contest_optimality_myerson}
    Assume that $F$ is Myerson-regular. For each $m\in[0,\int_{0}^{1} x \, \diff G(x)]$ there exists $\theta^{\star}_{m}$ such that the effort-maximizing assignment $x^{\star}_{m}$ features exclusion for all types below $\theta^{\star}_{m}$ and is positive assortative for higher types---i.e., $x^{\star}_{m}(\theta) = 0$ for all $\theta\in [0,\theta^{\star}_{m})$ and $x^{\star}_{m}(\theta)=G^{-1}\bigl(F(\theta)\bigr)$ for all $\theta \in [\theta^{\star}_{m},1]$. In particular, if $m=\int_{0}^{1} x \, \diff G(x)$, then $x_{m}^{\star}$ is positive assortative for all types---i.e., $\theta_{m}=0$.

    Furthermore, $\theta^{\star}_{m}$ is decreasing in $m$, and the constraint on disposal \labelcref{eqn:feasibility_contest_mean} binds whenever $I_{G^{-1}}\bigl(F(\theta^\star_0)\bigr) > -m$.
\end{proposition}

\cref{prop:contest_optimality_myerson} shows that with a sufficiently regular type distribution, more stringent standards on disposal lead to fewer types being excluded from the mechanism. It is also a generalization of Proposition 4.4 in \cite{Kleiner2021}, showing that the assortative assignment is effort-maximizing for a Myerson-regular type distribution. The proof can be found in \cref{secap:contest_optimality_myerson_proof}.

\subsection{Mean-Based Bayesian Persuasion with Informativeness Constraints}\label{sec:persuasion}

We now consider the Bayesian persuasion problem of \cite{Kamenica2011} with exogenous informativeness constraints. The standard Bayesian persuasion framework usually retains the assumption that all experiments are
feasible.\footnote{We refer to \cite{Kamenica2019} and \cite{Bergemann2019} for reviews of the Bayesian persuasion and information design literature.} In practice, however, information may be costly to provide, inducing an upper bound on informativeness. Moreover, legal or regulatory requirements may impose a minimum level of disclosure, inducing a lower bound on informativeness. This section shows how such information-constrained persuasion problems can be represented and solved using our tools.

\subsubsection{Model}

\paragraph{Primitives}

A state of the world $\omega\in\Omega\coloneqq[0,1]$ is drawn according to a prior distribution with continuous cumulative distribution function $F$. Before the state is realized, the sender commits to a Blackwell experiment $\sigma\colon \Omega \to \Delta(S)$. We assume the set of signal realization $S$ to be rich enough (in particular, we assume $\Omega\subseteq S$). The receiver observes the experiment chosen by the sender, updates the prior belief $F$ via Bayes' rule upon receiving any signal $s\in S$, and acts optimally given the resulting posterior belief. Following much of the recent literature after \cite{Gentzkow2016}, \cite{Kolotilin2018} and \cite{Dworczak2019}, we restrict attention to environments in which the receiver's optimal behavior and the sender's indirect utility $v$ depend only on the \emph{posterior mean} of the state, denoted $x\in X=[0,1]$.\footnote{In terms of primitives, this is equivalent to requiring that both the receiver's and the sender's preferences are affine in the state for every action. This assumption holds in various focal economic environments \citep[see][Remark 1]{Curello2024}.} The sender's value function $v$ is assumed to satisfy a minimal degree of regularity: (i) $v$ is absolutely continuous; and (ii) its almost-everywhere derivative $v'$ has bounded variation. We comment on these assumptions in \cref{rmk:regularity-v} below.

\paragraph{Distributions of posterior means}

By Bayes' rule, any Blackwell experiment induces a cumulative distribution function $G$ over posterior means. It is well known that, when the sender can freely access all possible Blackwell experiments, a distribution $G$ of posterior means can be induced if and only if $G$ \emph{majorizes} $F$ \citep[see][Proposition 2]{Kolotilin2018}.\footnote{This result can be attributed to \citet{Blackwell1951,Blackwell1953} and was later proven in greater generality by \citet{Strassen1965}.}

\paragraph{Informativeness}

The majorization order also provides a ranking of distributions of posterior means in terms of their \emph{informativeness}. Specifically, a feasible distribution $G_{1}$ is more informative than $G_{2}$ if and only if $G_{2}$ majorizes $G_{1}$. This informativeness order is in the spirit of Blackwell's: $G_{1}$ is more informative than $G_{2}$ precisely when it is preferred ex-ante by \emph{every} expected-utility decision maker whose preferences depend \emph{linearly} on the state.\footnote{Note, however, that this order is \emph{not} equivalent to Blackwell's, since Blackwell's order ranks an experiment higher whenever it yields a higher expected utility for \emph{any} decision problem.}

\subsubsection{Feasible, Extreme and Optimal Distribution of Posterior Means}

\paragraph{Constrained persuasion}

We thus incorporate informativeness constraints into the mean-based Bayesian persuasion problem by imposing bounds in the majorization order. Specifically, we assume that the sender can induce only those distributions of posterior means $G$ such that $\ubar{G} \succsim G \succsim \bar{G}$, where $\ubar{G} \succsim \bar{G} \succsim F$. This captures the idea that the sender must reveal at least as much information as in $\ubar{G}$, but no more than in $\bar{G}$.

Optimization problems involving such two-sided majorization constraints are usually intractable. We address this challenge by adopting the approach of \cite{Gentzkow2016}, reformulating the persuasion problem~\labelcref{eqn:cons_pers_pb} as a linear program over the convex function interval $\mathcal{I}_{\ubar{G},\bar{G}}$ (defined according to \labelcref{eqn:maj_CFI}).

\begin{lemma}\label{lem:baypers_infoconstraints}
    Assume that $\bar{G}$ is continuous so \cref{assu:diff} holds. The sender's problem can be written as
    \begin{equation}\label{eqn:cons_pers_pb_CFI}
    \max_{I \in \mathcal{I}_{\ubar{G},\bar{G}}} \int_{X} I \, \mathrm{d}\mu_v, \tag{Pers}
\end{equation}
where $\mu_v$ is the finite signed Radon measure defined by
\begin{equation}
    \mu_v\bigl([0,x]\bigr) = v'(x),
\end{equation}
for every $x \in [0,1]$.
\end{lemma}

The proof can be found in \cref{secap:baypers_infoconstraints_proof}. The measure $\mu_{v}$ can be understood as the ``second derivative'' of the sender's value function $v$. Intuitively, regions where $\mu_{v}$ assigns positive mass correspond to the sender preferring information revelation, while regions where $\mu_{v}$ assigns negative mass correspond to the sender preferring information concealment. 

\begin{remark}[Regularity of $v$]\label{rmk:regularity-v}
    The assumption that $v'$ has bounded variation is precisely the \emph{minimal} condition required for the second derivative of $v$ to exist as a finite signed Radon measure \citep[][Chapter~3]{Folland1999}. This condition is satisfied, in particular, under a regularity assumption commonly encountered in the literature: namely, when $v$ is continuous and piecewise convex-concave (see, e.g., \citealp{Dworczak2019}, Definition 1; or \citealp{Curello2024}, Definition 3). In this case, it follows from \cite{Dudley1977} that the intervals where $v$ is convex (resp.~concave) correspond precisely to the regions where the measure $\mu_{v}$ assigns positive (resp.~negative) mass.\footnote{Note also that $v'$ being of bounded variation is equivalent to $v$ being of ``\emph{bounded curvature}'' in the following sense: there exist two convex functions $f, g \colon X \to \mathbb{R}$ such that $v = f - g$ if and only if there exists a function of bounded variation $\phi \colon X \to \mathbb{R}$ such that $v(b) - v(a) = \int_{a}^{b} \phi(x) \, \mathrm{d}x$ for all $a, b \in X$ \citep[see][Theorem~A, p.~23]{Roberts1974}.} When, furthermore, $v'$ is absolutely continuous, the measure $\mu_{v}$ admits the classical second derivative $v''$ as its (signed) density function \citep[see, e.g.,][]{Lyu2025}.
\end{remark}

\paragraph{S-shaped value function}

We now apply our results to the canonical environment in which the sender has an S-shaped utility function \citep{Kolotilin2022}. Formally, there exists $\hat{x}\in X$ such that $v$ is convex on $[0,\hat{x}]$ and concave on $[\hat{x},1]$. Absent informativeness constraints, upper censorship---revealing the state fully below some cutoff and pooling all states above the cutoff---is therefore optimal following \cite{Kolotilin2022}, Theorem 1. We further impose that $v$ is differentiable, and thus say that $v$ is \emph{smoothly S-shaped}. This assumption allows us to derive results on optimality and comparative statics with respect to the informativeness constraint. For this define,
\begin{equation*}
    I_{x^\star}(x) \coloneqq \left\{\begin{array}{ll}
        I_{\bar{G}}(x)  & \text{if $x \in [0,x^{\star}]$,} \\[5pt]
        I_{\bar{G}}(x^{\star}) + \bar{G}(x^{\star})(x-x^{\star}) & \text{if $x\in \bigl(x^{\star},h(x^{\star})\bigr)$,} \\[5pt]
        I_{\ubar{G}}(x) & \text{if $x \in \bigl[h(x^{\star}), 1\bigr]$,}
    \end{array}
    \right.
\end{equation*}
for any $x^{\star} \in [0,1]$, where $h(x^{\star})$ satisfies $I_{\bar{G}}(x^{\star}) + \bar{G}(x^{\star})\bigl(h(x^{\star})-x^{\star}\bigr) = I_{\ubar{G}}\bigl(h(x^{\star})\bigr)$.\footnote{For any $x^{\star}\in X$, $h(x^{\star})$ exists and is unique. See \cref{secap:infodesign_sshaped_proof}.}

\begin{proposition}\label{lem:infodesign_sshaped}
    Suppose $v$ is \emph{smoothly S-shaped}. Then there exists $x^{\star}\in X$ such that the three following conditions hold:
        \vskip0.33cm
        \begin{enumerate}[(i)]
        \setlength\itemsep{0.33cm}
            \item $v''(x) \geq 0$ for all $x \in [0,x^{\star}]$.
            \item $\displaystyle\int_{x^{\star}}^{h(x^{\star})} (x-x^{\star})v''(x) \, \diff x = 0$.
            \item $v''(x) \leq 0$ for all $x \in \bigl[h(x^{\star}),1\bigr]$.
    \end{enumerate}
    
    These conditions imply that $I_{x^{\star}}$ is optimal for \labelcref{eqn:cons_pers_pb_CFI}. Furthermore, when the lower bound on informativeness becomes \emph{tighter} (i.e., if $\ubar{G}$ is replaced by $\tilde{\ubar{G}} \prec \ubar{G}$), the full revelation region becomes \emph{larger} (i.e., $x^{\star}$ becomes greater).
\end{proposition}

The proof mirrors the arguments in the proofs of \cref{prop:delegation_logconc,cor:delegation_compstat} in \cref{sec:delegation} on the delegation problem. We provide some details on this in \cref{secap:infodesign_sshaped_proof}.

The comparative statics result in \cref{lem:infodesign_sshaped} implies that more stringent informativeness constraints influence the amount of information that is revealed at the optimum in a non-trivial way. Crucially, a more stringent informativeness constraint not only influences how much information is revealed above the optimal censorship cutoff $x^{\star}$, but also influences its \emph{location}. In particular, higher informativeness standards lead to a \emph{larger} full-revelation region.

\section{Concluding remarks}

We studied \emph{convex function intervals} (CFIs), sets of one-dimensional convex functions that satisfy slope constraints and lie between two boundary functions. We characterized the extreme points of CFIs and provided sufficient optimality conditions for linear programming problems defined over them. These abstract results yield concrete insights across a range of economic design problems. In particular, we recover classical results as special cases and extend them to environments with additional constraints, including participation constraints in adverse selection, allocative constraints in contests, and informativeness restrictions in Bayesian persuasion.  

A broad class of fundamental economic problems, such as optimal income taxation \citep{Mirrlees1971}, monopolistic insurance \citep{Stiglitz1977}, or monopolistic screening with convex production costs \citep{Mussa1978}, can be formulated as \emph{convex programming problems} over CFIs, i.e., maximizing a \emph{concave} objective functional over a (compact and convex) CFI of indirect utility functions. Yet little is known about the solutions to variational problems subject to convexity constraint beyond their existence and regularity \citep{Lions1998,Carlier2001a,Carlier2001b,Carlier2008}. The only general methodology is that of \citet{Rochet1998}, which considers menus containing a unique outside option. Extending our analysis of type-dependent participation constraints to convex variational problems---where interior solutions may be optimal---appears to be a promising avenue for future research.

\bibliographystyle{ecta}
\bibliography{ref}
\clearpage

\appendix\label{secap:appendix}

\begin{center}
    \huge \textbf{Mathematical Appendix}
\end{center}

\section{Representation of CFIs and Main Theorems}

\subsection{Proof of \cref{prop:representation}}\label{secap:proof_representation}

Let $\mathcal{U}$ be a CFI. It is non-empty since $\ubar{u},\bar{u}\in\mathcal{U}$. We begin by showing the convexity of $\mathcal{U}$. Let $u_{1},u_{2}\in\mathcal{U}$ and let $\alpha\in[0,1]$. Therefore, $\alpha u_{1} + (1-\alpha) u_{2}$ is a convex function. Moreover, $\ubar{u}\leq \alpha u_{1} + (1-\alpha) u_{2} \leq \bar{u}$ and, by linearity of the subdifferential \citep[][Theorem 4.1.1]{Hiriart2001}, $\partial(\alpha u_{1} + (1-\alpha) u_{2})(x) = \alpha \partial u_{1}(x) + (1-\alpha) \partial u_{2}(x) \subseteq S$ for any $x\in X$. This implies that $\mathcal{U}$ is convex.

We now turn to compactness. By definition $\mathcal{U}$ is a subset of $\mathcal{C}$. We endow $\mathcal{U}$ with the supremum-norm $\lVert\cdot \rVert_{\infty}$. Any $u\in \mathcal{U}$ is bounded by $\max \{ \lVert \ubar{u} \rVert_{\infty}, \lVert \bar{u} \rVert_{\infty} \} $ (which implies that $\mathcal{U}$ is uniformly bounded), and is $K$-Lipschitz continuous since $\lvert u' \rvert \leq K \coloneqq \max\{\lvert\ubar{s}\rvert,\lvert\bar{s}\rvert\}$ (which implies that $\mathcal{U}$ is uniformly equicontinuous). Therefore, the Arzelà–Ascoli Theorem \citep[][Theorem 3]{Royden2010} implies that $\mathcal{U}$ is compact in the supremum norm. 

Since $\mathcal{U}$ is a compact and convex subset of the Banach space $(\mathcal{C},\lVert\cdot\rVert_{\infty})$, it is hence metrizable. Choquet's Theorem \citep[see][p.~14]{Phelps2001}  thus implies that any element of $\mathcal{U}$ can be represented by a probability measure supported in $\ex(\mathcal{U})$, and Proposition 1.2 in \cite{Phelps2001} implies that $\mathcal{U}$ is equal to the closed convex hull of its extreme points. \qed

\subsection{Proof of \cref{thm:ext_pt}}\label{secap:ext_pt_proof}

\subsubsection{Preliminaries}\label{secap:Bregman_perturb}

Let $u\in\mathcal{K}$. For any $x,y\in X$, we let
\begin{equation*}
  t_{u}(x ; y^{-}) = u(y) + \partial_{-}u(y)(x-y),
\end{equation*}
and
\begin{equation*}
    t_{u}(x ; y^{+}) = u(y) + \partial_{+}u(y)(x-y),
\end{equation*}
be the left (resp.~right) \emph{tangent segment} to $u$ at $y$ evaluated at $x$. Whenever $\partial_{-}u(y) = \partial_{+}u(y) = u'(y)$ we let $t_{u}(\, \cdot \, ; y)$ denote the unique tangent line to $u$ at $y$.

Moreover, for any $a,b\in X$ and $x\in [a,b]$, we let
\begin{equation*}
    c_{u}(x ; a, b) = u(a) + \Biggl(\frac{u(b)-u(a)}{b-a}\Biggr)(x-a),
\end{equation*}
be the \emph{chord segment} linking $u(a)$ to $u(b)$ evaluated at $x$.

Fix any compact interval $I=[a,b]\subseteq X$. We introduce the function $g_{a,b}\colon X \to \mathbb{R}$ defined by
\begin{equation}\label{eqn:Bregman_perturb}
    g_{a,b}(x)=\left\{
    \begin{array}{cc}
        0 & \text{if $x\notin I$} \\[5pt]
        u(x) - [t_{u}(x;a^{+}) \vee t_{u}(x;b^{-})] & \text{if $x\in I$}
    \end{array}
    \right..
\end{equation}
for each $x\in X$, where the symbol $\vee$ stands for the pointwise maximum operator.\footnote{When $u$ is a differentiable and strictly convex function on $\mathbb{R}$, the function $(x,y)\in \mathbb{R}^2 \mapsto u(x)-t_{u}(x;y)$ is known as the Bregman divergence of $u$. This function measures how much $u$ diverges from its first-order Taylor expansion around $y$ at any $x\in X$.}

Following Proposition B.1.2.1 in \cite{Hiriart2001}, for any $x,y\in X$ and $s\in \partial u(y)$, $u(x)\geq u(y)+s(x-y)$, with equality at $x=y$. Hence, $g_{a,b}\geq 0$. Moreover, $g_{a,b}(a)=g_{a,b}(b)=0$ which implies that $g_{a,b}\in\mathcal{C}$. Furthermore, for any $x\in X$,
\begin{equation*}
    (u-g_{a,b})(x)=\left\{
    \begin{array}{cc}
        u(x) & \text{if $x\notin I$} \\[5pt]
        t_{u}(x;a^{+}) \vee t_{u}(x;b^{-}) & \text{if $x\in I$}
    \end{array}
    \right.,
\end{equation*}
and,
\begin{equation*}
    (u+g_{a,b})(x)=\left\{
    \begin{array}{cc}
        u(x) & \text{if $x\notin I$} \\[5pt]
        2u(x) - [t_{u}(x;a^{+}) \vee t_{u}(x;b^{-})] & \text{if $x\in I$}
    \end{array}
    \right..
\end{equation*}

Since $u$ and $g_{a,b}$ are both continuous on $X$, $u\pm g_{a,b}$ are also continuous on $X$. Moreover, $u-g_{a,b}$ is convex on $X$, and thus belongs to $\mathcal{K}$. Note, however, that $u+g_{a,b}$ might fail to be convex.

For any real-valued function $f$ defined on $X$, let $\vex[f]$ denote its \emph{convexification}\footnote{The notion of convexification is also known in convex analysis as the \emph{closed-convex hull} of a function \cite[see][Definition B.2.5.3]{Hiriart2001}.} which, for each $x\in X$, is defined by
\begin{equation}\label{eqn:vex_def}
    \vex[f](x) = \sup \bigl\{ \, g(x) \; \vert \; \text{$g\colon X\to\mathbb{R}$ convex, $g\leq f$} \bigr\},
\end{equation}

For any $f\in \mathcal{C}$, $\vex[f]$ also belongs to $\mathcal{C}$ \citep[][Proposition B.2.5.2]{Hiriart2001} and, hence, belongs to $\mathcal{K}$. 

Finally, let $h_{a,b}\colon X \to \mathbb{R}$ be the function defined by
\begin{equation}\label{eqn:perturbation}
    h_{a,b}(x)=\vex[u+g_{a,b}](x) - u(x),
\end{equation}
for every $x\in X$. The function $h_{a,b}$ is continuous on $X$ (since $u+g_{a,b}$ is itself continuous). Moreover, $h_{a,b}\geq 0$: Since $u$ is a convex and $u\leq u+g_{a,b}$ we must have $\vex[u+g_{a,b}]\geq u$ by \labelcref{eqn:vex_def}.

We illustrate the construction of \labelcref{eqn:perturbation} on \cref{fig:Bregman_perturb} for a strictly convex and differentiable convex function $u$.
\begin{figure}[h]
    \centering
    \begin{subfigure}[t]{0.495\linewidth}
      \centering
      \includegraphics{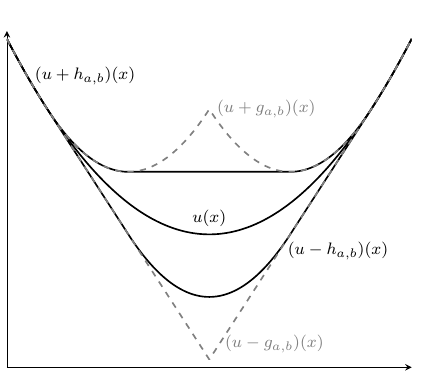}
      \caption{The functions $u$ and $u \pm h_{a,b}$.}
      \label{fig:convex_comb}
    \end{subfigure}
    \begin{subfigure}[t]{0.495\linewidth}
      \centering
      \includegraphics{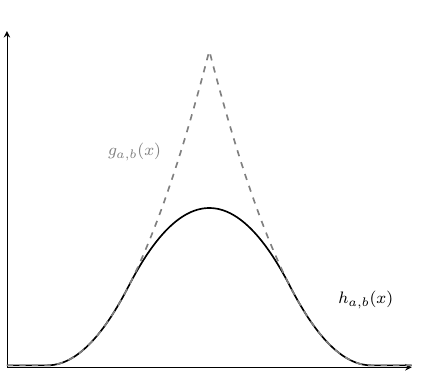}
      \caption{The perturbation $h_{a,b}$.}
      \label{fig:perturb}
    \end{subfigure}
    \caption{The perturbation $h_{a,b}$ for $u(x)=(x-\frac{1}{2})^2$ for all $x\in[0,1]$, and $a=0.1$ and $b=0.9$.}
    \label{fig:Bregman_perturb}
\end{figure}

We now prove that $h_{a,b}$ is not identically zero if and only if $u$ is not piecewise affine with at most two kinks.
\begin{lemma}\label{lem:perturbation}
  Let $u\in \mathcal{U}$. Assume there exists $I=[a,b]\subseteq X$ such that $\ubar{u}<u\vert_{I}<\bar{u}$, and let $h_{a,b}$ be defined according to \labelcref{eqn:perturbation}. Then, for any $\varepsilon\in[0,1]$, $u\pm \varepsilon h_{a,b}\in \mathcal{K}$ and $\partial(u\pm \varepsilon h_{a,b})(X)\subseteq S$. Furthermore, $h_{a,b}$ is identically zero on $X$ if and only if there exists three (not necessarily distinct) affine functions $\ell_{0}$, $\ell_{1}$ and $\ell_{2}$ defined on $I$ such that $u\vert_{I}=\ell_{0}\vee \ell_{1}\vee \ell_{2}$.
\end{lemma}
\begin{proof}[Proof of \cref{lem:perturbation}]
  Let $u\in \mathcal{U}$. Assume there exists $I=[a,b]\subseteq X$ such that $\ubar{u}<u\vert_{I}<\bar{u}$, and let $g_{a,b}$ and $h_{a,b}$ be respectively defined according to \labelcref{eqn:Bregman_perturb,eqn:perturbation}. For any $\varepsilon\in[0,1]$,
\begin{equation*}
  u+\varepsilon h_{a,b}=\varepsilon \vex[u+g_{a,b}]+(1-\varepsilon) u,
\end{equation*}
hence, $u+\varepsilon h_{a,b}\in \mathcal{K}$ as the convex combination between two continuous convex functions.

What remains to be proven is that, for any $\varepsilon\in [0,1]$, $(u-\varepsilon h_{a,b})\in\mathcal{K}$ and $\partial(u\pm \varepsilon h_{a,b})(X)\subseteq S$. We will prove simultaneously that $h_{a,b}$ is identically zero on $I$ if and only if there exists three (not necessarily distinct) affine functions $\ell_{0}$, $\ell_{1}$ and $\ell_{2}$ defined on $I$ such that, for any $x\in I$, $u\vert_{I}=\ell_{0}\vee \ell_{1}\vee \ell_{2}$. We proceed in two steps.

\textbf{Step 1:} We prove that $g_{a,b}$ is identically zero on $X$ if and only if $u$ can be written as the pointwise supremum of at most two affine functions on $I$. 

Assume that there exist (not necessarily distinct) affine functions defined on $I$ such that $u\vert_{I}=\ell_{0}\vee \ell_{1}$. If $\ell_{0}=\ell_{1}$ then $u\vert_{I}$ is affine and thus $u\vert_{I}=t_{u}(\cdot,a^{+})\vert_{I}=t_{u}(\cdot,b^{-})\vert_{I}$. Hence, $g_{a,b}(x)=0$ for all $x\in X$. If $\ell_{0}\neq \ell_{1}$, there exists a unique $\hat{x}\in \Int(I)$ such that $\ell_{0}(\hat{x})=\ell_{1}(\hat{x})$. Then, for all $x\in[a,\hat{x}]$, $\ell_{0}(x) = t_{u}(x;a^{+})$ and, for all $x\in[\hat{x},b]$, $\ell_{1}(x) = t_{u}(x;b^{-})$. Hence $g_{a,b}(x)=0$ for all $x\in X$.

Conversely, assume that $g_{a,b}(x)=0$ for all $x\in X$. Then, $u(x)=t_{u}(x; a) \vee t_{u}(x;b^{-})$ for all $x\in I$. Therefore, $u$ is the pointwise supremum of two affine functions, with $t_{u}(\, \cdot \, ; a^{+}) = t_{u}(\, \cdot \, ; b^{-})$ if and only if $u$ is affine on $X$.

\textbf{Step 2:} Step 1 implies that if $u$ can be written as the supremum of at most two affine functions on $I$, then $\vex[u+g_{a,b}] = \vex[u] = u$ so $h_{a,b}$ is identically null. Assume, therefore, that $g_{a,b}$ cannot be written as the pointwise supremum of at most two affine functions on $I$. We prove that there exists $[x_{0},x_{1}] \subseteq I$ (possibly equal to a singleton) such that (see \cref{fig:perturb} for a visualization):
\begin{equation}\label{eqn:perturb_tripartition}
    h_{a,b}(x) = \left\{
    \begin{array}{ll}
        0 & \text{if $x\in[0,a)$} \\[5pt]
        g_{a,b}(x) & \text{if $x\in[a,x_{0})$} \\[5pt]
        c_{u+g_{a,b}}(x; x_{0},x_{1}) - u(x) & \text{if $x\in[x_{0},x_{1}]$} \\[5pt]
        g_{a,b}(x) & \text{if $x\in(x_{1},b]$} \\[5pt]
        0 & \text{if $x\in(b,1]$}
    \end{array}
    \right..
\end{equation}

To do so we define $\phi_{a,b}=(u+g_{a,b})\vert_{I}$, and we let $\varphi_{a}(x)=2u(x) - t_{u}(x;a^{+})$ and $\varphi_{b}(x)=2u(x) - t_{u}(x;b^{-})$ for all $x\in I$. Since $u$ is convex and cannot be written as the maximum of two affine functions on $I$, the functions $\varphi_{a}$ and $\varphi_{b}$ are convex and non-affine functions on $I$. Furthermore, there must exist a unique $\hat{x}_{a,b}\in \Int(I)$ such that $t(\hat{x}_{a,b}; a^{+}) = t(\hat{x}_{a,b}; b^{-})$, and, for any $x\in I$,
\begin{equation*}
  \phi_{a,b}(x)=\left\{
        \begin{array}{ll}
            \varphi_{a}(x) & \text{if $x\in[a,\hat{x}_{a,b}]$,} \\
            \varphi_{b}(x) & \text{if $x\in[\hat{x}_{a,b},b]$.}
        \end{array}
    \right.
\end{equation*}

Furthermore, note that $\phi_{a,b}$ is convex if and only if $\partial_{+}u(\hat{x}_{a,b})-\partial_{-}u(\hat{x}_{a,b})\geq \bigl(\partial_{-}u(b)-\partial_{+}u(a)\bigr)/2$.

Importantly, for any $x\in I$, $\phi_{a,b}(x)=\varphi_{a}(x)\wedge \varphi_{b}(x)$, where $\wedge$ denotes the pointwise minimization operator. Hence, for any $x\in I$,
\begin{alignat}{2}\label{eq:vex_inf_chord}
    \vex[\phi_{a,b}](x) &= \vex[\varphi_{a}\wedge \varphi_{b}](x) \notag \\[5pt] 
    &=\inf \Bigl\{ \alpha \varphi_{a}(x_{0}) + (1-\alpha) \varphi_{b}(x_{1}) \; \big\vert \;&&\alpha\in[0,1], \; x_{0}, x_{1}\in I, \notag \\ 
    & && \alpha x_{0} + (1-\alpha) x_{1} = x
    \Bigr\},
\end{alignat}
where the second equality follows from Proposition B.2.5.4 in \cite{Hiriart2001}.

The convexity of $\varphi_a$ and $\varphi_b$ together with \labelcref{eq:vex_inf_chord} imply that $\vex[\phi_{a,b}]$ must have the following structure (see \cref{fig:convex_comb} for a visualization): there must exist an interval $[x_{0},x_{1}]\subseteq [a,b]$ containing $\hat{x}_{a,b}$---possibly equal to $\{\hat{x}_{a,b}\}$ in case $\phi_{a,b}$ is already convex\footnote{Which, we recall, is true if and only if $\partial_{+}u(\hat{x}_{a,b})-\partial_{-}u(\hat{x}_{a,b})\geq \bigl(\partial_{-}u(b)-\partial_{+}u(a)\bigr)/2$.}---such that
\begin{equation}\label{eqn:vex_phiab}
  \vex[\phi_{a,b}](x) = \left\{
  \begin{array}{ll}
      \varphi_{a}(x)  & \text{if $x\in[a,x_{0})$,} \\[5pt]
      \varphi_{a}(x_{0}) + \Biggl(\dfrac{\varphi_{b}(x_{1}) - \varphi_{a}(x_{0}) }{x_{1}-x_{0}}\Biggr) (x-x_{0})  & \text{if $x\in[x_{0},x_{1}]$,} \\[5pt]
      \varphi_{b}(x)  & \text{if $x\in(x_{1},b]$.}
  \end{array}
  \right.
\end{equation}

Note that $\varphi_{a}(a)=u(a)$ and $\varphi_{a}(b)=u(b)$. We can thus extend $\vex[\phi_{a,b}]$ by continuity to the whole domain $X$ by letting
\begin{equation*}
  \psi_{a,b}(x)= \left\{
    \begin{array}{ll}
        \vex[\phi_{a,b}](x)  & \text{if $x\in I$,} \\[5pt]
        u(x) & \text{if $x\notin I$.}
    \end{array}
    \right.
\end{equation*}

Since $\vex[\phi_{a,b}]$ is convex on $I$, this implies that $\partial_{+}\psi_{a,b}$ exists and is increasing on $I$. Remark now that $\partial_{+}\psi_{a,b}(a)=\partial_{+}\varphi_{a}(a)=\partial_{+}u(a)$ and $\partial_{-}\psi_{a,b}(b)=\partial_{-}\varphi_{b}(b)=\partial_{-}u(b)$. As $\partial_{+}\psi_{a,b}$ coincides with $\partial_{+}u$ outside of $I$, $\partial_{+}\psi_{a,b}$ is increasing on the whole domain $X$. Hence, $\psi_{a,b}$ is convex on $X$ \cite[this follows from][Theorem 6.4]{Hiriart2001} and, therefore, $\psi_{a,b}=\vex[u+g_{a,b}]$. This implies that $h_{a,b}$ takes the desired form \labelcref{eqn:perturb_tripartition}.

Moreover, by monotonicity of the left and right derivatives, we also have
\begin{equation*}
  \ubar{s}\leq \partial_{-}u(0) \leq \partial_{-}\vex[u+g_{a,b}]\leq \partial_{-}u(1)\leq  \bar{s},
\end{equation*}
and
\begin{equation*}
  \ubar{s}\leq \partial_{+}u(0) \leq \partial_{+}\vex[u+g_{a,b}]\leq \partial_{+}u(1)\leq \bar{s},
\end{equation*}
which implies that $\partial\vex[u+g_{a,b}](X) \subseteq \partial u(X) \subseteq S$. Now, take $\varepsilon\in [0,1]$. By linearity of the subdifferential
\begin{equation*}
  \partial(u+\varepsilon h_{a,b})(x) = \varepsilon \partial\vex[u+g_{a,b}](x) + (1-\varepsilon)\partial u(x) \subseteq S 
\end{equation*}
for all $x\in X$.

Next, note that given that $h_{a,b}$ has the form \labelcref{eqn:perturb_tripartition}, there must exist an interval $[x_{0},x_{1}]\subseteq I$ (possibly equal to a singleton) such that
\begin{equation}\label{eqn:piecewise_uminush}
    \begin{aligned}
      (u-h_{a,b})(x) = \left\{
      \begin{array}{ll}
          u(x) & \text{if $x\in [0,a)$,} \\[5pt]
          t_{u}(x;a^{+}) & \text{if $x\in[a,x_{0})$,} \\[5pt]
          2 u(x) - c_{u+g_{a,b}}(x; x_{0},x_{1}) & \text{if $x\in[x_{0},x_{1}]$,} \\[5pt]
          t_{u}(x;b^{-}) & \text{if $x\in(x_{1},b]$,} \\[5pt]
          u(x) & \text{if $x\in(b,1]$.}
      \end{array}
      \right.
    \end{aligned}
\end{equation}

The closed-form expression \labelcref{eqn:piecewise_uminush} shows that $u-h_{a,b}$ is piecewise convex: $u$ is convex by definition; $t_{u}(\,\cdot\,;a^{+})$ and $t_{u}(\,\cdot\,;b^{-})$ are convex because they are affine; and $2u-c_{u+g_{a,b}}(\,\cdot\,;x_{0},x_{1})$ is convex as the sum of a convex and an affine function. To verify global convexity of $u-h_{a,b}$ on $X$, it suffices to check that $\partial_{-}(u-h_{a,b}) \le \partial_{+}(u-h_{a,b})$ at the junctions between the pieces. This is immediate at $x\in\{a,b\}$. At $x_{0}$, the inequality $\partial_{-}(u-h_{a,b})(x_{0}) \le \partial_{+}(u-h_{a,b})(x_{0})$ is equivalent to
\begin{equation}\label{eqn:ineq_deriv_uminush_x0}
    \partial_{+}u(x_{0}) \ge \frac{\partial_{+}u(a)+m}{2},
\end{equation}
where
\begin{equation*}
    m \coloneqq \frac{\varphi_{b}(x_{1})-\varphi_{a}(x_{0})}{x_{1}-x_{0}}.
\end{equation*}

By construction of $h_{a,b}$ (see \labelcref{eqn:vex_phiab}), the tangency condition $m\in \partial \varphi_{a}(x_{0})$ holds, which implies the desired inequality \labelcref{eqn:ineq_deriv_uminush_x0}. A symmetric argument yields $\partial_{-}(u-h_{a,b})(x_{1}) \le \partial_{+}(u-h_{a,b})(x_{1})$. Therefore $u-h_{a,b}$ is convex on $X$. Moreover, \labelcref{eqn:piecewise_uminush} implies $\partial(u-h_{a,b})(x)=\partial u(x)$ for $x\in\{0,1\}$, hence $\partial(u-h_{a,b})(X)\subseteq \partial u(X)\subseteq S$. Consequently, for any $\varepsilon\in[0,1]$, the function $u-\varepsilon h_{a,b}=(1-\varepsilon)u+\varepsilon(u-h_{a,b})$ is convex (a convex combination of convex functions) and satisfies $\partial(u-\varepsilon h_{a,b})(X)\subseteq S$ (by linearity of the subdifferential).

We end the proof by showing that $h_{a,b}$ is identically zero if and only if $u$ can be written as the pointwise maximum of at most three affine curves on $I$. We have already done the cases where $u$ can be written as the pointwise supremum of one or two affine pieces in Step 1. Therefore, assume that there exists three distinct affine functions $\ell_{0}$, $\ell_{1}$, $\ell_{2}$ defined on $X$ such that $u\vert_{I}=\ell_{0}\vee \ell_{1} \vee \ell_{2}$. Hence, there exists $[x_{0},x_{1}]\subset I$ such that, for all $x\in [x_{0},x_{1}]$, $\ell_{1}(x)= c_{\ell_{0}\vee\ell_{2}}(x; x_{0},x_{1})$. Let $\hat{x}$ be the unique solution to $\ell_{0}(x)=\ell_{2}(x)$. We have that, for all $x\in[a,\hat{x}]$, $\ell_{0}(x)=t_{u}(x;a^{+})$ and that, for all $x\in[\hat{x},1]$, $\ell_{2}(x)=t_{u}(x;b^{-})$. \cref{eqn:perturb_tripartition} therefore implies that, for any $x\in X$,
\begin{equation*}
  h_{a,b}(x) = \left\{
  \begin{array}{ll}
      0 & \text{if $x\in[0,x_{0})$} \\[5pt]
      \ell_{1}(x)-\ell_{0}(x) & \text{if $x\in[x_{0},\hat{x}]$} \\[5pt]
      \ell_{1}(x)-\ell_{2}(x) & \text{if $x\in[\hat{x},x_{1}]$} \\[5pt]
      0 & \text{if $x\in(x_{1},1]$}
  \end{array}
  \right..
\end{equation*}
and, hence,
\begin{align*}
  (u+h_{a,b})(x) &= \left\{
  \begin{array}{ll}
      \ell_{0}(x) & \text{if $x\in[0,x_{0})$} \\[5pt]
      \ell_{1}(x) & \text{if $x\in[x_{0},\hat{x}]$} \\[5pt]
      \ell_{1}(x) & \text{if $x\in[\hat{x},x_{1}]$} \\[5pt]
      \ell_{2}(x) & \text{if $x\in(x_{1},1]$}
  \end{array}
  \right. \\[5pt]
  &= u(x),
\end{align*}
which implies that $h_{a,b}=0$. 

Conversely, assume now that $h_{a,b}=0$. \cref{eqn:perturb_tripartition} implies that $u(x)=t_{u}(x;a^{+})$ for all $x \in[a,x_{0})$, that $u(x)=c_{u+g_{a,b}}(x; x_{0},x_{1})$ for all $x \in[x_{0},x_{1}]$, and that $u(x)=t_{u}(x;b^{-})$ for all $x \in(x_{0},b]$. Hence, $u$ is a continuous piecewise affine function with at most three pieces. This ends the proof of \cref{lem:perturbation}.

\end{proof}

\subsubsection{The proof}

\begin{proof}[Proof of \cref{thm:ext_pt}]
\framebox{$\Rightarrow$}: Let $\mathcal{U}$ be a CFI. By \cref{prop:representation}, there exists $u\in\ex(\mathcal{U})$. The proof of necessity proceeds in two steps: in Step 1, we show that there exists a countable collection $\mathcal{X}=\{X_{n}\}_{n\in\mathbb{N}}$ of non-degenerate disjoint intervals $X_{n}=[a_{n},b_{n}]\subseteq X$ such that $u\vert_I$ is affine for all $I \in \mathcal{X}$. Step 2 shows that, for any $I\in\mathcal{X}$, $u\vert_{I}$ satisfies at least one of the Conditions \labelcref{cond:tangent_bounds}-\labelcref{cond:boundary_cases}.

\textbf{Step 1:} Fix an arbitrary $x\in X$, and suppose that $\ubar{u}(x)<u(x)<\bar{u}(x)$. By continuity of $u$, there must exist a non-degenerate compact interval $[a,b]\subseteq X$ such that $\ubar{u}<u\vert_{[a,b]}<\bar{u}$. Moreover, assume that $u\vert_{[a,b]}$ is non-affine, i.e., that $\partial_+u(a) < \partial_-u(b)$. 

Let $h_{a,b}$ be defined according to \labelcref{eqn:Bregman_perturb}. It follows from \cref{lem:perturbation} that, for any $\varepsilon\in(0,1]$, $u\pm \varepsilon h_{a,b}\in \mathcal{K}$, $\partial(u\pm \varepsilon h_{a,b})(X)\subseteq S$ and that $\varepsilon h_{a,b}\neq 0$ if and only if $u\vert_{[a,b]}$ cannot be written as the maximum of at most three affine functions. Note also that $(u\pm \varepsilon h_{a,b})_{\varepsilon\in(0,1]}$ converges uniformly to $u\in\mathcal{U}$ as $\varepsilon\to 0^{+}$. Hence, there must exist $\varepsilon>0$ such that $\ubar{u}<u\pm \varepsilon h_{a,b}<\bar{u}$, so $u\pm \varepsilon h_{a,b}\in \mathcal{U}$. Since $\varepsilon h_{a,b}\neq 0$, $u$ cannot be an extreme point. 

Therefore, for any $x\in X$, either $u(x)\in \{\ubar{u}(x),\bar{u}(x)\}$, or $u(x)\notin \{\ubar{u}(x),\bar{u}(x)\}$ and there exists a non-degenerate maximal interval $I$ containing $x$ such that $\ubar{u}<u\vert_{\Int(I)}<\bar{u}$ and $u\vert_{I}$ can be written as the supremum of at most three affine functions with slopes in $S$.\footnote{The interval $I$ being maximal means that there exists no larger interval $J\supset I$ where $u$ satisfies the same properties.} The set of kinks\footnote{A convex function $u$ is said to admit a kink at $x$ if $\partial u(x)$ admits more than one element \cite[][Definition D.2.1.6]{Hiriart2001}.} of a convex function being countable \citep[see][p.~28]{Niculescu2025}, this implies the existence of a countable collection $\mathcal{X}=\{X_{n}\}_{n\in\mathbb{N}}$ of non-degenerate disjoint intervals $X_{n}=[a_{n},b_{n}]$ where, for each $n\in\mathbb{N}$, $u\vert_{X_{n}}$ is affine and $\ubar{u}<u\vert_{\Int(X_{n})}<\bar{u}$. Moreover, if there exists $n\in\mathbb{N}$ such that $u(a_n) \notin \{\ubar{u}(a_n),\bar{u}(a_n)\}$, then there exists $m\in\mathbb{N}$ such that $b_m=a_n$ (and symmetrically for $u(b_n) \notin \{\ubar{u}(b_n),\bar{u}(b_n)\}$).

\textbf{Step 2:} We are now going to show that if there exists $n\in\mathbb{N}$ such that $u\vert_{X_{n}}$ satisfies none of the conditions \labelcref{cond:tangent_bounds,cond:chord_satur,cond:matching_grad,cond:boundary_cases} in \cref{thm:ext_pt}, then $u\notin\ex(\mathcal{U})$. Let $X_{n} \in \mathcal{X}$ and assume that $u\vert_{X_{n}}$ does not satisfy conditions \labelcref{cond:tangent_bounds,cond:chord_satur,cond:matching_grad,cond:boundary_cases}. There are two cases.

\begin{enumerate}[label=\itshape Case \arabic*,leftmargin=*]
    \item There exists $n\in\mathbb{N}$ such that $a_{n}>0$ and $b_{n}<1$. This and the assumption that $u\vert_{X_{n}}$ does not satisfy conditions \labelcref{cond:tangent_bounds,cond:chord_satur,cond:matching_grad,cond:boundary_cases} implies that  $\ubar{s}<u'\vert_{X_{n}}<\bar{s}$ and that for at least one $x\in \{a_n,b_n\}$, it holds that $u(x) \notin \{\ubar{u}(x),\bar{u}(x)\}$. Without loss of generality, assume that $u(b_n) \notin \{\ubar{u}(b_n), \bar{u}(b_n)\}$. Hence, there exists $m\in\mathbb{N}$ such that $b_n=a_m$ and $u\vert_{X_{m}}$ is affine by Step 1. However, since $u\vert_{X_{n}}$ fails to satisfy condition \labelcref{cond:chord_satur}, $u\vert_{X_{m}}$ does not satisfy \labelcref{cond:tangent_bounds}. We distinguish three subcases: (1) $b_m<1$, (2) $b_m=1$ and $u'\vert_{X_{m}}<\bar{s}$, and (3) $b_m=1$ and $u'\vert_{X_{m}}=\bar{s}$.     

    \begin{enumerate}[label=\itshape Subcase 1.\arabic*]
    
        \item First, assume that $b_m<1$. Define the continuous perturbation $h_1$ as
        \begin{equation}\label{eqn:lin_perturb_right}
            h_1(x)=\frac{x-a_n}{b_n-a_n}\Ind_{X_{n}}(x) + \frac{b_m-x}{b_m-a_m}\Ind_{X_{m}}(x),
        \end{equation}
        for each $x\in X$. For any $\varepsilon>0$, $(u\pm\varepsilon h_1)\in \mathcal{C}$ by continuity of $u$ and $h_1$. Furthermore, the maximality of $X_{n}$ and $X_{m}$, convexity of $\ubar{u}$ and the fact that both $u\vert_{X_{n}}$ and $u\vert_{X_{m}}$ fail to satisfy condition \labelcref{cond:tangent_bounds}, imply that $\ubar{u}<u\vert_{(a_n, b_m]}<\bar{u}$ and $\partial_{-}u(a_n)<u'\vert_{X_{n}}<u'\vert_{X_{m}}<\partial_{+}u(b_m)$. Therefore, if $h_1$ is defined by \labelcref{eqn:lin_perturb_right}, one can always find $\varepsilon>0$ sufficiently small such that $\ubar{u}<(u\pm\varepsilon h_1)\vert_{X_{n}\cup X_{m}}<\bar{u}$, and that
         \begin{equation*}
           \partial_{-}u(a_n) \leq  u'\vert_{X_{n}} \pm \frac{\varepsilon}{b_n-a_n} \leq u'\vert_{X_{m}} \pm \frac{\varepsilon}{b_m-a_m} \leq \partial_{+} u(b_m),
        \end{equation*}
        which implies $u\pm\varepsilon h_1 \in \mathcal{U}$, contradicting $u\in\ex(\mathcal{U})$. \label{cond:necessity_subcase11}
        
        \item Second, consider the case $b_m=1$ and $u'\vert_{X_{m}}<\bar{s}$. In this case, take again the continuous perturbation $h_1$ as defined in \labelcref{eqn:lin_perturb_right}. Again, for any $\varepsilon>0$, $(u\pm\varepsilon h_1)\in \mathcal{C}$ by continuity of $u$ and $h_1$. Since $u\vert_{X_{m}}$ does not satisfy condition \labelcref{cond:tangent_bounds}, if $u(b_m)=\bar{u}(b_m)$, then $u'\vert_{X_{m}} > \partial_{-}\bar{u}(b_m)$ (and $u'\vert_{X_{m}}<\bar{s}$ by assumption). This implies that $\ubar{u}<u\vert_{(a_n, b_m)}<\bar{u}$ and $\partial_{-}u(a_n)<u'\vert_{X_{n}}<u'\vert_{X_{m}}<\partial_{-}\bar{u}(b_m)$. Therefore, if $h_1$ is defined by \labelcref{eqn:lin_perturb_right}, one can always find $\varepsilon>0$ sufficiently small such that $\ubar{u}<(u\pm\varepsilon h_1)\vert_{(a_n, b_m)}<\bar{u}$, $u'\vert_{X_{m}} \pm \frac{\varepsilon}{b_m-a_m} \leq \partial_{-} \bar{u}(b_m)$ and that
         \begin{equation*}
          \partial_{-}u(a_n) \leq  u'\vert_{X_{n}} \pm \frac{\varepsilon}{b_n-a_n} \leq u'\vert_{X_{m}} \pm    \frac{\varepsilon}{b_m-a_m} <\bar{s},
         \end{equation*}
        which implies $u\pm\varepsilon h_1 \in \mathcal{U}$, contradicting $u\in\ex(\mathcal{U})$. \label{cond:necessity_subcase12}

        \item Last, assume $b_m=1$ and $u'\vert_{X_{m}}=\bar{s}$. Since $u\vert_{X_{m}}$ does not satisfy condition \labelcref{cond:tangent_bounds}, this means that $u(b_m) < \bar{u}(b_m)$. In that case, we use the continuous perturbation
         \begin{equation}\label{eqn:lin_perturb_right_max_slope}
            h_2(x)=\frac{x-a_n}{b_n-a_n}\Ind_{X_{n}}(x)+\Ind_{X_{m}}(x),
         \end{equation}
         for any $x\in X$. For any $\varepsilon>0$, $(u\pm\varepsilon h_2)\in \mathcal{C}$ by continuity of $u$ and $h_2$. Furthermore, the maximality of $X_{n}$ and $X_{m}$ and the fact that $u\vert_{X_{m}}$ fails to satisfy condition \labelcref{cond:tangent_bounds} imply that $\ubar{u}<u\vert_{(a_n, b_m]}<\bar{u}$ and $\partial_{-}u(a_n)<u'\vert_{X_{n}}<u'\vert_{X_{m}}=\bar{s}$. Therefore, if $h_2$ satisfies \labelcref{eqn:lin_perturb_right_max_slope}, one can always find $\varepsilon>0$ sufficiently small such that $\ubar{u}<(u\pm\varepsilon h_2)\vert_{(a_n, b_m]}<\bar{u}$, and that
        \begin{equation*}
          \partial_{-}u(a_n) \leq  u'\vert_{X_{n}} \pm \frac{\varepsilon}{b_n-a_n} \leq \bar{s},
         \end{equation*}
        again a contradiction to $u\in\ex(\mathcal{U})$. \label{cond:necessity_subcase13}
        
    \end{enumerate}

     Therefore, we can conclude that if there exists $n\in\mathbb{N}$ such that $X_{n}\in\mathcal{X}$ with $0<a_n<b_n<1$, then $u\vert_{X_{n}}$ necessarily satisfies conditions \labelcref{cond:tangent_bounds} or \labelcref{cond:chord_satur}. \label{cond:necessity_case1}

    \item Second, $a_{n}=0$ or $b_n=1$. Consider the case where $b_n=1$. The case $a_{n}=0$ is symmetric and therefore omitted. We have assumed that $u\vert_{X_{n}}$ satisfies none of the conditions \labelcref{cond:tangent_bounds,cond:chord_satur,cond:matching_grad,cond:boundary_cases}. We again distinguish multiple subcases. The arguments are similar to those presented in \labelcref{cond:necessity_case1} above. We therefore indicate the perturbations used below and refer to \labelcref{cond:necessity_case1} for the exact arguments. 
    
    \begin{itemize}
    
        \item If $u'\vert_{X_{n}}=\bar{s}$, then, since $u\vert_{X_{n}}$ does not satisfy conditions \labelcref{cond:tangent_bounds} and \labelcref{cond:matching_grad}, it holds that $u(a_n)\notin \{\ubar{u}(a_n), \bar{u}(a_n)\}$. Thus, by Step 1, there exists an adjacent interval $X_{m}$ with $b_m=a_n$ if $a_n>0$. Again, since $u\vert_{X_{n}}$ does not satisfy condition \labelcref{cond:matching_grad}, $u\vert_{X_{m}}$ does not satisfy condition \labelcref{cond:tangent_bounds}. Using the same perturbation $h_2$ as defined in \labelcref{eqn:lin_perturb_right_max_slope} (where we flip the labels $m$ and $n$) and applying the same arguments as in \labelcref{cond:necessity_subcase13}, this contradicts $u\in \ex(\mathcal{U})$.

        If $a_n=0$, then the perturbation $h_3(x)= \Ind_{X_{n}}(x)$, again applying the arguments from \labelcref{cond:necessity_case1}, contradicts $u\in \ex(\mathcal{U})$.
        
        \item This works symmetrically for $u'\vert_{X_{n}}=\ubar{s}$. Therefore, assume $u'\vert_{X_{n}} \in (\ubar{s}, \bar{s})$ from now on.
        
        \item If $u(b_n)=\bar{u}(b_n)$, then since $u\vert_{X_{n}}$ does not satisfy \labelcref{cond:tangent_bounds,cond:boundary_cases}, $u(a_n) \notin \{\ubar{u}(a_n), \bar{u}(a_n)\}$ has to hold. By Step 1, there is an interval $X_{m}\in\mathcal{X}$ such that $b_m=a_n$. Again, since $u\vert_{X_{n}}$ does not satisfy \labelcref{cond:boundary_cases}, $u\vert_{X_{m}}$ does not satisfy \labelcref{cond:tangent_bounds}. Using the same perturbation $h_1$ as defined in \labelcref{eqn:lin_perturb_right} (where we flip the labels $m$ and $n$) and applying the same arguments as in \labelcref{cond:necessity_subcase12}, this contradicts $u\in \ex(\mathcal{U})$.
        
        \item If $u(b_n) \notin \{\ubar{u}(b_n), \bar{u}(b_n) \}$, then, since $u\vert_{X_{n}}$ does not satisfy \labelcref{cond:tangent_bounds,cond:boundary_cases}, $u(a_n) \notin \{\ubar{u}(a_n), \bar{u}(a_n)\}$ has to hold. Therefore, the perturbation $h_3(x)= \Ind_{X_{n}}(x)$, again applying the arguments from \labelcref{cond:necessity_case1}, shows that $u\notin \ex(\mathcal{U})$.
        
        \item If $u(b_n)=\ubar{u}(b_n)$, then, since $u\vert_{X_{n}}$ does not satisfy \labelcref{cond:tangent_bounds,cond:chord_satur}, $u(a_n) \notin \{\ubar{u}(a_n), \bar{u}(a_n)\}$ has to hold. By Step 1, there is an interval $X_{m}\in\mathcal{X}$ such that $b_m=a_n$. Again, since $u\vert_{X_{n}}$ does not satisfy \labelcref{cond:boundary_cases}, $u\vert_{X_{m}}$ does not satisfy \labelcref{cond:tangent_bounds}. Using the same perturbation $h_1$ as defined in \labelcref{eqn:lin_perturb_right} (where we flip the labels $m$ and $n$) and applying the same arguments as in \labelcref{cond:necessity_subcase11}, this contradicts $u\in \ex(\mathcal{U})$.
        
    \end{itemize}

     Therefore, if $u\in\ex(\mathcal{U})$ and there exists $n\in\mathbb{N}$ such that $X_{n}\in\mathcal{X}$ with $a_n=0$ or $b_n=1$ (or both) then $u\vert_{X_{n}}$ must satisfy one of the conditions \labelcref{cond:tangent_bounds,cond:chord_satur,cond:matching_grad,cond:boundary_cases} in \cref{thm:ext_pt}. This concludes the proof of necessity.
\end{enumerate}

\framebox{$\Leftarrow$}: Let $u\colon X\to\mathbb{R}$ satisfy conditions \labelcref{cond:level_satur,cond:not_level_satur} in \cref{thm:ext_pt}, and let $\mathcal{X}=\{X_{n}\}_{n\in\mathbb{N}}$ be the corresponding collection of intervals defined in \cref{thm:ext_pt}. It is easily seen that $u\in\mathcal{U}$. Therefore, it is legitimate to prove that $u\in\ex(\mathcal{U})$. Take $h\in\mathcal{C}$, $h\neq 0$, and suppose, by way of contradiction, that $u\pm h\in\mathcal{U}$.

As a preliminary, we prove some regularity properties that $h$ must satisfy. Since $u\pm h\in\mathcal{U}$, $h=\hat{u}-u$ for some $\hat{u}\in\mathcal{U}$. Therefore, $h$ must be $K$-Lipschitz continuous (hence, absolutely continuous) as the difference between two $K$-Lipschitz continuous functions, with $K=\max\{\lvert\ubar{s}\rvert,\vert\bar{s}\rvert\}$. Hence, $h$ admits a derivative $h'$ almost everywhere in $X$ with $\lvert h'\rvert \leq K$, and $h(x)=h(a)+\int_{a}^{x} h'(s) \, \diff s$ for any $a,x\in X$ \citep[see][Theorems C.2 and C.3]{Yeh2014}.

Assume first that $h(x)>0$ for some $x\notin \bigcup_{n\in\mathbb{N}}X_{n}$. Then, either $u(x)-h(x)<\ubar{u}(x)$ or $u(x)+h(x)>\bar{u}(x)$, contradicting $u\pm h\in\mathcal{U}$. Hence, $h(x)=0$ at all $x\notin\bigcup_{n\in\mathbb{N}}X_{n}$.

Suppose now that $h(x_0)>0$ for some $x_0\in X_{n} \in \mathcal{X}$. Let us consider first the case where $u$ satisfies condition \labelcref{cond:tangent_bounds} on $X_{n}$. Consider the case where $u(a_n)=\bar{u}(a_n)$. By the above, $h(a_n)=0$. Since $h(x_0)=\int_{a_n}^{x_0} h'(s) \, \diff s>0$ there must exist $\hat{x}\in [a_n,x_{0}]$ such that $h'(\hat{x})>0$. If $u'\vert_{X_{n}} \notin \{\ubar{s}, \bar{s}\}$, this implies that $(u-h)'(\hat{x})<\bar{u}'(a_n)=(u-h)'(a_n)$, which contradicts the convexity of $u-h$ since $\hat{x}>a_n$. If $u'\vert_{X_{n}}=\ubar{s}$, then $(u-h)'(\hat{x})<\ubar{s}$, contradicting $u-h \in \mathcal{U}$ because of the slope constraints. The argument is symmetric if $u(b_n)=\bar{u}(b_n)$. Hence, $h\vert_{X_{n}}=0$ for all $X_{n}\in \mathcal{X}$ where $u$ satisfies condition \labelcref{cond:tangent_bounds}. 

Next, assume that $u$ satisfies condition \labelcref{cond:chord_satur} on $X_{n}$. Hence, for each $x\in\{a_n,b_n\}$, either there exists an interval adjacent to $X_{n}$ at $x$ such that $u$ satisfies condition \labelcref{cond:tangent_bounds} on that interval, or $u(x)=\ubar{u}(x)$. Note that the previous arguments imply that, if $x\in\{a_n,b_n\}$, then $h(x)=0$, since $h(x)=0$ for all $x\notin\bigcup_{n\in\mathbb{N}}X_{n}$ and $h\vert_{I}=0$ for all $I\in\mathcal{X}$ where $u$ satisfies condition \labelcref{cond:tangent_bounds}. Moreover, observe that, as an affine function, $u\vert_{X_{n}}$ is the pointwise largest convex function connecting $u(a_n)$ and $u(b_n)$. Since $h(x_0)>0$, this implies that $(u+h)\vert_{I}$ is not convex, a contradiction. Hence, $h\vert_{I}=0$ for all $I\in\mathcal{X}$ such that $u\vert_{I}$ satisfies condition \labelcref{cond:chord_satur}. The argument is analogous if $u\vert_{X_{n}}$ satisfies condition \labelcref{cond:boundary_cases}.

It therefore remains to treat the case that $u\vert_{X_{n}}$ satisfies condition \labelcref{cond:matching_grad}. Without loss of generality assume $u'\vert_{X_{n}}=\bar{s}$. By the previous arguments, $h(a_n)=0$. Since $h(x_0)=\int_{a_n}^{x_0} h'(s) \, \diff s>0$ there must exist $\hat{x}\in [a_n,x_{0}]$ such that $h'(\hat{x})>0$. Thus $(u+h)'(\hat{x})=\bar{s}+h'(\hat{x})>\bar{s}$ which implies that $u+h\notin \mathcal{U}$, a contradiction.
\end{proof}

\subsection{Proof of \cref{thm:charac_opt}} \label{secap:proof_charac_opt}

Let $\mathcal{U}$ be a CFI with domain $X\coloneqq[0,1]$, let $u^{\star}\in \ex(\mathcal{U})$, and let $\mu\in\mathcal{M}(X)$. To prove \cref{thm:charac_opt}, we begin by finding an equivalent primal problem to \labelcref{eqn:Lin_Pb}. Then we set up its dual problem, applying results from infinite dimensional conic linear programming \citep{Shapiro2001}. Next, we use conditions (i)-(v) to construct Lagrange multipliers that certify the optimality of $u^{\star}$.

\subsubsection{Primal Problem}

We refer to \labelcref{eqn:Primal} as the primal problem:
\begin{equation}\label{eqn:Primal}
    \begin{aligned}
        \max_{u \in \mathcal{K}} & \int_{X} u \,  \diff \mu \\
        \text{s.t.} & u \leq \bar{u}\\
        & u \geq \ubar{u}\\
        & u \geq u(1)-\bar{s}(1-\Id)\\
        & u \geq u(0)+\ubar{s}\Id
    \end{aligned}
    \tag{$\mathrm{P}_{\mathcal{U}, \mu}$}
\end{equation}

The next lemma shows that \labelcref{eqn:Primal} and \labelcref{eqn:Lin_Pb} are equivalent.

\begin{lemma}\label{lem:equiv_primal_lin}
    A function $u$ is feasible for \labelcref{eqn:Primal} if and only if $u\in \mathcal{U}$.
\end{lemma}
\begin{proof}
    Feasibility of $u\in \mathcal{U}$ for \labelcref{eqn:Primal} is immediate. For the other direction, let $u$ be a function that is feasible for \labelcref{eqn:Primal}. The only aspect that has to be shown explicitly is that the subgradient constraint is satisfied, i.e., that $\partial u(x) \subseteq [\ubar{s}, \bar{s}]$ for all $x \in [0,1]$. Since $u$ is convex, $\partial_{+}u(0)$ is bounded below by $\ubar{s}$ and $\partial_{-}u(1)$ is bounded above by $\bar{s}$. We can therefore always extend $u$ linearly to the whole $\mathbb{R}$. We denote this linear extension by $\tilde{u}$. In particular, $\tilde{u}$ is differentiable at $0$ and $1$ and $\partial\tilde{u}(x)=\partial u(x)$ for all $x \in (0,1)$. This implies the following inequalities: 
    \begin{equation*}
        \ubar{s} \leq \partial_{+}u(0) = \partial_{+}\tilde{u}(0) \leq \partial_{-}\tilde{u}(x) = \partial_{-}u(x) \leq \partial_{+}u(x) = \partial_{+}\tilde{u}(x) \leq \partial_{+}\tilde{u}(1)=\partial_{-}u(1)\leq\bar{s},
    \end{equation*}
    for all $x\in [0,1]$. Because $\partial_{-}u(x)$ and $\partial_{+}u(x)$ are non-decreasing, this implies that $\partial u(X) \subseteq S$ by Theorem 2.1.2 in \cite{Niculescu2025}.
\end{proof}

\subsubsection{Dual Problem and Weak Duality}\label{secap:dual_weakduality}

We now state the dual problem \labelcref{eqn:Dual} to \labelcref{eqn:Primal} and show that weak duality holds. For the construction of the dual we refer to \cref{secap:dual_problem}. Let $\mathcal{M}_{+}(X)$ denote the set of finite positive Radon measures on $X=[0,1]$.

The dual problem \labelcref{eqn:Dual} is given by
\begin{equation}\label{eqn:Dual}
    \begin{aligned}
        \min_{(\gamma_i)_{i\in\{1,\dots,4\}}\in (\mathcal{M}_{+}(X))^{4}} & \; \int_{X} \bar{u} \, \diff \gamma_{1} - \int_{X} \ubar{u} \, \diff \gamma_{2} + \bar{s}\int_{X} (1-\Id) \, \diff \gamma_{3} - \ubar{s}\int_{X} \Id \, \diff \gamma_{4} \\
        \text{s.t.} & \; \gamma_{1}-\gamma_{2} - \gamma_{3} + \delta^{\gamma_{3}}_1 - \gamma_{4} + \delta^{\gamma_{4}}_0 \geq_{\mathrm{cx}} \mu
    \end{aligned}
     \tag{$\mathrm{D}_{\mathcal{U}, \mu}$}
\end{equation}
where $\delta_1^{\gamma_{3}}$ (resp.~$\delta_0^{\gamma_{4}}$) denotes a point mass at $1$ of mass $\gamma_{3}(X)$ (resp.~at $0$ of mass $\gamma_{4}(X)$).

Weak duality between \labelcref{eqn:Primal} and \labelcref{eqn:Dual} is easy to establish: Let $u\in \mathcal{K}$ be feasible for \labelcref{eqn:Primal} and $(\gamma_{1}, \gamma_{2}, \gamma_{3}, \gamma_{4}) \in \mathcal{M}_+(X)^4$ be feasible for \labelcref{eqn:Dual}. Then
\begin{align*}
	\int_{X} u \, \diff \mu & \leq \int_{X} u \, \diff \bigl[\gamma_{1} - \gamma_{2} - \gamma_{3} + \delta_1^{ \gamma_{3} }-\gamma_{4} + \delta_0^{ \gamma_{4} }\bigr] \\
	& \leq  \int_{X} \bar{u} \, \diff \gamma_{1} - \int_{X} \ubar{u} \, \diff \gamma_{2}+ \bar{s} \int_{X} (1-\Id) \, \diff \gamma_{3}- \ubar{s}\int_{X} \Id \, \diff \gamma_{4},
\end{align*}
where the first inequality follows because $ \gamma_{1}-\gamma_{2} - \gamma_{3} + \delta^{\gamma_{3}}_1 - \gamma_{4} + \delta^{\gamma_{4}}_0 \geq_{\mathrm{cx}} \mu$ by assumption, and the second inequality by feasibility of $u$ for \labelcref{eqn:Primal}. For completeness, we show in \cref{secap:strongduality} that strong duality holds as well, but this is not needed for the proof.

\subsubsection{Construction of Multipliers} \label{sec:proof_thm_opt_suff}

We can now prove that conditions (i)-(v) given in \cref{thm:charac_opt} are  sufficient for optimality. For convenience, we use the shorthand notation $\cup\mathcal{Y}'\coloneqq\bigcup_{Y \in \mathcal{Y}'} Y$ for all $\mathcal{Y}'\subset \mathcal{Y}$.

We define four dual Lagrange multipliers $(\gamma_i)_{i\in\{1,\dots,4\}} \in \mathcal{M}_+(X)^{4}$ such that, for each $i\in \{ 1, \dots, 4 \} $, $\supp(\gamma_{i})=\cup \mathcal{Y}_i$.

For each $x \in X$, we let $Y(x)$ be the element of $\mathcal{Y}$ such that $x \in Y(x)$. Moreover, let $\lvert\mu\rvert\coloneqq \mu^{+}+\mu^{-}$ be the total variation of $\mu$. We split $\mathcal{Y}_0$ into $\mathcal{Y}_0^+ \coloneqq\Bigl\{ Y \in \mathcal{Y}_0 \; \big\vert \; \mu|_Y(Y) \geq 0 \Bigr\}$ and $\mathcal{Y}_0^- \coloneqq\Bigl\{ Y \in \mathcal{Y}_0 \; \big\vert \; \mu|_Y(Y) < 0 \Bigr\}$.

We first observe that $\int_{\cup \mathcal{Y}_{5}}\int_{X} u^{\star}(s) \, \diff \mu\vert_{Y(x)}(s) \, \diff \lvert\mu\rvert(x)=0$ by $\mu\vert_{Y} \leq_{\mathrm{cx}}\mathbf{0}$ and $u\vert_{Y}$ being affine for $Y \in \mathcal{Y}_{5}$.

\paragraph{Construction of $\bm{\gamma_1}$}

We construct the multiplier $\gamma_{1}$ by aggregating the point masses $\delta_{Y}$ for $Y \in \mathcal{Y}_1\cup \mathcal{Y}_0^+$ according to $\lvert\mu\rvert$. That is, for any Borel measurable $A\subseteq X$,
\begin{equation*}
    \gamma_{1}(A)\coloneqq \int_{\cup \mathcal{Y}_{1} \bigcup \cup \mathcal{Y}_0^+}\delta_{Y(x)}(A) \, \diff \lvert\mu\rvert(x).
\end{equation*}

By assumption, $\mu\vert_{Y}(Y) \geq 0$ for $Y \in \mathcal{Y}_{1}$, thus $\gamma_{1} \in \mathcal{M}_+(X)$. Furthermore,
\begin{align*}
    \int_{ \cup \mathcal{Y}_{1} \bigcup \cup \mathcal{Y}_0^+ }u^{\star} \, \diff \mu &= \int_{ \cup \mathcal{Y}_{1} \bigcup \cup \mathcal{Y}_0^+ }\int_{X} u^{\star}(s) \, \diff \mu\vert_{Y(x)}(s) \, \diff \lvert\mu\rvert(x) \\
    &=  \int_{ \cup \mathcal{Y}_{1} \bigcup \cup \mathcal{Y}_0^+ }\int_{X} u^{\star}(s) \, \diff \delta_{R(x)}(s) \, \diff \lvert\mu\rvert(x) \\
    &= \int_{X} u^{\star} \, \diff \gamma_{1} \\
    &= \int_{X} \bar{u} \, \diff \gamma_{1},
\end{align*}
where the second equality follows from $\mu\vert_{Y} \leq_{ \mathrm{cx} }\delta_{Y}$ and $u^{\star}\vert_{Y}$ being affine for $Y \in \mathcal{Y}_{1}$, the third equality follows from the definition of $\gamma_{1}$, and the last equality follows from $u^{\star}$ and $\bar{u}$ coinciding on the support of $\gamma_{1}$.

\paragraph{Construction of $\bm{\gamma_2}$}

Next, we construct $\gamma_{2}$ by aggregating $\lvert\mu\rvert\bigr\vert_{Y}$ on $\mathcal{Y}_{2} \cup \mathcal{Y}_0^-$ according to $\lvert\mu\rvert$. For any Borel measurable $A\subseteq X$, let
\begin{equation*}
    \gamma_{2}(A)\coloneqq \int_{\cup \mathcal{Y}_{2} \bigcup \cup \mathcal{Y}_0^-} \lvert\mu\rvert\bigr\vert_{Y(x)}(A)\, \diff \lvert\mu\rvert(x).
\end{equation*}

By definition, $\gamma_{2} \in \mathcal{M}_+(X)$. Additionally,
\begin{align*}
    \int_{ \cup \mathcal{Y}_{2} \bigcup \cup \mathcal{Y}_0^- }u^{{\footnotesize\star}} \, \diff \mu &= \int_{ \cup \mathcal{Y}_{2} \bigcup \cup \mathcal{Y}_0^- }\int_{X} u^{\star}(s) \, \diff \mu\vert_{Y(x)}(s) \, \diff \lvert\mu\rvert(x) \\
    &= - \int_{ \cup \mathcal{Y}_{2} \bigcup \cup \mathcal{Y}_0^- }\int_{X} u^{\star}(s) \, \diff \mu\vert_{Y(x)}(s) \, \diff \lvert\mu\rvert(x) \\
    &= - \int_{X} \ubar{u} \, \diff \gamma_{2}
\end{align*}
where the second equality follows from $\mu\vert_{Y}(Y) \leq 0$ for $Y \in \mathcal{Y}_{2}$, and the last equality follows from the definition of $\gamma_{2}$ and $u^{\star}$ and $\ubar{u}$ coinciding in its support.

\paragraph{Construction of $\bm{\gamma_3}$}

We now turn to $\gamma_{3}$. Let $Y$ be the unique element of $\mathcal{Y}_{3}$ if it exists.\footnote{Note that $\lvert\mathcal{Y}_{3}\rvert \leq 1$ because $u^{\star}$ is convex.} We apply different versions of Strassen's Theorem to construct the dual multiplier $\gamma_{3}$ from $\mu\vert_{Y}^{+}$

Let $W$, $Z$ be two random variables distributed according to $\mu\vert_{Y}^{+}$, $\mu\vert_{Y}^{-}$ respectively. Then, since $\mu\vert_{Y}^{+} \leq_{\mathrm{icx}}\mu\vert_{Y}^{-}$, Theorem 4.A.5 in \cite{Shaked2007} implies the existence of two random variables $\hat{W}$, $\hat{Z}$ that have the same distributions than $W$ and $Z$, respectively, and that satisfy $\mathbb{E}[\hat{Z} \mid \hat{W}]\leq \hat{W}$ almost surely.

Let $\alpha$ be the probability measure according to which the conditional expectation $\mathbb{E}[\hat{Z} \mid \hat{W}]$ is distributed. The law of iterated expectations implies that $\mathbb{E}\bigl[\hat{Z}\mid \mathbb{E}[\hat{Z} \mid \hat{W}]\bigr]=\mathbb{E}[\hat{Z} \mid \hat{W}]$ almost surely. Applying Theorem 3.A.4 in \cite{Shaked2007}, we obtain that $\alpha\leq_{\mathrm{cx}}\mu\vert_{Y}^{-}$. Furthermore, $\mu\vert_{Y}^{+} \leq_{\mathrm{dcx}} \alpha$ must hold, since $\mathbb{E}\bigl[\mathbb{E}[\hat{Z}\mid\hat{W}] \mid  \hat{W}\bigr]=\mathbb{E}[\hat{Z} \mid \hat{W}]\leq \hat{W}$ almost surely.

Therefore, we can apply Theorem 2.6.1 in \cite{Muller2002}, which implies the existence of a transition kernel $\kappa\colon \mathcal{B}(Y)\times Y \to [0,1]$ such that, for all $A \in \mathcal{B}(Y)$, $\alpha(A) = \int_{Y} \kappa(A,x) \, \diff \mu\vert_{Y}^{+} (x)$, and $\int_{Y} s \, \kappa(\diff s,x) \leq x$ for all $x \in Y$. This last property implies that
\begin{equation*}
\forall x \in Y, \quad \beta(x) \coloneqq \frac{1-x}{1-\int_{Y} s \, \kappa(\diff s,x)} \leq 1.
\end{equation*} 

We now define a new transition kernel $\tau\colon \mathcal{B}(Y)\times Y\to [0,1]$ that, for each $x \in Y$, places a weight of $1-\beta(x)$ on a point mass at $1$ and a weight of $\beta(x)$ on the probability measure $\kappa(\, \cdot \,, x)$. Formally,
\begin{equation*}
\forall A \in \mathcal{B}(Y), \; \forall x \in Y, \quad \tau(A, x) \coloneqq \bigl(1-\beta(x)\bigr) \delta_1(A) + \beta(x) \kappa(A,x).
\end{equation*}

Note that, for all $x \in Y$, $\tau(\, \cdot \,, x)$ is a probability measure on $(Y, \mathcal{B}(Y))$ and that, for all $A\in \mathcal{B}(Y)$, $\tau(A, \, \cdot \,)$ is a measurable function. Hence, $\tau$ is indeed a transition kernel. Furthermore, we show that the transition kernel $\tau$ is a \emph{dilation} \cite[for a definition, see][Section 15]{Phelps2001}:
\begin{align*}
\forall x \in Y, \quad\int_{Y} s \, \tau(\diff s, x) = \bigl(1-\beta(x)\bigr)\times 1 + \beta(x) \int_{Y} s \, \kappa(\diff s, x) = x.
\end{align*}

Let $\eta (A) \coloneqq \int_{Y} \tau(A, x) \, \diff \mu\vert_{Y}^{+} (x)$. Jensen's inequality implies that 
\begin{align*}
\forall u\in \mathcal{K}, \; \forall x \in Y, \quad \int_{Y} u(s) \, \tau(\diff s,x) \geq u\biggl(\int_{Y} s \, \tau(\diff s, x)\biggr)=u(x).
\end{align*}

Applying Theorem 2.6.1 from \cite{Muller2002} again implies that
$\eta \geq_{\mathrm{cx}} \mu\vert_{Y}^{+}$. Combining this observation with the relationship $\alpha \leq_{\mathrm{cx}}\mu\vert_{Y}^{-}$ from above, we obtain:
\begin{equation*}
    -\alpha + \eta \geq_{\mathrm{cx}} - \alpha + \mu\vert_{Y}^{+} \geq_{\mathrm{cx}} -\mu\vert_{Y}^{-} + \mu\vert_{Y}^{+} = \mu\vert_{Y}.
\end{equation*}

Let us define the multiplier $\gamma_{3}$ as follows:
\begin{equation*}
    \forall A\in\mathcal{B}(Y), \quad \gamma_{3}(A)\coloneqq \int_{A} \int_{Y} \bigl(1-\beta(x)\bigr) \kappa(A,x) \, \diff \mu\vert_{Y(s)}^{+}(x) \, \diff \lvert\mu\rvert(s)
\end{equation*}
where, for all $x \in Y$, we extend $\kappa(\cdot, x)$ to the measurable space $(X, \mathcal{B}(X))$. For each $A \in \mathcal{B}(Y)$, we thus have
\begin{align*}
    -\alpha(A) + \eta(A) &= \int_{Y} \bigl[\tau(A,x)-\kappa(A,x)\bigr] \, \diff \mu\vert_{Y}^{+} (x) \\
    &= \delta_1(A)\int_{Y} \bigl(1-\beta(x)\bigr) \, \diff \mu\vert_{Y}^{+}(x) - \int_{Y} \bigl(1-\beta(x)\bigr) \kappa(A,x) \, \diff \mu\vert_{Y}^{+} (x) \\
    &= \delta_1^{ \gamma_{3} }\vert_{Y}(A)-\gamma_{3}\vert_{Y}(A).
\end{align*}

This implies that $\gamma_{3}\in \mathcal{M}_+(X)$, and that
\begin{align*}
    \delta_1^{ \gamma_{3} }\vert_{Y}-\gamma_{3}\vert_{Y} \geq_{\mathrm{cx}} \mu\vert_{Y}.
\end{align*}

Therefore,
\begin{align*}
    \int_{ \cup \mathcal{Y}_{3} }u^{\star} \, \diff \mu &= \int_{ \cup \mathcal{Y}_{3} }\int_{X} u^{\star}(s) \, \diff \mu\vert_{Y(x)}(s) \, \diff \lvert\mu\rvert(x) \\
    &=  \int_{ \cup \mathcal{Y}_{3} }\int_{X} u^{\star}(s) \, \diff \bigl[\delta_1^{ \gamma_{3} }\vert_{Y}-\gamma_{3}\vert_{Y}\bigl](s) \, \diff \lvert\mu\rvert(x) \\
    &= \int_{X} u^{\star} \, \diff \bigl[\delta_1^{ \gamma_{3} }-\gamma_{3}\bigl] \\
    &= \int_{X} \bigl(u^{\star} - u^{\star}(1)\bigl) \, \diff \bigl[\delta_1^{ \gamma_{3} }-\gamma_{3}\bigl] \\
    &= \bar{s}\int_{X} (1-\Id) \, \diff \gamma_{3}
\end{align*}
where the second equality follows from $u^{\star}\vert_{Y}$ being affine and $\delta_1^{ \gamma_{3} }\vert_{Y}-\gamma_{3}\vert_{Y} \geq_{\mathrm{cx}} \mu\vert_{Y}$  and the last equality from the slope constraint binding for $Y \in \mathcal{Y}_{3}$.

\paragraph{Construction of $\bm{\gamma_4}$}

The construction of $\gamma_{4}$ is similar to that of $\gamma_3$ and thus omitted for brevity.

\subsubsection{Certification of optimality}

Having defined the Lagrange multipliers $(\gamma_i)_{i\in\{1,\dots,4\}} \in \mathcal{M}_+(X)^{4}$, it remains to show that they are feasible for \labelcref{eqn:Dual}, and that they certify the optimality of $u^{\star}$. Feasibility follows from the respective property on every element of the partition. We show optimality of $u^{\star}$ by showing that the dual multipliers achieve the value of \labelcref{eqn:Primal}, using the properties derived above:
\begin{align*}
	\int_{X} u^{\star} \, \diff \mu  = & \int_{\cup \mathcal{Y}_{1} \bigcup \cup \mathcal{Y}_0^+} u^{\star} \, \diff \mu + \int_{\cup \mathcal{Y}_{2} \bigcup \cup \mathcal{Y}_0^-} u^{\star} \, \diff \mu + \int_{\cup \mathcal{Y}_{3}} u^{\star} \, \diff \mu + \int_{\cup \mathcal{Y}_{4}} u^{\star} \, \diff \mu + \int_{\cup \mathcal{Y}_{5}} u^{\star} \, \diff \mu \\
	 =  &  \int_{X} \bar{u} \, \diff \gamma_{1} - \int_{X} \ubar{u} \, \diff \gamma_{2}+ \bar{s} \int_{X} (1-\Id) \, \diff \gamma_{3}- \ubar{s}\int_{X} \Id \, \diff \gamma_{4}
\end{align*}
This completes the proof that $u^{\star}$ is optimal.

\subsection{Proof of \cref{thm:opt_lower_bd}} \label{proof:thm:opt_lower_bd}

\paragraph{\labelcref{cond:opt_lb_afffinebound_0}}
Let $\ubar{u}^{\star} \in \mathcal{K}_{\bar{u}, S}^0$ be strictly convex. Since $\mu$ is well-behaved for $\mathcal{K}_{\bar{u}, S}^0$, there is an extreme point $u^\star \neq \bar{u}$ that achieves $\max_{u \in \mathcal{U}_{\ubar{u}^\star}} \int_X u \, \diff \mu$ (i.e., that is optimal for $\mathrm{LP}_{\mathcal{U}_{\ubar{u}^\star},\mu}$). 

By the arguments in \cref{secap:concavification}, there exists a partition $\mathcal{Y}=\bigcup_{i=1}^{5}\mathcal{Y}_i$ with the properties described in \cref{thm:charac_opt}. Since the CFI $\mathcal{U}_{\ubar{u}^\star}$ is affinely bounded with $\bar{u}'=\bar{s}$ and $u^\star \neq \bar{u}$, $\mathcal{Y}_1=\mathcal{Y}_4=\varnothing$ (i.e., $u_{\ubar{u}}$ does not coincide with $\bar{u}$ nor does it satisfy lower slope saturation). Therefore, we have $\mathcal{Y}=\mathcal{Y}_0 \cup \mathcal{Y}_2 \cup \mathcal{Y}_3 \cup \mathcal{Y}_5$. We first collect some properties of the partition $\mathcal{Y}$:
\begin{itemize}
    \item By affine boundedness and $\bar{u}'=\bar{s}$, $\mathcal{Y}_0= \bigl\{\{0\}\bigl\}$.
    \item By definition, the elements of $\mathcal{Y}_2$ are singletons.
    \item If $\mathcal{Y}_3$ is nonempty, it contains only a single partition element (by convexity of $u^\star$). In this case, write $\mathcal{Y}_3 = \bigl\{ [a^*, 1]\bigl\}$.
    \item By a slight abuse of notation, write $Y=[a_{Y}, b_{Y}]$ for $Y \in \mathcal{Y}_5$. Since $\mathcal{U}_{\ubar{u}^\star}$ is affinely bounded, $Y \in \mathcal{Y}_5$ means that $u^\star\vert_{Y}(x)=\ubar{u}^\star(a_{Y}) + \frac{x-a_{Y}}{b_{Y}-a_{Y}}(\ubar{u}^\star(b_{Y})-\ubar{u}^\star(a_{Y}))$.
\end{itemize}

For any $\ubar{u}\in \mathcal{K}_{\bar{u}, S}^0$, we now define a function $u_{\ubar{u}}$ as follows:
\begin{equation*}
    u_{\ubar{u}}\vert_{Y}(x)=\left\{\begin{matrix*}[l] 
    \ubar{u}(x) & & R= \bigl\{\{x\}\bigl\} \in \mathcal{Y}_0 \cup \mathcal{Y}_2 \\
    \ubar{u}(a^\star)+\bar{s}(x-a^\star) & & Y \in \mathcal{Y}_3  \\
    \ubar{u}(a_{Y}) + \frac{x-a_{Y}}{b_{Y}-a_{Y}}(\ubar{u}(b_{Y})-\ubar{u}(a_{Y})) & & Y \in \mathcal{Y}_5
    \end{matrix*} \right.
\end{equation*}

By the joint properties of $\mathcal{Y}$ and $\mu$, ensured by optimality of $u^\star$ for $\mathrm{LP}_{\mathcal{U}_{\ubar{u}^\star},\mu}$, $u_{\ubar{u}}$ is optimal for $\mathrm{LP}_{\mathcal{U}_{\ubar{u}},\mu}$ by \cref{thm:charac_opt} (i.e., $u_{\ubar{u}} \in \underset{u \in \mathcal{U}_{\ubar{u}}}{\argmax} \; \int_{X} u \, \diff  \mu$)\footnote{See also \cref{secap:concavification}}.

Let $\alpha \in (0,1)$ and $\ubar{u}_1, \ubar{u}_2 \in \mathcal{K}_{\bar{u}, S}^0$. By definition of $u_{\ubar{u}}$,
\begin{equation*}
    u_{\alpha \ubar{u}_1 + (1-\alpha) \ubar{u}_2} = \alpha u_{\ubar{u}_1} + (1-\alpha) u_{\ubar{u}_2}.
\end{equation*}

Optimality of $u_{\ubar{u}}$ for $\mathrm{LP}_{\mathcal{U}_{\ubar{u}},\mu}$ for any $\ubar{u}\in \mathcal{K}_{\bar{u}, S}^0$ therefore ensures that $\ubar{u} \in \mathcal{K}_{\bar{u}, S}^0 \mapsto u_{\ubar{u}}$ is linear. Combining this with linearity of $V$ in $u_{\ubar{u}}$ shows that $V\vert_{\mathcal{K}_{\bar{u}, S}^0}$ is linear as well.

\paragraph{\labelcref{cond:opt_lb_afffinebound_1}}
This works analogously to the previous case.

\paragraph{\labelcref{cond:opt_lb_concave}}
Let $\alpha \in (0,1)$ and $\ubar{u}_1, \ubar{u}_2 \in \mathcal{K}_{\bar{u}, S}$. Then there exist extreme points $u_{\ubar{u}_i} \in \underset{u \in \mathcal{U}_{\ubar{u}_i}}{\argmax} \; \int_{X} u \, \diff  \mu$ for $i \in \{1,2\}$. Since $\nu=\mu$, we get that
\begin{align*}
    \alpha V(\ubar{u}_1) + (1-\alpha) V(\ubar{u}_2) &= \alpha \int_X u_{\ubar{u}_1} \, \diff \mu + (1-\alpha) \int_X u_{\ubar{u}_2} \, \diff \mu \\
    &= \int_X \bigl[ \alpha u_{\ubar{u}_1} + (1-\alpha) u_{\ubar{u}_2} \bigl] \, \diff \mu \\
    &\leq \int_X u_{\alpha \ubar{u}_1 + (1-\alpha) \ubar{u}_2} \, \diff \mu \\
    &= V\bigl(\alpha \ubar{u}_1 + (1-\alpha) \ubar{u}_2\bigr),
\end{align*}
where the inequality follows because $\alpha u_{\ubar{u}_1} + (1-\alpha) u_{\ubar{u}_2}$ is not necessarily an extreme point. Therefore, $V$ is concave.

\section{Other Omitted Proofs in the Main Text}

\subsection{Proof of \cref{prop:opt_rev_cutoff}}\label{secap:opt_screening_proofs}

Suppose that $F$ is Myerson-regular. That is, the virtual value function defined by
\begin{equation*}
    v(\theta) = \theta - \frac{1-F(\theta)}{f(\theta)},
\end{equation*}
for every $\theta\in \Theta$, is non-decreasing. For each $\theta\in \Theta$, let
\begin{align*}
    \Psi_{\mathrm{R}}(\theta)
    &= \mu_{\mathrm{R}}\bigl([\theta,1]\bigr) \\[5pt]
    &= \psi_{\mathrm{R}}(1)+\int_{\theta}^{1} \psi_{\mathrm{R}}(t) \, \diff t \\[5pt]
    &= \theta f(\theta) - \bigl(1-F(\theta)\bigr) \\[5pt]
    &= f(\theta) v(\theta),
\end{align*}
for every $\theta\in \Theta$. Hence, $\Psi_{\mathrm{R}}$ has the same sign as $v$. We have $v(0)=-1/f(0)<0$ and $v(1)=1$. Moreover, since $v$ is continuous and nondecreasing, the Intermediate Value Theorem implies that $v$ crosses zero from below at a unique $\theta^{\star}\in (0,1)$. Hence $\Psi_{\mathrm{R}}$ is single-crossing from below at $\theta^{\star}$.

Let $\mathcal{Y}$ be the partition of $\Theta$ induced by the extreme point $u_{\theta^{\star}}$ of $\mathcal{U}_{\mathrm{S}}$ defined according to \labelcref{eqn:cutoff_utility}. We have $\mathcal{Y}_{0}=\{0\}$, $\cup_{Y\in \mathcal{Y}_{2}} Y=(0,\theta^{\star})$, and $\mathcal{Y}_{3}=[\theta^{\star},1]$, while $\mathcal{Y}_{i}=\varnothing$ for all $i\in\{1,4,5\}$. By \cref{thm:charac_opt}, it thus suffices to show that $\psi_{\mathrm{R}}(\theta)\leq 0$ for all $\theta \in (0,\theta^{\star})$ and that $\mu_{\mathrm{R}}\vert_{[\theta^{\star},1]}^{+} \leq_{\mathrm{dcx}}\mu_{\mathrm{R}}\vert_{[\theta^{\star},1]}^{-}$ to verify the optimality of $u_{\theta^{\star}}$.

Note that $\Psi_{\mathrm{R}}(\theta^{\star})=0$ is equivalent to
\begin{equation}\label{eqn:mass_equality_rev}
    \mu_{\mathrm{R}}\bigl([\theta^{\star},1]\bigr)=0.
\end{equation}

Next, for every $\theta\in \Theta$,
\begin{align*}
    \int_{[\theta,1]} t \, \diff\mu_{\mathrm{R}}(t)
    &= \psi_{\mathrm{R}}(1)+\int_{\theta}^{1} t\psi_{\mathrm{R}}(t) \, \diff t \\[5pt]
    &= \theta^2 f(\theta) \geq 0,
\end{align*}
where the second equality follows by integration by parts.  
In particular,
\begin{equation}\label{eqn:mean_inequality_rev}
    \int_{[\theta^{\star},1]} \theta \, \diff\mu_{\mathrm{R}}(\theta) \geq 0.
\end{equation}

By the Hahn–Jordan Theorem \citep[][Theorem 10.21]{Yeh2014}, there exists a unique decomposition $\mu_{\mathrm{R}}=\mu_{\mathrm{R}}^{+}-\mu_{\mathrm{R}}^{-}$. The equality \labelcref{eqn:mass_equality_rev} is thus equivalent to  
\begin{equation*}
    \mu_{\mathrm{R}}^{+}\bigl([\theta^{\star},1]\bigr)=\mu_{\mathrm{R}}^{-}\bigl([\theta^{\star},1]\bigr)\neq 0,
\end{equation*}
so $\mu_{\mathrm{R}}\vert_{[\theta^{\star},1]}^{+}$ and $\mu_{\mathrm{R}}\vert_{[\theta^{\star},1]}^{-}$ have equal nonzero mass and can thus be normalized to probability measures. Moreover, \labelcref{eqn:mean_inequality_rev} shows that the barycenter of $\mu_{\mathrm{R}}\vert^{+}_{[\theta^{\star},1]}$ is weakly larger than that of $\mu_{\mathrm{R}}\vert^{-}_{[\theta^{\star},1]}$:
\begin{equation*}
    \int_{[\theta^{\star},1]} \theta \, \diff\mu_{\mathrm{R}}\vert^{+}_{[\theta^{\star},1]}(\theta)
    \geq
    \int_{[\theta^{\star},1]} \theta \, \diff\mu_{\mathrm{R}}\vert^{-}_{[\theta^{\star},1]}(\theta).
\end{equation*}

By Theorem~4.A.2 of \citet{Shaked2007},  
$\mu_{\mathrm{R}}\vert_{[\theta^{\star},1]}^{+} \leq_{\mathrm{dcx}}\mu_{\mathrm{R}}\vert_{[\theta^{\star},1]}^{-}$  
is then equivalent to the following \emph{weak majorization} condition:
\begin{align*}
    \forall \theta \in [\theta^{\star},1], &\quad 
    \int_{\theta}^{1} \max\{0,\Psi_{\mathrm{R}}(t)\} \, \diff t
    \;\geq\; \int_{\theta}^{1} \max\{0,-\Psi_{\mathrm{R}}(t)\} \, \diff t , \\[5pt]
    \iff \forall \theta \in [\theta^{\star},1], &\quad 
    \int_{\theta}^{1} \Psi_{\mathrm{R}}(t) \, \diff t \;\geq\; 0,
\end{align*}
which holds because $\Psi_{\mathrm{R}}(\theta) \geq 0$ for all $\theta \in [\theta^{\star},1]$.

Finally, we let
\begin{equation*}
    \theta_{0}\coloneqq \inf\bigl\{\theta\in (0,1) \; \vert \; \psi_{\mathrm{R}}(\theta)>0\bigr\},
\end{equation*}
and adopt the convention that $\theta_{0}\coloneqq1$ if $\psi_{\mathrm{R}}(\theta)\leq 0$ for all $\theta\in \Theta$. Since $\lim_{\theta \to 0^{+}} \psi_{\mathrm{R}}(\theta) =-2f(0)<0$, we have $\theta_{0}>0$ and $\psi_{\mathrm{R}}(\theta) \leq 0$ for all $\theta\in [0, \theta_{0}]$. Therefore, if $\theta_{0}=1$ it is immediate that $\theta^{\star}<\theta_{0}$ and $\psi_{\mathrm{R}}(\theta) < 0$ for all $\theta \in (0,\theta^{\star})$. Suppose that $\theta_{0}<1$. Note that
\begin{equation*}
    \psi_{\mathrm{R}}(\theta) = -f'(\theta) v(\theta) - f(\theta) v'(\theta),
\end{equation*}
for every $\theta\in \Theta$. Therefore, $\psi_{\mathrm{R}}(\theta^{\star}) = - f(\theta^{\star}) v'(\theta^{\star}) \leq 0$. This implies $\theta^{\star}<\theta_{0}$ and, therefore, $\psi_{\mathrm{R}}(\theta) \leq  0$ for all $\theta \in (0,\theta^{\star})$. \qed

\subsection{Omitted proofs in \cref{sec:majorization}}

\subsubsection{Proof of \cref{lem:two-sided-maj}}\label{secap:two_sided_maj_proof}

Throughout the proof, we fix $f,g\in\mathcal{F}$ satisfying $f\succsim g$. Let $I_{f}$ and $I_{g}$ be defined according to \labelcref{eqn:integral_maj}, and let $\mathcal{I}_{f,g}$ be defined according to \labelcref{eqn:maj_CFI}. 

For any $\varphi\in\mathcal{F}$, the function $I_{\varphi}$ is convex and continuous \citep[][Section 1.5]{Niculescu2025}, and satisfies $I_{\varphi}(0)=-m_{\varphi}$, $I_{\varphi}(1)=0$, and $\partial I_{\varphi} \subseteq \mathrm{Im}(\varphi)\subseteq\bigl[\varphi(0),\varphi(1)\bigr]$ since $\varphi$ is non-decreasing. In particular, $\partial I_{f} \subseteq \bigl[f(0),f(1)\bigr]$.

\framebox{$\Rightarrow$} Take $\varphi \in \mathcal{F}$ and suppose that $f\succsim \varphi \succsim g$. We show that $I_{\varphi}\in\mathcal{I}_{f,g}$, where $I_{\varphi}$ is defined according to \labelcref{eqn:integral_maj}. 

First, we have already established that $I_{\varphi}\in \mathcal{K}$. Furthermore, the majorization condition $f\succsim \varphi \succsim g$ implies $I_{f}(x) \leq I_{\varphi}(x) \leq I_{g}(x)$ for every $x\in (0,1)$, and $I_{f}(x) = I_{\varphi}(x) = I_{g}(x)$ for $x\in\{0,1\}$. Since $I_{f}(0) = I_{\varphi}(0) = I_{g}(0) = 0$, we have $m_f=m_\varphi=m_g$, which gives us
\begin{equation}\label{eqn:mean_zero_x}
    \frac{1}{x}\int_{0}^{x} f(s) \, \mathrm{d}s \leq \frac{1}{x}\int_{0}^{x} \varphi(s) \, \mathrm{d}s \leq \frac{1}{x}\int_{0}^{x} g(s) \, \mathrm{d}s,
\end{equation}
for any $x\in(0,1]$.

By Lebesgue's differentiation theorem for increasing functions \citep[][Theorem 12.10]{Yeh2014}, taking the limit as $x \to 0^{+}$ in \labelcref{eqn:mean_zero_x} yields $f(0) \leq f(0^{+}) \leq \varphi(0^{+}) \leq g(0^{+})$. Moreover, for every $x\in [0,1)$, we have
\begin{equation}\label{eqn:mean_x_one}
    \frac{1}{1-x}\int_{x}^{1} f(s) \, \mathrm{d}s \geq \frac{1}{1-x}\int_{x}^{1} \varphi(s) \, \mathrm{d}s \geq \frac{1}{1-x}\int_{x}^{1} g(s) \, \mathrm{d}s. 
\end{equation}

Taking the limit as $x \to 1^{-}$ in \labelcref{eqn:mean_x_one}, we obtain $g(1^{-}) \leq \varphi(1^{-}) \leq f(1^{-})\leq f(1)$, which implies $\partial I_{\varphi} \subseteq \bigl[f(0),f(1)\bigr]$. This also shows that $\partial I_{g} \subseteq \bigl[f(0),f(1)\bigr]$, confirming that $\mathcal{I}_{f,g}$ satisfies \cref{def:CFI}.

\framebox{$\Leftarrow$} Take $I\in \mathcal{I}_{f,g}$. We show that there exists $\varphi\in \mathcal{F}$ such that $f\succsim \varphi \succsim g$ and $I=I_{\varphi}$, where $I_{\varphi}$ is defined according to \labelcref{eqn:integral_maj}. 

Let $\varphi$ be any selection of the subdifferential of $u$. Since $u$ is convex, such a selection always exists and is a non-decreasing function on $X$ \citep[][Corollary 2.1.3]{Niculescu2025}. Moreover, by \citep[][Theorem 1.5.2]{Niculescu2025}, we have $\varphi\in L^{1}(X)$ (hence $\varphi\in \mathcal{F}$) and
\begin{equation}\label{eqn:int_rep_cvx}
    \forall x\in X, \quad I(x)=I(0)+\int_{0}^{x} \varphi(s) \, \mathrm{d}s.
\end{equation}

Since $I(1) = 0$, evaluating \labelcref{eqn:int_rep_cvx} at $x=1$ yields $I(0)=-m_{\varphi}$. Therefore, using the definitions of $I_{f}$ and $I_{g}$, the representation of $I$ given by \labelcref{eqn:int_rep_cvx}, and the fact that $I\in \mathcal{I}_{f,g}$, we obtain
\begin{equation*}
    \forall x\in X, \quad \int_{0}^{x}f(s) \, \diff s - m_{f} \leq \int_{0}^{x} \varphi(s) \, \diff s - m_{\varphi} \leq \int_{0}^{x}g(s) \, \diff s - m_{g}.
\end{equation*}

Since $I_{f}(0)=I(0)=I_{g}(0)$, we obtain $m_{f}=m_{\varphi}=m_{g}$. Therefore, $f\succsim \varphi \succsim g$.
    
\qed

\subsubsection{Formalization of the claim in \cref{rmk:isomorphism}}\label{secap:bijection_ext_points}

Let $\Phi_{f,g} \coloneqq \bigl\{\varphi \in \mathcal{F} \; \vert \; \text{$\varphi$ right-continuous and $f\succsim \varphi \succsim g$}\bigr\}$. Since all functions in $\Phi_{f,g}$ share the same mean value $m\in\mathbb{R}$, we define the linear transformation $T\colon \Phi_{f,g}  \to \mathcal{I}_{f,g}$ by
\begin{equation*}
    T[\varphi](x) = \int_{0}^{x} \varphi(s)  \, \mathrm{d}s - m,
\end{equation*}
for all $\varphi\in \Phi_{f,g}$ and $x\in [0,1]$.

The mapping $T$ is an order-isomorphism between the partially ordered sets $(\Phi_{f,g}, \succsim)$ and $(\mathcal{I}_{f,g}, \leq)$, with (linear) inverse $T^{-1}$ defined by $T^{-1}[I](x)=\partial_{+} I(x)$ for all $I\in \mathcal{I}_{f,g}$ and all $x\in X$.\footnote{Lemma 2 in \cite{Curello2024} establishes this result for mean-preserving contractions. Their argument extends naturally to the setting of two-sided (weak) majorization. The restriction to right-continuous functions simply disciplines what can happen at discontinuities of $I$, and ensures uniqueness of the inverse mapping.} Since it is linear, the transformation $T$ preserves the convex structure of $\Phi_{f,g}$. Furthemore, when we endow $\Phi_{f,g}$ with the $L^1$ norm and $\mathcal{I}_{f,g}$ with the supremum norm, both sets become compact and both $T$ and $T^{-1}$ are continuous.\footnote{Compactness of $\Phi_{f,g}$ follows from Helly's selection theorem \citep[see][proof of Proposition 1, p.~1579]{Kleiner2021}. Continuity of $T^{-1}$ follows from the fact that if $(I_{n})_{n\in\mathbb{N}}$ is a sequence of convex functions converging uniformly to some continuous function $I$ on a compact interval, then $I$ is convex and the sequence $(\partial_{+}I_{n})_{n\in\mathbb{N}}$ converges pointwise to $\partial_{+}I$ \citep[see][Theorem 1.4.4]{Niculescu2025}.} Hence, $T$ is a homeomorphism, and thus also preserves compactness.

As a result, there is a one-to-one correspondence between the \emph{extreme points} of $\Phi_{f,g}$ and those of $\mathcal{I}_{f,g}$:
\begin{lemma}\label{lem:bij_ext_pt}
    $\ex(\mathcal{I}_{f,g}) = T\bigl(\ex(\Phi_{f,g})\bigr)$.
\end{lemma}
\begin{proof}[Proof of \cref{lem:bij_ext_pt}]
    The proof relies on the following result:
    \begin{lemma}[Affine mapping lemma, \citealp{Yang2025}] \label{lem:affine_mapping}
        Let $K$ be a compact convex subset of a locally convex topological vector space, and let $E$ be a topological vector space. For any continuous affine map $L\colon K \to E$, we have
        \begin{equation*}
            \ex\bigl(L(K)\bigr) \subseteq L\bigl(\ex(K)\bigr).
        \end{equation*}
    \end{lemma}
    
    Applying this lemma to our setting, we obtain
    \begin{align*}
       \ex(\mathcal{I}_{f,g}) &= \ex\bigl(T(\Phi_{f,g})\bigr) \\
       & \subseteq T\bigl(\ex(\Phi_{f,g})\bigr),
    \end{align*}
    where the first equality follows from $T$ being a bijection and the second inclusion from \cref{lem:affine_mapping}. 
    
    For the reverse inclusion, we apply the same argument to the inverse mapping:
    \begin{align*}
       \ex(\Phi_{f,g}) &= \ex\bigl(T^{-1}(\mathcal{I}_{f,g})\bigr) \\
       & \subseteq T^{-1}\bigl(\ex(\mathcal{I}_{f,g})\bigr),
    \end{align*}
    which implies $T\bigl(\ex(\Phi_{f,g})\bigr) \subseteq \ex(\mathcal{I}_{f,g})$. Combining both inclusions yields the desired equality.
\end{proof}

\subsection{Omitted Proofs in \cref{sec:delegation}}

\subsubsection{Proof of \cref{lem:delegation_cfi_lp}}\label{secap:delegation_cfi_lp_proof}

Lemma 1 in \cite{Kleiner2022} shows that in the delegation problem without outside options, an indirect utility $u$ satisfies \labelcref{eqn:IC_delegation} if and only if $u\leq \bar{u}$. By \labelcref{eqn:IR_delegation}, any implementable indirect utility $u$ in the delegation problem with type-dependent outside options has the property $u \geq u_0$. Thus, an indirect utility $u$ is implementable if and only if $u_0 \leq u \leq \bar{u}$. This constitutes a CFI: as in \cite{Kleiner2022}, the upper bound $\bar{u}$ is differentiable by the assumptions on $b$. Moreover, because $u_0$ is tangent to $\bar{u}$ at $\ubar{\theta}$ and $\bar{\theta}$, this fixes the set of possible slopes (redundantly) to be $[\bar{a}(\ubar{\theta}), \bar{a}(\bar{\theta})]$.

The principal's problem consists in maximizing 
\begin{equation*}
    \int_{\ubar{\theta}}^{\bar{\theta}} \Bigl\{\bigl(\theta + \beta(\theta)\bigr) a_{\Gamma}(\theta)+ b_{\Gamma}(\theta)\Bigr\} \, \diff F(\theta),
\end{equation*}
over mechanisms $\Gamma$ that satisfy (IC) and (IR). Note that under any such mechanism, $\bigl(\theta + \beta(\theta)\bigr) a_{\Gamma}(\theta)+ b_{\Gamma}(\theta) = \beta(\theta) a_{\Gamma}(\theta) + u(\theta)$ for all $\theta\in \Theta$. Moreover, \cref{lem:delegation_interval} implies that $a_{\Gamma}(\theta) = u'(\theta)$ almost everywhere on $\Theta$. Therefore,
\begin{equation}\label{eqn:obj_del_ut}
    \int_{\ubar{\theta}}^{\bar{\theta}} \Bigl\{\bigl(\theta + \beta(\theta)\bigr) a_{\Gamma}(\theta)+ b_{\Gamma}(\theta)\Bigr\} \, \diff F(\theta) = \int_{\ubar{\theta}}^{\bar{\theta}} \bigl\{\beta(\theta) u'(\theta) + u(\theta)\bigr\} f(\theta) \, \diff \theta.
\end{equation}

Integrating the right-hand side of \labelcref{eqn:obj_del_ut} by parts shows that the principal's problem is equivalent to \labelcref{eqn:delegation_LP}.

\subsubsection{Proof of \cref{prop:delegation_logconc}}\label{secap:delegation_logconc_proof}

Under \cref{ass:regular_delegation}, the density $\psi_{\mathrm{D}}$ on $(\ubar{\theta}, \bar{\theta})$ is given by $\psi_{\mathrm{D}}(\theta)=f(\theta) - \beta f'(\theta)$. By logconcavity of $f$, there exists $\theta_0\in \Theta$ such that $\psi_{\mathrm{D}}$ is single-crossing from below at $\theta_0$:\footnote{Corollary 3 in \cite{Kleiner2022} makes use of the same fact.}
\begin{equation}\label{eqn:sing_cross_dens}
    \psi_{\mathrm{D}}(\theta) \; \left\{\begin{array}{cc}
        \leq 0  & \text{if $\theta\leq \theta_0$} \\
        \geq 0 & \text{if $\theta\geq \theta_0$}
    \end{array}
    \right.
\end{equation}

If $\theta_0 = \ubar{\theta}$, then \cref{thm:charac_opt} immediately implies that $\bar{u}$ is an optimal solution to \labelcref{eqn:delegation_LP}. The indirect implementation of $\bar{u}$ corresponds to giving the agent full discretion.

If $\theta_0 = \bar{\theta}$, then \cref{thm:charac_opt} immediately implies that $\ubar{u}$ is an optimal solution to \labelcref{eqn:delegation_LP}. This corresponds to letting the agent choose an item from his menu of outside options.

If $\theta_0 \in ( \ubar{\theta}, \bar{\theta})$, define, for each $\theta \in [\ubar{\theta}, \theta_0]$ and $\theta^{\star} \in \Theta$,
\begin{equation*}
    A_{\theta^{\star}}(\theta) = \int_\theta^{\theta^{\star}} (t-\theta^{\star}) \psi_{\mathrm{D}}(t) \, \diff t.
\end{equation*}

The proof will proceed by finding a $\theta^{\star} \in [\theta_0, \bar{\theta}]$ with $\ell(\theta^{\star}) \in [\ubar{\theta}, \theta_0]$ such that $A_{\theta^{\star}}\bigl(\ell(\theta^{\star})\bigr)=0$. Before proving this, we show that this implies the statement of the proposition. In particular, since $\ell(\theta^{\star}) \in [\ubar{\theta}, \theta_0]$, we have $\psi_{\mathrm{D}}(\theta) \leq 0$ for all $\theta \leq \ell(\theta^{\star})$ following \labelcref{eqn:sing_cross_dens}. Likewise, since $\theta^{\star} \in [\theta_0, \bar{\theta}]$, \labelcref{eqn:sing_cross_dens} implies $\psi_{\mathrm{D}}(\theta) \geq 0$ for all $\theta \geq \theta^{\star}$. Furthermore, by Theorem 3.A.44 in \cite{Shaked2007}, $A_{\theta^{\star}}\bigl(\ell(\theta^{\star})\bigr)=0$ implies that $\mu\vert_{[\ell(\theta^{\star}), \theta^{\star}]} \leq_{\mathrm{cx}} \delta_{[\ell(\theta^{\star}), \theta^{\star}]}$. The condition $A_{\theta^{\star}}\bigl(\ell(\theta^{\star})\bigr)=0$ also implies that $\mu\vert_{[\ell(\theta^{\star}), \theta^{\star}]}\geq0$:
\begin{equation*}
    0 = A_{\theta^{\star}}\bigl(\ell(\theta^{\star})\bigr) \geq -(\theta^{\star}-\theta_0) \int_{\ell(\theta^{\star})}^{\theta^{\star}} \psi_{\mathrm{D}}(t) \, \diff t
\end{equation*}
Therefore, \cref{thm:charac_opt} implies that $u_{\theta^{\star}}$ is optimal for \labelcref{eqn:delegation_LP}.

We will now prove the existence of $\theta^{\star} \in [\theta_0, \bar{\theta}]$ with $\ell(\theta^{\star}) \in [\ubar{\theta}, \theta_0]$ such that $A_{\theta^{\star}}\bigl(\ell(\theta^{\star})\bigr)=0$. Assume first that $\bar{\theta}^\star \coloneqq \max \bigl\{\theta \in \Theta \, \vert \, A_\theta(\ubar{\theta})= 0 \bigl\}$ is well defined. We treat the case $A_\theta(\ubar{\theta})> 0$ for all $\theta \in \Theta$ below. If $\bar{\theta}^\star$ is well defined, then, for all $\theta^{\star} \in (\theta_0, \bar{\theta}^\star)$, $A_{\theta^{\star}}(\ubar{\theta}) \geq 0$. It is thus clear that $A_{\theta_0}(\ubar{\theta}) \geq 0$. Furthermore, $A_{\theta^{\star}}(\ubar{\theta})$ is concave in $\theta^{\star}$:
\begin{equation*}
    \frac{\partial}{\partial \theta^{\star}}A_{\theta^{\star}}(\ubar{\theta})=-\int_{\ubar{\theta}}^{\theta^{\star}} \psi_{\mathrm{D}}(t) \, \diff t \qquad \text{and} \qquad \frac{\partial^2}{\partial (\theta^{\star})^2} A_{\theta^{\star}}(\ubar{\theta})=-\psi_{\mathrm{D}}(\theta^{\star}) \leq 0
\end{equation*}

Therefore, $A_{\theta^{\star}}(\ubar{\theta}) \geq 0$ for all $\theta^{\star} \in (\theta_0, \bar{\theta}^\star)$.

From now on, let $\theta^{\star} \in (\theta_0, \bar{\theta}^\star)$. Observe that $A_{\theta^{\star}}$ is continuous and decreasing in $\theta$ with $A_{\theta^{\star}}(\theta_0) \leq 0$. By assumption, $A_{\theta^{\star}}(\ubar{\theta}) \geq 0$. Therefore, by the Intermediate Value Theorem, we can define a correspondence $k(\theta^{\star})$ such that $A_{\theta^{\star}}(\theta)=0$ for all $\theta \in k(\theta^{\star})$ and $\theta^{\star} \in [\theta_0, \bar{\theta}^\star]$. By the concavity of $A_{\theta^{\star}}(\theta)$ in $\theta^{\star}$ for each fixed value of $\theta$, $k$ is decreasing (in the sense that $\min k(\theta) \geq \max k(\theta')$ for $\theta' > \theta$). Furthermore, $k$ is upper hemicontinuous and has image $[k(\bar{\theta}^\star), \theta_0]$.

Recall that for each $\theta^{\star} \in \Theta$, $\ell(\theta^{\star})$ is such that $\ubar{u}$ and $\bar{u}(b) - (b-\theta)\bar{u}'(\theta)$ intersect at $\ell(\theta^{\star})$. By the convexity and differentiability of $\bar{u}$, $\ell$ is continuous and increasing in $\theta^{\star}$. Restricted to $[\theta_0, \bar{\theta}^\star]$, the function $\ell$ therefore has image $[\ell(\theta_0), \ell(\bar{\theta}^\star)]$. Since $k(\bar{\theta}^\star)=\ubar{\theta}$ and $\ell(\bar{\theta}^\star) \geq \ubar{\theta}$, there exists $\theta^{\star}$ such that $\ell(\theta^{\star})\in k(\theta^{\star})$. Therefore, we have found $\theta^{\star} \in [\theta_0, \bar{\theta}]$ with $\ell(\theta^{\star}) \in [\ubar{\theta}, \theta_0]$ such that $A_{\theta^{\star}}\bigl(\ell(\theta^{\star})\bigr)=0$. By the above, this implies optimality of $u_{\theta^{\star}}$.

Now for the case that $A_\theta(\ubar{\theta})> 0$. Let $\ubar{\theta}^{\star}$ be such that $ A_{\bar{\theta}}(\ubar{\theta}^{\star})=0$. This is well defined because $A_{\bar{\theta}}(\theta)$ is decreasing and continuous in $\theta$, $A_{\bar{\theta}}(\ubar{\theta})>0$ by assumption, and $A_{\bar{\theta}}(\theta_0)<0$. The proof now works analogously as above, replacing $\ubar{\theta}$ by $\ubar{\theta}^{\star}$ and observing that $\ubar{\theta}^{\star} = k(\bar{\theta}) \leq \ell(\bar{\theta}) = \bar{\theta}$. This again ensures the existence of $\theta^{\star}$ such that $\ell(\theta^{\star})\in k(\theta^{\star})$, completing the proof. \qed

\subsubsection{Proof of \cref{cor:delegation_compstat}}\label{secap:delegation_compstat_proof}

Denote by $\ell_1$, $\ell_2$ the version of the function $\ell$ as defined in the text corresponding to $\ubar{u}_1$, $\ubar{u}_2$ respectively. Since $\ubar{u}_1 \geq \ubar{u}_2$, $\ell_1(\theta) \geq \ell_2(\theta)$ for all $\theta \in \Theta$. The correspondence $k$ as defined in the proof of \cref{prop:delegation_logconc} in \cref{secap:delegation_logconc_proof} is independent of the lower level boundary function. This means that $\theta_{1}^{\star} \leq \theta_{2}^{\star}$ where $\theta_{i}^{\star}$ is such that $\ell(\theta_{i}^{\star}) \in k(\theta_{i}^{\star})$ for $i \in \{1,2\}$. \qed

\subsection{Omitted Proofs in \cref{sec:contest}}

\subsubsection{Proof of \cref{prop:imp_quantil_assign}}\label{secap:imp_quantil_assign_proof}

Observe that for a feasible assignment $\Gamma$ with expected quantile assignment $\chi$ there exists transfers $t$ such that $\Gamma$ is implementable (i.e., \labelcref{eqn:IC_contest,eqn:IR_contest} hold) if and only if $\chi$ is increasing \citep{Rochet1987}. Therefore, our proof is only concerned with \labelcref{eqn:feasibility_contest_capacity,eqn:feasibility_contest_mean}.

For completeness and to introduce notation, we include the proof for the free disposal case (i.e., $m=0$).

\begin{lemma}\label{lem:contestdesign_freedisposal}
    Under free disposal ($m=0$), an expected quantile assignment $\chi\colon[0,1]\to [0,1]$ is implementable if and only if it is non-decreasing and $G^{-1}\succsim_{\mathrm{w}} \chi $.
\end{lemma}
\begin{proof}
\framebox{$\Leftarrow$} Let $I_{\chi} \in \ex(\mathcal{I}^\mathrm{w}_{0,G^{-1}})$ where $\mathcal{I}^\mathrm{w}_{0,G^{-1}}$ is defined according to \labelcref{eqn:weak_maj_CFI}. Therefore, $\partial_{-}I_{\chi}=\chi$ is non-decreasing and left-continuous (hence a valid quantile function), and is an extreme point of the set of mean-decreasing spreads of $G^{-1}$ following \cref{lem:bij_ext_pt} (see \cref{rmk:isomorphism}). Therefore, by \cite{Kleiner2021} Corollary 2, there exists an extreme point $\tilde{\chi}$ of the set of mean-preserving spreads of $G^{-1}$ and $c\in [0,1]$ such that $\chi(q)=\tilde{\chi}(q)\Ind_{q \geq c}$ for all $q\in [0,1]$. By Proposition 4 in \cite{Kleiner2021}, there exists an assignment $\tilde{\Gamma}$ such that $\chi(q)=x_{\tilde{\Gamma}}\bigl(F^{-1}(q)\bigr)$. Hence, $\int_{\Theta} \tilde{\Gamma}(x \, \vert \,  \theta) \, \diff F(\theta)=G(x)$ for all $x \in [0,1]$ and $\int_{\Theta} \int_{0}^{1} x \, \diff \tilde{\Gamma} (x \, \vert \, \theta) \, \diff F (\theta) =\int_{0}^{1} x \, \diff G(x)$. Now define
\begin{equation*}
    \Gamma (x \, \vert \,  \theta) \coloneqq \left\{\begin{array}{cl}
    1 & \text{if $\theta < F^{-1}(c)$} \\
    \tilde{\Gamma}(x \, \vert \,  \theta) & \text{if $\theta \geq F^{-1}(c)$.}
    \end{array}
    \right.
\end{equation*}

Then we have $\int_{\ubar{\theta}}^{\bar{\theta}} \Gamma(x \, \vert \,  \theta) \, \diff F(\theta) \geq \int_{\ubar{\theta}}^{\bar{\theta}} \tilde{\Gamma}(x \, \vert \,  \theta) =G(x)$ for all $x \in [0,1]$, so \labelcref{eqn:feasibility_contest_capacity} holds.
Furthermore, observe that $\Gamma$ induces $\chi$:
\begin{align*}
    x_\Gamma\bigl(F^{-1}(q)\bigr) &= \left\{\begin{array}{cl}
    0 & \text{if $q < c$} \\[5pt]
    \displaystyle\int_{0}^{1} x \, \diff \tilde{\Gamma}\bigl(x \, \vert \, F^{-1}(q)\bigr) & \text{if $q \geq c$}
    \end{array}
    \right. \\
    &= \left\{\begin{array}{cl}
    0 & \text{if $q < c$} \\
    \tilde{\chi}(q) & \text{if $q \geq c$}
    \end{array}
    \right. \\
    &= \chi(q)
\end{align*}
where the second equality follows from $\tilde{\Gamma}$ inducing $\tilde{\chi}$. Since \labelcref{eqn:feasibility_contest_mean} holds trivially, this shows that $\chi$ is implementable.

Thus, all extreme points of $\mathcal{I}_{0,G^{-1}}^\mathrm{w}$ are implementable. Applying \cref{prop:representation}, we can therefore conclude that all $\chi \prec_{\mathrm{w}} G^{-1}$ are implementable.

\framebox{$\Rightarrow$} Let the expected quantile assignment $\chi$ be implementable. Then, by \labelcref{eqn:IC_contest}, it is non-decreasing. It also satisfies the majorization constraint:
\begin{align*}
    \int_q^1 \chi (t) \, \diff t &= \int_q^1 \int_{0}^{1} x \, \diff \Gamma \bigl(x \, \vert \, F^{-1}(t)\bigr) \, \diff t \\
    &\leq \int_q^1 \int_{0}^{1} x \, \diff \Gamma_{\mathrm{PAM}} \bigl(x \, \vert \, F^{-1}(t)\bigr) \, \diff t \\
    &= \int_q^1 G^{-1}(t) \, \diff t,
\end{align*}
for all $q \in [0,1]$, where $\Gamma_{\mathrm{PAM}}$ denotes the positive assortative assignment. The inequality is a consequence of the fact that no assignment induces a higher average quality among types with quantile $q$ or higher than the positive assortative assignment.
\end{proof}

Let now $\chi$ be an expected quantile assignment such that $G^{-1}\succsim_{\mathrm{w}} \chi \succsim_{\mathrm{w}} H_{m}^{-1}$. By \cref{lem:contestdesign_freedisposal}, there exists an assignment $\Gamma$ that induces it and that satisfies \labelcref{eqn:feasibility_contest_capacity}. Therefore, it is left to show that \labelcref{eqn:feasibility_contest_mean} holds. This is a direct consequence of the weak majorization constraint:
\begin{equation}\label{eqn:mean_constraint_sat_contest}
    \int_{\Theta} \int_{0}^{1} x \, \diff \Gamma (x \, \vert \, \theta) \, \diff F(\theta) = \int_{0}^{1} \int_{0}^{1} x \, \diff \Gamma \bigl(x \, \vert \, F^{-1}(t)\bigr) \, \diff t = \int_{0}^{1} \chi(t) \, \diff t = - I_{\chi}(0) \geq m.
\end{equation}

For the other direction, suppose that $\chi$ is an implementable expected quantile assignment. By the same argument as in \cref{lem:contestdesign_freedisposal}, $\chi \prec_{\mathrm{w}} G^{-1}$ and $\chi$ non-decreasing. The same reasoning that lies behind inequality \labelcref{eqn:mean_constraint_sat_contest} also shows that $I_{\chi}(0)\leq - m$. Since $I_{H_{m}^{-1}}$ is affine, this means that $I_{\chi}$ is a non-decreasing convex function bounded by $I_{G^{-1}}$ and $I_{H_{m}^{-1}}$. Therefore, $G^{-1}\succsim_{\mathrm{w}} \chi \succsim_{\mathrm{w}} H_{m}^{-1}$ by \cref{lem:two-sided-weak-maj}. This completes the proof.

\subsubsection{Proof of \cref{lem:contest_cfi_lp}}\label{secap:contest_cfi_lp_proof}

Consider the following problem:
\begin{equation}\label{eqn:contest_pb_welf}
    \max_{(\Gamma,t)} \; \int_{\Theta} \pi(\theta)\bigl(\theta x_{\Gamma}(\theta) - t(\theta)\bigr) \, \diff F(\theta)+ \alpha\int_{\Theta} t(\theta) \, \diff F(\theta),
\end{equation}
subject to $(\Gamma,t)$ satisfying \labelcref{eqn:feasibility_contest_capacity}, \labelcref{eqn:feasibility_contest_mean}, \labelcref{eqn:IC_contest,eqn:IR_contest}, where $\pi(\theta)$ is the designer's Pareto welfare weight on agents of type $\theta$, and $\alpha>0$ is the weight on the designer's revenue/agents' aggregate effort.

As shown in \cite{Akbarpour2024a} Claim~3, standard envelope and integration-by-parts arguments imply that the designer's problem \labelcref{eqn:contest_pb_welf} can be rewritten as  
\begin{equation}\label{eqn:contest_pb_virtual}
    \max_{\Gamma} \; \int_{\Theta} v(\theta) \, x_{\Gamma}(\theta) \, \diff F(\theta),
\end{equation}
where
\begin{equation*}
    v(\theta) \coloneqq \frac{1}{f(\theta)}\left[\int_{\theta}^{1} \pi(t) \, \diff F(t) + \alpha\Bigl\{\theta f(\theta) - \bigl(1-F(\theta)\bigr)\Bigr\}\right],
\end{equation*}
is the designer's virtual value from allocating a good of expected quality to an agent of type $\theta \in \Theta$.\footnote{Since we normalize the lowest type to zero and restrict transfers to be weakly positive, the lowest type's utility in any mechanism $(\Gamma,t)$ satisfying \labelcref{eqn:IC_contest,eqn:IR_contest} must be zero. This is why it does not appear in \labelcref{eqn:contest_pb_virtual}.} Note that if $\pi(\theta)=0$ for all $\theta\in \Theta$ and $\alpha=1$, we have $v(\theta)=\theta-\frac{1-F(\theta)}{f(\theta)}$ for all $\theta\in \Theta$.

Letting $V(q)\coloneqq v\bigl(F^{-1}(q)\bigr)$ for all $q\in [0,1]$ and making the change of variable $\theta=F^{-1}(q)$ shows that problem \labelcref{eqn:contest_pb_virtual} is equivalent to
\begin{equation}\label{eqn:contest_pb_virtual_quantile}
    \max_{\Gamma} \; \int_{0}^{1} V(q) \, \chi_{\Gamma}(q) \, \diff q,
\end{equation}

Let $I_{\Gamma}(q) = \int_{0}^{q} \chi_{\Gamma}(s) \, \diff s - \int_{0}^{1} \chi_{\Gamma}(s) \, \diff s$ for each $q\in [0,1]$. Hence, since $I_{\Gamma}$ is continuous and $V$ is differentiable, we obtain
\begin{align*}
    \int_{0}^{1} V(q) \, \chi_{\Gamma}(q) \, \diff q &= \int_{0}^{1} V(q) \, \diff I_{\Gamma}(q) \\[5pt]
    &= V(1)I_{\Gamma}(1)  - V(0) I_{\Gamma}(0) - \int_{0}^{1} V'(q) I_{\Gamma}(q) \, \diff q\\[5pt]
    &= \int_{[0,1]} I_{\Gamma} \, \diff \mu_{\mathrm{C}},
\end{align*}
where the second equality follows from (Lebesgue–Stieltjes) integration by parts. 

By \cref{prop:imp_quantil_assign}, an expected quantile assignment $\chi$ is implementable if and only if $G^{-1}\succsim_{\mathrm{w}} \chi \succsim_{\mathrm{w}} H_{m}^{-1}$. By \cref{lem:two-sided-weak-maj}, we can therefore write the designer's problem as \labelcref{eqn:contest_LP}.
\qed

\subsubsection{Proof of \cref{prop:contest_optimality_myerson}}\label{secap:contest_optimality_myerson_proof}

Since the principal is maximizing aggregate effort, $v(\theta)=\theta-\frac{1-F(\theta)}{f(\theta)}$ for all $\theta \in \Theta$. Therefore, we get the following formula for $\psi_\mathrm{C}$:
\begin{equation*}
    \psi_{\mathrm{C}}(q)=\left\{\begin{array}{ll}
        v(0) & \text{if $q=0$} \\[10pt]
        \displaystyle-\frac{v'\bigl(F^{-1}(q)\bigr)}{f\bigl(F^{-1}(q)\bigr)} & \text{if $q \in (0,1)$} \\[10pt]
        v(1) & \text{if $q=1$}
    \end{array}\right.
    \quad=\left\{\begin{array}{ll}
        \displaystyle-\frac{1}{f(0)} & \text{if $q=0$} \\[10pt]
        \displaystyle-\frac{v'\bigl(F^{-1}(q)\bigr)}{f\bigl(F^{-1}(q)\bigr)} & \text{if $q \in (0,1)$} \\[10pt]
        1 & \text{if $q=1$}
    \end{array}\right.
\end{equation*}

Therefore, $\mu_\mathrm{C}$ has positive atoms at $0$ and $1$ and otherwise a non-positive density, because $v'$ is weakly positive by Myerson-regularity.

Let
\begin{equation*}
    q^\star_0\coloneq \min \bigl\{q \in [0,1] \, \bigm\vert \, v\bigl(F^{-1}(q)\bigr)=0\bigr\},
\end{equation*}
which exists by Myerson-regularity. 

Then, $\chi^\star_0(q)\coloneqq G^{-1}(q) \Ind_{q\geq q^\star_0}$ is an effort-maximizing expected assignment. To see this, observe that the partition $\mathcal{Y}$ from \cref{thm:charac_opt} corresponding to the extreme point $I_{\chi^\star_0}$ consists of $\mathcal{Y}_4=\bigl\{[0, q^\star_0]\bigl\}$, $\mathcal{Y}_2=\bigl\{ \{q\} \, \vert \, q \in (q^\star_0,1) \bigl\}$ and $\mathcal{Y}_0 \coloneqq \bigl\{ \{1\}\bigl\}$. For $q \in \cup\mathcal{Y}_2$, we have $\psi_{\mathrm{C}}(q) \leq 0$ as noted above. Additionally, $\mu\vert^+_{[0,q^\star_0]} \leq_{\mathrm{icx}} \mu\vert^-_{[0,q^\star_0]}$ because $\mu\vert^+_{[0,q^\star_0]}$ is a point mass at $0$ and
\begin{align*}
    \mu\bigl([0, q^\star_0]\bigr)&=-v(0)+\int_0^{q^\star_0}\frac{-v'\bigl(F^{-1}(q)\bigr)}{f\bigl(F^{-1}(q)\bigr)} \, \diff q \\[5pt]
    &=-v(0)-v(F^{-1}(q^\star_0))+v(0) \\[5pt]
    &=0.
\end{align*}

Therefore, by \cref{thm:charac_opt} and \cref{lem:affine_mapping}, $\chi^\star_0$ is an effort-maximizing assignment.

Now, for $m >0$, observe that if $I_{G^{-1}}(q^\star_0) \leq -m$, the same argument as for $m=0$ is applicable to show that $\chi^\star_0$ remains optimal.

Suppose for the rest of the proof that $I_{G^{-1}}(q^\star_0) > -m$. Then \labelcref{eqn:feasibility_contest_mean} is binding. Let
\begin{equation*}
    q^\star_m \coloneqq \min \bigl\{q \in [0,1] \, \bigm\vert \, I_{G^{-1}}(q)=-m \bigr\},
\end{equation*}
which exists, since $I_{G^{-1}}$ is continuous and its range includes $m$.

Note that $I_{G^{-1}}(q^\star_0) > -m$ and $I_{G^{-1}}(q^\star_m) = -m$ imply $v(q^\star_m) \leq 0$ by Myerson-regularity. Therefore, there exists $\kappa \in [0,1]$ such that
\begin{equation*}
    \kappa \int_0^1 \psi_{\mathrm{C}}(q) \, \diff \delta_0= \int_0^{q^\star_m} \psi_{\mathrm{C}}(q) \, \diff q \iff \kappa v(0)=-v\bigl(F^{-1}(q^\star_m)\bigr)+v(0).
\end{equation*}

Let $\mathcal{Y}$ be the partition of $[0,1]$ consisting of the three collections $\mathcal{Y}_0 \coloneqq \bigl\{ \{1\}\bigl\}$, $\mathcal{Y}_2\coloneqq \bigl\{ \{q\} \, \vert \, q \in (q^\star_m,1) \bigl\}$ and $\mathcal{Y}_4=\bigl\{[0, q^\star_m]\bigl\}$. We apply a similar technique as in the proof of \cref{thm:charac_opt} to prove optimality of $\chi^\star_m (q) \coloneqq G^{-1}(q) \Ind_{q\geq q^\star_m}$. We define multipliers $\gamma_1, \gamma_2, \gamma_3, \gamma_4 \in \mathcal{M}_+(X)$ and show optimality of $I_{\chi^\star_m}$ for \labelcref{eqn:contest_LP} by showing that
\begin{enumerate}
    \item $\gamma_1, \gamma_2, \gamma_3, \gamma_4 \in \mathcal{M}_+(X)$ are feasible for the dual problem, i.e., $\gamma_{1}-\gamma_{2} - \gamma_{3} + \delta^{\gamma_{3}}_1 - \gamma_{4} + \delta^{\gamma_{4}}_0 \geq_{\mathrm{cx}} \mu_\mathrm{C}$, and;
    \item duality is achieved, i.e.,
    \begin{equation*}
        \int_{[0,1]} I_{\chi^\star_m} \, \diff \mu_\mathrm{C}  =   \int_{[0,1]} I_{H_m^{-1}} \, \diff \gamma_{1} - \int_{[0,1]} I_{G^{-1}} \, \diff \gamma_{2}+ \bar{s} \int_{[0,1]} (1-\Id) \, \diff \gamma_{3}- \ubar{s}\int_{[0,1]} \Id \, \diff \gamma_{4},
    \end{equation*}
    with $\ubar{s}=0$ and $\bar{s}=1$.
\end{enumerate}

First, since $\mathcal{Y}_3$ is empty, let $\gamma_3 \coloneqq \mathbf{0}$. Let $\gamma_1=-(1-\kappa)v(0)\delta_0+v(1)\delta_1$. Since $v(0) <0$ and $v(1)>0$, $\gamma_1 \in \mathcal{M}_+(X)$. 

Let
\begin{equation*}
    \gamma_2(A) \coloneqq -\int_{A\cap [q^\star_m,1)} \psi_{\mathrm{C}}(q) \, \diff q,
\end{equation*}
for any $A\in\mathcal{B}\bigl([0,1]\bigr)$. Since $\psi_{\mathrm{C}}(q)\leq 0 $ for all $q\in(0,1)$ by Myerson-regularity, we have $\gamma_2\in \mathcal{M}_+(X)$.

Furthermore, we have
\begin{equation*}
    \mu_\mathrm{C}\bigl(\, \cdot \, \cap [0, q^\star_m]\bigr) - \gamma_1\bigl(\, \cdot \, \cap [0, q^\star_m]\bigr)\leq_{\mathrm{icx}} \mathbf{0},
\end{equation*}
by the definition of $\kappa$ and its positive part being a point mass at $0$. We can therefore use the same approach as in the proof of \cref{thm:charac_opt} to construct $\gamma_4 \in \mathcal{M}_+(X)$ such that $-\gamma_4+\delta_0^{\gamma_4} \geq_{\mathrm{cx}} \mu_\mathrm{C}(\, \cdot \, \cap [0, q^\star_m])$ and
\begin{equation}\label{eqn:contest_gamma4}
    \int_{[0, q^\star_m]} I_{\chi^\star_m} \, \diff [\mu-\gamma_1] = \ubar{s} \int_{[0,1]} \Id \, \diff \gamma_4=0.
\end{equation}

Overall, we obtain that
\begin{align*}
    \int_{[0,1]} I_{\chi^\star_m} \, \diff \mu_\mathrm{C} &= \int_{[0,1]} I_{\chi^\star_m} \, \diff \gamma_1 + \int_{[0,1]} I_{\chi^\star_m} \, \diff [\mu_\mathrm{C}-\gamma_1] \\
    &= \int_{[0,1]} I_{\chi^\star_m} \, \diff \gamma_1 + \int_0^{q^\star_m} I_{\chi^\star_m} \, \diff [\mu_\mathrm{C}-\gamma_1] + \int_{q^\star}^1 I_{\chi^\star_m} \, \diff \mu_\mathrm{C} \\
    &= \int_{[0,1]} I_{H_m^{-1}} \, \diff \gamma_{1} - \int_{[0,1]} I_{G^{-1}} \, \diff \gamma_{2},
\end{align*}
where the last equality follows from $I_{H_m^{-1}}=I_{\chi^\star_m}$ on $\supp (\gamma_1)=\{0,1\}$, \labelcref{eqn:contest_gamma4} and 
$\int_{q^\star}^1 I_{\chi^\star_m} \, \diff \mu_\mathrm{C}=-\int_{[0,1]} I_{G^{-1}} \, \diff \gamma_{2}$ by the definition of $\gamma_2$ and $I_{\chi^\star_m}(1)=0$.

It is left to show that $\gamma_i$ satisfy the convex order condition to ensure feasibility for the dual. This is satisfied by the convex order condition for each of the multipliers:
\begin{align*}
    \int_{[0,1]} u \, \diff \mu_\mathrm{C} &= \int_{[0,1]} u \, \diff \gamma_1 + \int_0^{q^\star_m} u \, \diff [\mu_{\mathrm{C}}-\gamma_1] + \int_{q^\star_m}^1 u(q) \psi_\mathrm{C}(q) \, \diff q \\
    &\leq \int_{[0,1]} u \, \diff \gamma_1 + \int_0^{q^\star_m} u \, \diff [-\gamma_4+\delta_0^{\gamma_4}] + \int_{q^\star_m}^1 u(q) \psi_\mathrm{C}(q) \, \diff q
\end{align*}
for each $u \in \mathcal{K}$, where the inequality follows from the construction of $\gamma_4$.
\qed

\subsection{Omitted Proofs in \cref{sec:persuasion}}

\subsubsection{Proof of \cref{lem:baypers_infoconstraints}}\label{secap:baypers_infoconstraints_proof}

Let $\ubar{G} \succsim \bar{G} \succsim F$. It is a direct implication of \cite{Strassen1965} that the sender can induce any distribution of posterior means $G$ that satisfies $\ubar{G} \succsim G \succsim \bar{G}$. The sender's constrained persuasion problem can therefore be written with a two-sided majorization constraint: 
\begin{equation}\label{eqn:cons_pers_pb}
    \max_{\ubar{G} \succsim G \succsim \bar{G}} \; \int_{0}^{1} v(x) \, \mathrm{d}G(x)
\end{equation}

Since $G$ is a cumulative distribution function, it is of bounded variation. This allows us to apply (Lebesgue–Stieltjes) integration by parts twice to the objective function:
\begin{align*}
    \int_{0}^{1} v(x) \, \mathrm{d}G(x) &= v(1) \overbrace{G(1^{+})}^{=1} - v(0) \overbrace{G(0^{-})}^{=0} - \int_{0}^{1} G(x) v'(x) \, \mathrm{d}x \\[5pt]
    &= v(1) - v'(1^{+}) \underbrace{I_{G}(1)}_{=0} + \int_{0}^{1} I_{G}(x) \, \mathrm{d}v'(x) \\[5pt]
    &= v(1) + \int_{[0,1]} I_{G} \, \mathrm{d}\mu_{v},
\end{align*}
where $I_{G}(x) \coloneqq \int_{0}^{x} G(s) \, \mathrm{d}s - m_G$ for all $x \in [0,1]$.

Since $v(1)$ is constant, the previous computation, together with \cref{lem:two-sided-maj}, shows that problem~\labelcref{eqn:cons_pers_pb} is equivalent to \labelcref{eqn:cons_pers_pb_CFI}.

\subsubsection{Proof of \cref{lem:infodesign_sshaped}}\label{secap:infodesign_sshaped_proof}

The proof follows from the same arguments as the proofs \cref{secap:delegation_logconc_proof,secap:delegation_compstat_proof}. To see this, observe that $\mu_v$ has density $v''$ because $v$ is assumed to be smoothly S-shaped. The density $v''$ is first positive, then negative. The problem now is the same as \labelcref{eqn:delegation_LP} with flipped signs, allowing us to draw the same conclusions.

\clearpage

\begin{center}
    \huge \textbf{Additional Technical Appendix}
\end{center}

\section{Construction of the Dual Problem}\label{secap:dual_problem}

The goal of this section is to formulate the primal problem \labelcref{eqn:Primal} as a conic linear problem and to derive its dual using \cite{Shapiro2001}. We start by introducing some notation and basic properties before proceeding to reformulate the primal problem as a conic linear problem.

The primal problem is set on the space $\mathcal{C}$ paired with its dual space $\mathcal{M}(X)$ \citep[Riesz-Markov representation theorem,][Theorems 19.54 and 19.55]{Yeh2014} and bilinear form $\langle \cdot , \cdot \rangle \colon \mathcal{M}(X) \times \mathcal{C} \to \mathbb{R}$, $\langle \mu, f \rangle \mapsto \int_{X} f \, \diff \mu$. The dual problem is set on $\mathcal{C}^4$ paired with $\mathcal{M}(X)^4$, endowed with the bilinear form $\langle \cdot , \cdot \rangle \colon \mathcal{M}(X)^4 \times \mathcal{C}^4 \to \mathbb{R}$, $\langle \mu, f \rangle \mapsto \sum_{i=1}^4 \int_{X} f_i \, \diff \mu_i$.

Let $\mathcal{P}\coloneqq\{ u \in \mathcal{C} \; \vert \; u\geq 0\}$. Note that both $\mathcal{K}$ and $\mathcal{P}$, as well as $\mathcal{K}^4$ and $\mathcal{P}^4$ are convex cones in $\mathcal{C}$ or $\mathcal{C}^4$ respectively. The polar cone of $\mathcal{P}$ is
\begin{equation*}
    \mathcal{P}^{*} \coloneqq \bigl\{ \mu \in \mathcal{M}(X) \; \vert \; \forall \, f \in \mathcal{P}, \; \langle \mu, f \rangle \geq 0 \bigr\}.
\end{equation*}

By the Riesz-Markov representation theorem \citep[][Theorem 19.55]{Yeh2014}, $\mathcal{P}^{*}$ equals the set of finite positive Radon measures on $X$, denoted as $\mathcal{M}_{+}(X)$. Therefore, the polar cone $(\mathcal{P}^*)^4$ of $\mathcal{P}^4$ equals $\mathcal{M}_{+}(X)^4$. The polar cone of $\mathcal{K}$ is
\begin{equation*}
    \mathcal{K}^* \coloneqq \bigl\{ \mu \in \mathcal{M}(X) \; \vert \; \forall  u \in \mathcal{K}, \; \langle \mu, u \rangle \geq 0\bigr\}.
\end{equation*}

Let the linear mapping $A \colon \mathcal{C} \to \mathcal{C}^4$ and $b \in \mathcal{C}^4$ be defined by
\begin{equation*}
    Af(x) = \begin{pmatrix}
        -f(x) \\ f(x) \\ f(x)-u(1) \\ f(x)-u(0)
    \end{pmatrix}
    \qquad \text{and} \qquad b(x)=\begin{pmatrix}
        \bar{u}(x) \\ -\ubar{u}(x) \\ \bar{s}(1-x) \\ -\ubar{s}x
    \end{pmatrix}.
\end{equation*}

We can now write the primal problem as 
\begin{align*}\label{eqn:Primal_conic}
    \min_{u \in \mathcal{K}} \qquad & \langle - \mu, u \rangle \\
    \text{s.t.} \qquad & Au + b \in \mathcal{P}^4
\end{align*}
which is the required form in \cite{Shapiro2001}. To use the dual formulation stated there, we have to ensure that for each $\gamma \in \mathcal{M}(X)^4$ there is a unique $\eta \in \mathcal{M}(X)$ such that $\langle \gamma, Af\rangle = \langle \eta, f\rangle$ for all $f\in \mathcal{C}$ \citep[(A1) in][]{Shapiro2001}. Let $\gamma \in \mathcal{M}(X)^4$ and $f \in \mathcal{C}$. We perform the following computation:
\begin{align*}
    \langle\gamma, Af\rangle &= \langle \gamma_1, -f\rangle + \langle \gamma_2, f\rangle + \langle \gamma_3, f-f(1) \rangle + \langle \gamma_4, f-f(0)\rangle \\
    &= \langle -\gamma_1 + \gamma_2 + \gamma_3 - \delta_1^{\gamma_3} + \gamma_4 - \delta_0^{\gamma_4}, f\rangle
\end{align*}
By the Riesz-Markov representation theorem and linearity of $\langle \gamma, A \, \cdot \, \rangle$, $\eta \coloneqq -\gamma_1 + \gamma_2 + \gamma_3 - \delta_1^{\gamma_3} + \gamma_4 - \delta_0^{\gamma_4}$ is the unique element of $\mathcal{M}(X)$ that satisfies this. We can therefore define the adjoint mapping $A^* \colon \mathcal{M}(X)^4 \to \mathcal{M}(X)$ by $\gamma \mapsto -\gamma_1 + \gamma_2 + \gamma_3 - \delta_1^{\gamma_3} + \gamma_4 - \delta_0^{\gamma_4}$.

By equation (2.4) in \cite{Shapiro2001}, the dual \labelcref{eqn:Dual} is therefore given by
\begin{align*}
    \max_{\gamma \in -(\mathcal{P}^{*})^4} \qquad & \langle \gamma, b \rangle \\
    \text{s.t.} \qquad & A^* \gamma - \mu \in \mathcal{K}^*.
\end{align*}
Substituting $-\gamma$ for $\gamma$, using $(\mathcal{P}^{*})^4=\mathcal{M}_+(X)^4$ and observing that $-A^* \gamma - \mu \in \mathcal{K}^*$ if and only if $\gamma_{1}-\gamma_{2} - \gamma_{3} + \delta^{\gamma_{3}}_1 - \gamma_{4} + \delta^{\gamma_{4}}_0 \geq_{\mathrm{cx}} \mu$ yields the formulation of \labelcref{eqn:Dual} stated in \cref{secap:dual_weakduality}.

\section{Strong Duality}\label{secap:strongduality}

For completeness, we show that strong duality between \labelcref{eqn:Primal} and \labelcref{eqn:Dual} holds:

\begin{lemma}\label{lem:strong_duality}
    $u^{\star}$ is optimal for \labelcref{eqn:Primal} if and only if there exists measures $(\gamma_{1}, \gamma_{2}, \gamma_{3}, \gamma_{4}) \in \mathcal{M}_+(X)^4$ such that
    \begin{enumerate}[(i)]
        \item $\displaystyle\gamma_{1}-\gamma_{2} - \gamma_{3} + \delta^{\gamma_{3}}_1 - \gamma_{4} + \delta^{\gamma_{4}}_0 \geq_{\mathrm{cx}} \mu$
        \item $\displaystyle\int_{X} u^{\star}(x) \, \diff \mu = \int_{X} \bar{u} \, \diff \gamma_{1} - \int_{X} \ubar{u} \, \diff \gamma_{2}+ \bar{s} \int_{X} (1-\Id) \, \diff \gamma_{3}- \ubar{s}\int_{X} \Id \, \diff \gamma_{4}$
    \end{enumerate}
\end{lemma}
\begin{proof}
Since the feasible set of \labelcref{eqn:Primal} may have empty interior, we cannot apply Slater's condition. Therefore, we parametrize the problem and a apply a result from \cite{Gretsky2002}.\footnote{In doing this, we follow a similar approach as Lemma 6 in \cite{Kleiner2022}.}

Fix a finite signed Radon measure $\mu$. For any convex function interval $\hat{\mathcal{U}}$, let $V(\hat{\mathcal{U}})$ be the value of $\mathrm{P}_{\hat{\mathcal{U}}, \mu}$, the linear problem with feasible set $\hat{\mathcal{U}}$. Theorem 1 and Condition 3 on page 266 in \cite{Gretsky2002} state that to establish strong duality, it is sufficient to show that
\begin{equation}\label{eqn:nonemptysubgradient}
    \frac{V(\hat{\mathcal{U}})- V(\mathcal{U})}{\lVert (\hat{\bar{u}}, \hat{\ubar{u}}, \hat{\bar{s}}, \hat{\ubar{s}})-(\bar{u}, \ubar{u}, \bar{s}, \ubar{s})\rVert}
\end{equation}
is bounded above. Here, we endow the product space $\mathcal{K}\times \mathcal{K}\times \mathbb{R} \times \mathbb{R}$ with the norm given by the sum of the norms on the individual spaces.\footnote{That is, $\lVert (\hat{\ubar{u}}, \hat{\bar{u}}, \hat{\ubar{s}}, \hat{\bar{s}})-(\ubar{u}, \bar{u}, \ubar{s}, \bar{s})\rVert = \Vert\hat{\bar{u}}-\bar{u}\rVert+\lVert\hat{\ubar{u}}-\ubar{u}\rVert + \lVert\hat{\bar{s}}-\bar{s}\rVert + \lVert\hat{\ubar{s}}-\ubar{s}\rVert$.}

To show that \labelcref{eqn:nonemptysubgradient} is bounded above, we will construct some $u\in \mathcal{U}$ for each $\hat{u}\in \hat{\mathcal{U}}$ such that 
\begin{equation*}
    \lVert\hat{u}-u\rVert \leq \lVert(\hat{\bar{u}}, \hat{\ubar{u}}, \hat{\bar{s}}, \hat{\ubar{s}})-(\bar{u}, \ubar{u}, \bar{s}, \ubar{s})\rVert.
\end{equation*}

This will yield the desired result. Since $\mu$ is bounded, we get
\begin{align*}
    \int_X \hat{u} \, \diff \mu - \int_X u \, \diff \mu &\leq \lVert\mu\rVert_{\mathrm{TV}} \lVert\hat{u}-u\rVert \\
    &\leq \lVert\mu\rVert_{\mathrm{TV}} \, \lVert(\hat{\bar{u}}, \hat{\ubar{u}}, \hat{\bar{s}}, \hat{\ubar{s}})-(\bar{u}, \ubar{u}, \bar{s}, \ubar{s})\lVert,
\end{align*}
which shows that \labelcref{eqn:nonemptysubgradient} is bounded above.\footnote{Recall that the total variation norm is defined by $\lVert\mu\rVert_{\mathrm{TV}}=\lvert \mu \rvert(X)$ for any $\mu\in\mathcal{M}(X)$.}

Let $\hat{u}\in \hat{\mathcal{U}}$. To construct $u$, we will define an operator $\Psi_{s}$. For any convex function $f$ on the domain $X$ with $\partial f \cap [\ubar{s}, \bar{s}] \neq \varnothing$, let $x_{\ubar{s}} \coloneqq \inf \{ x \in X \; \vert \; \partial_+ f (x) \geq \ubar{s} \}$ and $x_{\bar{s}} \coloneqq \sup \{ x \in X \; \vert \; \partial_- f (x) \leq \bar{s} \}$.
We define $\Psi_{s}(f)$ as
\begin{equation*}
    \Psi_{s}(f) = \left\{\begin{matrix*}[l] 
    f(x_{\ubar{s}})-(x_{\ubar{s}}-x)\ubar{s} & & x \in [0, x_{\ubar{s}}] \\
    f(x) & & x \in (x_{\ubar{s}}, x_{\bar{s}}) \\
    f(x_{\bar{s}}) + (x-x_{\bar{s}})\bar{s} & & x \in [x_{\bar{s}}, 1]
    \end{matrix*} \right.
\end{equation*}
Intuitively, $\Psi_{s}$ replaces the parts where $f$ violates $\mathcal{U}$'s gradient constraints with affine pieces where the gradient constraint is binding.

We are now ready to construct $u$: let
\begin{equation*}
    u \coloneqq \Psi_{s}\bigl( \vex[\min\{\max\{\hat{u}, \ubar{u}\}, \bar{u} \}] \bigl).
\end{equation*}
Observe that $\ubar{u} \leq \vex[\min\{\max\{\hat{u}, \ubar{u}\}, \bar{u} \}] \leq \bar{u}$. $\Psi_{s}$ ensures that $\partial u \subset[\ubar{s}, \bar{s}]$. Thus, $u \in \mathcal{U}$. It remains to show that $\lVert\hat{u}-u\rVert$ is bounded above by $ \Vert (\hat{\bar{u}}, \hat{\ubar{u}}, \hat{\bar{s}}, \hat{\ubar{s}})-(\bar{u}, \ubar{u}, \bar{s}, \ubar{s})\rVert$. We get the following:
\begin{align*}
    \lVert \hat{u}-u\rVert &= \bigl\Vert \hat{u} - \Psi_{s}\bigl( \vex[\min\{\max\{\hat{u}, \ubar{u}\}, \bar{u} \}] \bigl) \bigr\rVert \\
    &\leq \Vert \hat{u} - \vex[\min\{\max\{\hat{u}, \ubar{u}\}, \bar{u} \}]\rVert + \lVert\hat{\ubar{s}}-\ubar{s} \rVert + \lVert \hat{\bar{s}}-\bar{s} \rVert \\
    &\leq \lVert \hat{u} - \min\{\max\{\hat{u}, \ubar{u}\}, \bar{u} \} \rVert +  \lVert\hat{\ubar{s}}-\ubar{s} \rVert + \lVert \hat{\bar{s}}-\bar{s} \rVert\\
    &= \sup_{x \in \{s \in X \; \vert \; \max\{\hat{u}(s), \ubar{u}(s)\} \geq \bar{u}(s) \}} \vert \hat{u}(x) - \bar{u}(x) \rvert \\
    & \quad + \sup_{x \in \{s \in X \; \vert \; \max\{\hat{u}(s), \ubar{u}(s)\} < \bar{u}(s) \}} \lvert \hat{u}(x) - \max\{\hat{u}(x), \ubar{u}(x)\} \rvert \\
    & \quad +  \Vert\hat{\ubar{s}}-\ubar{s} \rVert + \lVert \hat{\bar{s}}-\bar{s}\rVert \\
    &\leq \sup_{x \in \{s \in X \; \vert \; \hat{u}(s) \geq \bar{u}(s) \}}  \hat{u}(x) - \bar{u}(x) \\
    & \quad + \sup_{x \in \{s \in X \; \vert \; \ubar{u}(s) < \hat{u}(s) < \bar{u}(s) \}} \hat{u}(x) - \hat{u}(x) \\
    & \quad + \sup_{x \in \{s \in X \; \vert \; \ubar{u}(s) > \hat{u}(s) \}} \ubar{u}(x) - \hat{u}(x) \\
    & \quad + \lVert \hat{\ubar{s}}-\ubar{s} \rVert + \lVert \hat{\bar{s}}-\bar{s}\rVert \\
    &\leq \Vert \hat{\bar{u}}-\bar{u} \rVert + \lVert \hat{\ubar{u}}-\ubar{u} \rVert + \lVert\hat{\bar{s}}-\bar{s}\rVert + \Vert\hat{\ubar{s}}-\ubar{s}\rVert \\
    &= \lVert(\hat{\ubar{u}}, \hat{\bar{u}}, \hat{\ubar{s}}, \hat{\bar{s}})-(\ubar{u}, \bar{u}, \ubar{s}, \bar{s})\rVert,
\end{align*}
The first inequality is a consequence of the substitution caused by the operator $\Psi_{s}$ and the fact that we normalized $X=[0,1]$. The second inequality is a consequence of $\lVert\vex (f) - \vex(g)\rVert \leq \lVert f-g\rVert$ for any $f,g \in \mathcal{C}$.

This completes the proof, as we have now constructed, for every $\hat{u} \in \hat{\mathcal{U}}$, a function $u\in \mathcal{U}$ such that $\lVert \hat{u}-u \rVert$ is bounded above by $ \lVert (\hat{\bar{u}}, \hat{\ubar{u}}, \hat{\bar{s}}, \hat{\ubar{s}})-(\bar{u}, \ubar{u}, \bar{s}, \ubar{s})\rVert$.
\end{proof}

\cref{lem:strong_duality} also implies that the measures $\gamma_i$, $i \in \{1, \dots, 4\}$ for an optimal $u^\star$ satisfy the following support conditions:
\begin{align}
    \supp(\gamma_1) & \subseteq \bigl\{ x \in X \; \vert \; u(x) = \bar{u}(x) \bigr\} \label{eqn:support_condition1}\\
    \supp(\gamma_2) &\subseteq \bigl\{ x \in X \; \vert \; u(x) = \ubar{u}(x) \bigr\} \label{eqn:support_condition2}\\
    \supp(\gamma_3) &\subseteq \bigl\{ x \in X \; \vert \; \partial_- u(x) = \bar{s} \bigr\} \label{eqn:support_condition3}\\
    \supp(\gamma_4) &\subseteq \bigl\{ x \in X \; \vert \; \partial_+u(x) = \ubar{s} \bigr\}\label{eqn:support_condition4}
\end{align}

\section{Relation to Concavification}\label{secap:concavification}

To solve linear problems on affinely bounded CFIs (in particular, for $\bar{s}$-affine boundedness, e.g., \cite{Dworczak2024}) the literature has used a concavification approach to determine optimal mechanisms. We now show how our optimality conditions as stated in \cref{thm:charac_opt} are related to the concavification approach for affinely bounded CFIs.\footnote{For a definition of affine boundedness, see \cref{sec:lower_bound}.}

For this section, we assume that the signed measure $\mu$ takes the following form:
\begin{equation*}
    \mu(A) = \int_{A} \psi \, \diff \nu,
\end{equation*}
with $\nu\coloneqq\delta_{1} -\delta_{0} + \lambda$. In other words, we assume that $\mu$ admits a density on $(0,1)$. In most applications in the literature, this is satisfied (as an example, see \cref{sec:screening_run_example}).

\subsection{Concavification for Upper Affine Boundedness}

For an $\bar{s}$-affinely bounded CFI $\mathcal{U}$, any $u\in \mathcal{U}$ satisfies $u(0)=\bar{u}(0)=\ubar{u}(0)$. The following method was used in \cite{Dworczak2024} with a variable $u(0)$, but remains valid for our case.

Let $W(x)=\int_x^1 \psi(s) \, \diff s + \psi(1)$ and define
\begin{equation*}
    \mathcal{W}(x) \coloneqq \int_x^1 W(s) \, \diff s \qquad \text{and} \qquad \longbar{\mathcal{W}} \coloneqq \cav(\mathcal{W}).
\end{equation*}
We further define a cutoff type $\bar{x}^\star \coloneqq \inf \{x \in x \; \vert \; \longbar{\mathcal{W}}'(x) = 0 \}$ if it exists (and $\bar{x}^\star \coloneqq 1$ otherwise) and a collection $\mathcal{X}_{\mathrm{C}}$ of maximal intervals $(a,b)$ within $(0, \bar{x}^\star)$ such that $\mathcal{W}$ lies strictly below $\longbar{\mathcal{W}}$ on $(a,b)$. This gives a partition of $[0,1]$:
\begin{itemize}
	\item $\mathcal{Y}_{0} = \bigl\{ \{0\} \bigr\} $.
	\item $\mathcal{Y}_{1} = \emptyset$.
	\item $\mathcal{Y}_{2} = \Bigl\{ \{ x \}\subset X \; \big\vert \; x < \bar{x}^\star, \; x \notin \bigcup \mathcal{X}_{\mathrm{C}}  \Bigr\}$.
	\item $\mathcal{Y}_{3} = \Bigl\{ [\bar{x}^\star, 1]\Bigr\} $.
	\item $\mathcal{Y}_{4} = \emptyset $.
	\item $\mathcal{Y}_{5} = \Bigl\{ (a,b) \subset X \; \big\vert \; (a,b) \in \mathcal{X}_{\mathrm{C}} \Bigr\}$.
\end{itemize}

This partition is independent of $\ubar{u}$ and it satisfies the optimality conditions of \cref{thm:charac_opt}:

\begin{lemma}\label{lem:concavification_upper}
    The following properties hold:
    \begin{enumerate}
        \item If $\{x\} \in \mathcal{Y}_{2}$, then $\psi(x) \leq 0$. \label{cond:connection_upperconc_y2}
        \item If $Y \in \mathcal{Y}_{3}$, then $\mu\vert_{Y} \leq_{\mathrm{dcx}} \mathbf{0}$. \label{cond:connection_upperconc_y3}
        \item If $Y \in \mathcal{Y}_{5}$, then $\mu\vert_{Y} \leq_{\mathrm{cx}} \mathbf{0}$. \label{cond:connection_upperconc_y5}
    \end{enumerate}
\end{lemma}

Before proving the lemma, we note that it ensures optimality of
\begin{equation*}
    u^\star (x) \coloneqq
    \left\{\begin{array}{ll}
        \ubar{u}(x) & \text{if $\{x\} \in \mathcal{Y}_2$} \\[3pt]
        \ubar{u}(a) + \frac{x-a}{b-a}\bigl(\ubar{u}(b)-\ubar{u}(a)\bigl) & \text{if $x\in (a,b)\in \mathcal{Y}_5$} \\[3pt]
        \ubar{u}(\bar{x}^\star) + (x-\bar{x}^\star)\bar{s} & \text{if $x\geq \bar{x}^\star$}
    \end{array}\right.
\end{equation*}
for \labelcref{eqn:Lin_Pb}. The function $u^\star$ irons $\ubar{u}$ on all the intervals $(a,b) \in \mathcal{X}_C$ and departs affinely from $\ubar{u}$ with the maximum slope at $\bar{x}^\star$. Since the partition $\mathcal{Y}$ is independent of $\ubar{u}$, it provides an algorithm for finding an optimal $u^\star$ for any lower bound $\ubar{u}$ such that the CFI $\mathcal{U}$ is $\bar{s}$-affinely bounded.

\begin{proof}
    First, let $x$ be such that $\{x\} \in \mathcal{Y}_2$. Then $\mathcal{W}(x)=\longbar{\mathcal{W}}(x)$, meaning that $\mathcal{W}$ is concave in a neighborhood around $x$. This means that $\psi(x)=\mathcal{W}''(x) \leq 0$.

    Now, let $Y \in \mathcal{Y}_5$, i.e., $Y=(a,b) \in \mathcal{X}_{\mathrm{C}}$. Restricted to $(a,b)$, the measure $\mu$ is absolutely continuous with respect to the Lebesgue measure with density $\psi$. We thus have to show that
\begin{itemize}
    \item on $(a,b)$, the positive and negative parts of $\mu$ have equal mass, i.e.\begin{equation}\label{eqn:connection_Dworczak2024_r5_equalmass}
        \int_a^b \psi(s) \, \diff s =0
    \end{equation}
    \item the "CDF"s of the positive and negative part satisfy the convex order condition, i.e.
    \begin{equation}\label{eqn:connection_Dworczak2024_r5_cx}
        \int_a^x \int_a^t \psi(s) \, \diff s \, \diff t \leq 0.
    \end{equation}
    for all $x \in (a,b)$ with equality for $x=b$ \citep[Theorem 3.A.1]{Shaked2007}.
\end{itemize}
First, we show \labelcref{eqn:connection_Dworczak2024_r5_equalmass}. We have that
\begin{align*}
    \int_a^b \psi(s) \, \diff s = \int_a^1 \psi(s) \, \diff s - \int_b^1 \psi(s) \, \diff s = W(a)-W(b) =- \longbar{\mathcal{W}}'(a) + \longbar{\mathcal{W}}'(b)=0.
\end{align*}
Here, the third equality follows from $\mathcal{W}$ being continuously differentiable with derivative $\mathcal{W}'(x)$ equal to $\longbar{\mathcal{W}}'(x)$ whenever $\mathcal{W}(x)=\longbar{\mathcal{W}}(x)$. The last equality is implied by $\longbar{\mathcal{W}}$ being affine on $(a,b)$.

\noindent Now, we prove \labelcref{eqn:connection_Dworczak2024_r5_cx}. Using the properties of $\mathcal{W}$ and $\longbar{\mathcal{W}}$, we get the following for $x \in (a,b)$:
\begin{align*}
    \int_a^x \int_a^t \psi(s) \, \diff s \, \diff t &= \int_a^x W(a)-W(t) \, \diff t \\
    &= (x - a) W(a) - (\mathcal{W}(a) - \mathcal{W}(x)) \\
    &= - (x - a) \longbar{\mathcal{W}}'(a)- \longbar{\mathcal{W}}(a) + \mathcal{W}(x) \\
    &= - \longbar{\mathcal{W}}(x) + \mathcal{W}(x) \\
    &\leq 0
\end{align*}
For $x=b$, equality holds because $\longbar{\mathcal{W}}(b)=\mathcal{W}(b)$ by definition. We have thus shown $\mu\vert_{Y} \leq_{\mathrm{cx}} \mathbf{0}$.

Lastly, consider $Y=[\bar{x}^\star, 1] \in \mathcal{Y}_3$. As in the previous case, we have to show that
\begin{itemize}
    \item on $[\bar{x}^\star, 1]$, the positive and negative parts of $\mu$ have equal mass, i.e.\begin{equation}\label{eqn:connection_Dworczak2024_r3_equalmass}
        \int_{\bar{x}^\star}^1 \psi(s) \, \diff s + \psi(1)=0
    \end{equation}
    \item the "CDF"s of the positive and negative part satisfy the decreasing convex order condition, i.e.
    \begin{equation}\label{eqn:connection_Dworczak2024_r3_cx}
        \int_{\bar{x}^\star}^x \int_{\bar{x}^\star}^t \psi(s) \, \diff s \, \diff t \leq 0.
    \end{equation}
    for all $x \in (\bar{x}^\star,1)$ \citep[Theorem 4.A.2]{Shaked2007}.\footnote{The atom at $1$ does not change the value of the integral.} 
\end{itemize}
\cref{eqn:connection_Dworczak2024_r3_equalmass} follows from the definition of $\bar{x}^\star$: at $\bar{x}^\star$, it holds that
\begin{equation*}
    0=\longbar{\mathcal{W}}'(\bar{x}^\star)=\mathcal{W}'(\bar{x}^\star)=-W(\bar{x}^\star)=\int_{\bar{x}^\star}^1 \psi(s) \, \diff s + \psi(1),
\end{equation*}
establishing \labelcref{eqn:connection_Dworczak2024_r3_equalmass}. \cref{eqn:connection_Dworczak2024_r3_cx} follows from the same arguments as \labelcref{eqn:connection_Dworczak2024_r5_cx}. Thus, $\mu\vert_{[\bar{x}^\star, 1]} \leq_{\mathrm{dcx}} \mathbf{0}$.
\end{proof}

\subsection{Concavification for Lower Affine Boundedness}

We now turn to the analogous result for $\ubar{s}$-affinely bounded CFIs. The approach is similar. However, since we now have $u(1)=\bar{u}(1)=\ubar{u}(1)$, we have to adjust the definitions of $W$ and $\mathcal{W}$.

Let $W(x)=-\psi(0) + \int_0^x \psi(s) \, \diff s$ and define
\begin{equation*}
    \mathcal{W}(x) \coloneqq \int_0^x W(s) \, \diff s \qquad \text{and} \qquad \longbar{\mathcal{W}} \coloneqq \cav(\mathcal{W}).
\end{equation*}
We further define a cutoff type $\ubar{x}^\star \coloneqq \inf \{x \in x \; \vert \; \longbar{\mathcal{W}}'(x) = 0 \}$ if it exists (and $\bar{x}^\star \coloneqq 0$ otherwise) and a collection $\mathcal{X}_{\mathrm{C}}$ of maximal intervals $(a,b)$ within $( \bar{x}^\star,1)$ such that $\mathcal{W}$ lies strictly below $\longbar{\mathcal{W}}$ on $(a,b)$. This gives a partition of $[0,1]$:
\begin{itemize}
	\item $\mathcal{Y}_{0} = \bigl\{ \{0\} \bigr\} $.
	\item $\mathcal{Y}_{1} = \emptyset$.
	\item $\mathcal{Y}_{2} = \Bigl\{ \{ x \}\subset X \; \big\vert \; x > \bar{x}^\star, \; x \notin \bigcup \mathcal{X}_{\mathrm{C}}  \Bigr\}$.
	\item $\mathcal{Y}_{3} = \emptyset $.
	\item $\mathcal{Y}_{4} = \Bigl\{ [0,\ubar{x}^\star]\Bigr\} $.
	\item $\mathcal{Y}_{5} = \Bigl\{ (a,b) \subset X \; \big\vert \; (a,b) \in \mathcal{X}_{\mathrm{C}} \Bigr\}$.
\end{itemize}

This partition is independent of $\ubar{u}$ and it satisfies the optimality conditions of \cref{thm:charac_opt}:

\begin{lemma}\label{lem:concavification_upper}
    The following properties hold:
    \begin{enumerate}
        \item If $\{x\} \in \mathcal{Y}_{2}$, then $\psi(x) \leq 0$. \label{cond:connection_upperconc_y2}
        \item If $Y \in \mathcal{Y}_{4}$, then $\mu\vert_{Y} \leq_{\mathrm{icx}} \mathbf{0}$. \label{cond:connection_upperconc_y3}
        \item If $Y \in \mathcal{Y}_{5}$, then $\mu\vert_{Y} \leq_{\mathrm{cx}} \mathbf{0}$. \label{cond:connection_upperconc_y5}
    \end{enumerate}
\end{lemma}

Before proving the lemma, we note that it ensures optimality of
\begin{equation*}
    u^\star (x) \coloneqq
    \left\{\begin{array}{ll}
        \ubar{u}(\ubar{x}^\star) + (x-\ubar{x}^\star)\ubar{s} & \text{if $x\leq \ubar{x}^\star$}  \\[3pt]
        \ubar{u}(x) & \text{if $\{x\} \in \mathcal{Y}_2$} \\[3pt]
        \ubar{u}(a) + \frac{x-a}{b-a}\bigl(\ubar{u}(b)-\ubar{u}(a)\bigl) & \text{if $x\in (a,b)\in \mathcal{Y}_5$}
    \end{array}\right.
\end{equation*}
for \labelcref{eqn:Lin_Pb}. The function $u^\star$ irons $\ubar{u}$ on all the intervals $(a,b) \in \mathcal{X}_C$ and departs affinely from $\ubar{u}$ with the maximum slope at $\bar{x}^\star$. Since the partition $\mathcal{Y}$ is independent of $\ubar{u}$, it provides an algorithm for finding an optimal $u^\star$ for any lower bound $\ubar{u}$ such that the CFI $\mathcal{U}$ is $\ubar{s}$-affinely bounded.

The proof follows the same steps as for the $\bar{s}$-affinely bounded case and is therefore omitted.

\end{document}